\documentclass[
 reprint,]{revtex4-2}

\usepackage{comment}
\usepackage{graphicx}
\usepackage{dcolumn}
\usepackage{bm}
\usepackage{xcolor}
\usepackage{cancel}
\newcommand{\ket}[1]{| {#1} \rangle} 
\newcommand{\bra}[1]{\langle {#1} |} 

\usepackage[utf8]{inputenc}
\usepackage[T1]{fontenc}
\usepackage[english,serbian]{babel}
\usepackage[OT1]{fontenc}
\usepackage{mathrsfs}
\usepackage{graphicx}
\usepackage{amssymb}
\usepackage{amsbsy}
\usepackage{latexsym}
\usepackage{mathtools}
\usepackage{titlesec}
\usepackage{amsmath}
\usepackage{amsthm}
\usepackage{color}
\usepackage[utf8]{inputenc}
\usepackage{lipsum}
\usepackage[normalem]{ulem}
\usepackage{fontenc}
\newtheorem{theorem}{Theorem}[section]
\newtheorem{lemma}[theorem]{Lemma}
\newtheorem{corollary}{Corollary}[theorem]

\usepackage{chngcntr}

\def\nn{\nonumber}

\def\de{\textrm{d}}

\def\dj{d\kern-0.4em\char"16\kern-0.1em}
\def \Dj {\mbox{\raise0.3ex\hbox{-}\kern-0.4em D}}

\begin{document}

\preprint{APS/123-QED}

\title{Collapse Scenario and Final State of Evaporation for Schwarzschild Black Hole in Dimensionally-Reduced Model of Dilaton Gravity}

\author{Stefan \Dj or\dj evi\'{c} and Voja Radovanovi\'{c}}
\affiliation{Faculty of Physics, University of Belgrade, Studentski Trg 12-16, 11000 Belgrade, Serbia}

\selectlanguage{english}

\date{\today}

\begin{abstract}

We study a model of (1+1)-dimensional dilaton gravity derived from the four-dimensional Einstein-Hilbert action by dimensional reduction in a semiclassical approximation including back-reaction. The reduced action involves the cosmological constant and admits black hole solutions; among these, the solutions of interest are the evaporating black holes. We solve the equations of motion perturbatively by demanding that the initial state geometry is a Minkowski space-time. When the infalling matter intersects the space-time boundary, the black hole forms and begins to evaporate. We find that as the black hole evaporates, its horizon shrinks and at a finite space-time point, it meets the singularity and a shockwave occurs. Along this hypersurface, the metric can be continuously matched to a static end-state geometry. This end-state geometry is Minkowski space-time within the first-order of perturbation theory.

\end{abstract}

\maketitle

\selectlanguage{english}

\section{Introduction}

\par Various models of (1+1)-dimensional dilaton gravity, such as Jackiw-Teitelboim (JT) \cite{JT1,JT2} and Callan–Giddings–Harvey–Strominger (CGHS) \cite{CGHS}, have proved to be very useful for analytical investigation of the black hole formation and evaporation, as well as the eternal black holes. The utility of these models is that after integration of the fluctuations of matter fields and including $1$-loop quantum corrections, field equations can be made exactly solvable by introducing suitable correction terms as in Russo-Susskind-Thorlacius (RST), Bose-Parker-Peleg (BPP) and CGHS model \cite{RST1,RST2,BPP,QCGHS}. Even when not completely solvable, the two-dimensional equations are less complicated and can be solved perturbatively, as in the case of the DREH (Dimensionally-Reduced Einstein-Hilbert) model. The DREH model was studied in \cite{DREH1,Dimred} in case of the eternal black hole solution. Dimensional reduction of the Einstein-Hilbert action coupled with electrodynamics has also been analyzed in \cite{Frolov,MVM,MV}. A more comprehensive account of dilaton gravity models can be found in \cite{Fabri, DWV}.
\par Two-dimensional dilaton gravity models are mostly used as toy models for the purpose of resolving the information loss paradox \cite{Hawking1}. It is widely believed that the resolution of this paradox could be an important milestone at understanding the quantum nature of gravity, as well as the black hole entropy \cite{Bekenstein1,Bekenstein2,Hawking2}. Hawking's calculation \cite{Hawking3} predicts that the entanglement entropy of the Hawking radiation monotonically increases with time, even beyond the entropy  limit set up by the Bekenstein-Hawking (BH) formula: $S_{\text{BH}}=A(\text{horizon})/4G^{(4)}_{\text{N}}$ (we set $c=\hbar=k_{\text{B}}=1$, and $G^{(4)}_{\text{N}}$ stands for the Newton's constant in four space-time dimensions). This behavior of the entanglement entropy leads to the violation of the unitarity principle of quantum mechanics, and, in turn, to the loss of information, hence the name - information loss paradox.
\par The results derived from many of the two-dimensional gravity models indeed suggest that information does get lost in the process of black hole evaporation\cite{RST1,RST2,QCGHS,BPP}. Recently, a new approach has been studied in which the so called "island formula" for calculation of the entanglement entropy in gravitational systems is used to reproduce the Page curve \cite{Page_1993,Page_2013}; the curve that the entanglement entropy of the radiation should follow so that the unitary evolution is achieved.
\par The Page curve has recently been reproduced in variety of two-dimensional dilaton gravity models \cite{MaldecenaJT,RSTislands,BPPislands,Schislands,notes,DREH1,islands1,islands2,islands3,islands4}. Another important step in the analysis of unitarity is to find the end-state geometry, after the evaporation of the black hole. This is the problem that we shell address in this paper.
\par This paper is organized in the following manner. In the subsequent section we study a model of (1+1)-dimensional dilaton gravity derived from Einstein-Hilbert action by dimensional reduction. The main result of this section is a classical gravitational collapse scenario. Section III analyses the contribution of $1$-loop quantum corrections to the energy-momentum tensor of matter fields. The metric for an evaporating black hole is computed in section IV, up to first-order in $\hbar$, with the initial condition being that the space-time before the formation of the black hole is Minkowski space-time. Finally, in Section V, the end-state of the radiation is discussed. A conclusion and some proposals for the future work are given in Section VI.    
\section{DREH model}
\par The DREH model is a (1+1)-dimensional model of dilaton gravity obtained from the usual four-dimensional Einstein-Hilbert (EH) action by using a spherically symmetric ansatz and integrating out the angles. It is similar to the CGHS model of dilaton gravity \cite{CGHS}. In particular, it admits Schwarzschild black hole solution. 
\par Dimensional reduction is a well-known procedure, and we only give a brief overview. More technical details can be found in \cite{MVM}. We begin with the EH action in four dimensions,  
\begin{equation}
S_{\text{EH}}=\frac{1}{16\pi G_{\text{N}}^{(4)}}\int \de^{4}x\sqrt{-g^{(4)}}R^{(4)},
\end{equation}
where $G_{\text{N}}^{(4)}$ is the four-dimensional Newton's constant, $g^{(4)}_{AB}$ $(A,B=0,1,2,3)$ is the four-dimensional metric and $R^{(4)}$ is the four-dimensional curvature scalar.

Consider the following spherically symmetric ansatz for the metric, 
\begin{equation}\label{ansatz}
\rm{d}s^{2}=g_{\mu\nu}\mathrm{d} x^{\mu}\mathrm{d} x^{\nu}+\lambda^{-2}e^{-2\phi}\left[\rm{d}\theta^{2}
+\sin^{2}\theta \rm{d}\varphi^{2}\right],     
\end{equation}
where $\mu,\nu=0,1$, and $g_{\mu\nu}$ depends only on $x^{0}$ and $x^{1}$,the dilaton field $\phi$ is related to the radial coordinate $r=\lambda^{-1}e^{-\phi}$, and $\lambda^{2}$ is a constant that plays the role of the cosmological constant in the reduced theory. Using the ansatz (\ref{ansatz}), we can derive the relation between the curvature scalar $R$ of the reduced (1+1)-dimensional theory, and the curvature scalar $R^{(4)}$ of the four-dimensional theory,
\begin{align}
R^{(4)}=R+2(\nabla\phi)^{2}+2\lambda^{2}e^{2\phi}-2e^{2\phi}\Box e^{-2\phi}.
\end{align}
Also, we have
\begin{align}
\de^{4}x\sqrt{-g^{(4)}}=\de^{2}x\de\theta\de\varphi\sqrt{-g}\frac{e^{-2\phi}}{\lambda^{2}}\sin^{2}{\theta}.
\end{align}
The EH gravity action, reduces to the dilaton gravity action $S_{\phi}$  (up to a surface term):
\begin{equation}\label{Dilaton_action}
S_{\phi}=\frac{1}{4G}\int \de^{2}x\sqrt{-g}\Big[e^{-2\phi}\left(R+2(\nabla\phi)^{2}\right)+2\lambda^{2}\Big],    
\end{equation}
where $G\equiv\lambda^{2}G_{\text{N}}^{(4)}$ is the Newton's constant of the reduced theory. 
\par To take the quantum corrections into account, we introduce a massless scalar field $f$ minimally coupled to gravity. The corresponding action is denoted by $S_{m}$, so the classical DREH action is given by:
\begin{align}
&S_{\text{DREH}}=S_{\phi}+S_{m}\nn\\
&=\frac{1}{4G}\int \de^{2}x\sqrt{-g}\Big[e^{-2\phi}\left(R+2(\nabla\phi)^{2}\right)+2\lambda^{2}\Big]\nn\\
&-\frac{1}{2}\int \de^{2}x\sqrt{-g}\left(\nabla f\right)^{2}.\label{dejstvo}
\end{align}
The classical field equations are obtained by varying $S_{\rm{DREH}}$ with respect to $g_{\mu\nu}$, $\phi$ and $f$, respectively, resulting in:
\begin{widetext}
\begin{align} 
\left[2\nabla_{\mu}\nabla_{\nu}\phi-2\nabla_{\mu}\phi\nabla_{\nu}\phi +g_{\mu\nu}\left(3(\nabla\phi)^{2}-2\Box\phi-\lambda^{2}e^{2\phi}\right)\right] e^{-2\phi}&=2GT^{(f)}_{\mu\nu,\text{class}}, \label{g-equation}\\
            (\nabla\phi)^{2}-\Box\phi&=\frac{R}{2},\label{phi-equation}\\
            \Box f&=0\label{box_f},
\end{align}
\end{widetext}
with classical energy-momentum tensor of field $f$, given by:
\begin{equation}
T^{(f)}_{\mu\nu,\text{class}}=\frac{-2}{\sqrt{-g}}\frac{\delta S_{m}}{\delta g^{\mu\nu}}=\nabla_{\mu}f\nabla_{\nu}f-\frac{1}{2}g_{\mu\nu}(\nabla f)^{2}.
\end{equation}
Equations (\ref{g-equation}-\ref{phi-equation}) can be simplified. By contracting equation (\ref{g-equation}) with the metric tensor $g^{\mu\nu}$ we obtain: $\Box e^{-2\phi}=2\lambda^{2}+2GT$, where $T=g^{\mu\nu}T_{\mu\nu}$. Introducing a new field $\varphi=e^{-\phi}$, after a few manipulations we obtain:
\begin{align}
    \nabla_{\mu}\nabla_{\nu}\varphi-\frac{1}{2}R_{\mu\nu}\varphi&=-G\left(T_{\mu\nu}-\frac{1}{2}g_{\mu\nu}T\right)\varphi^{-1},\label{chi}\\
    \Box\varphi^{2}&=2\lambda^{2}+2GT.\label{chi+box}
\end{align}
Equations (\ref{chi}-\ref{chi+box}), along with (\ref{box_f}) is the set of quations that we will solve.
\subsection{Vacuum solution of classical equations of motion}

Now we will find the general solution to the classical vacuum equations of motion in the conformal gauge: $\de s^{2}=-e^{2\rho}\de x^{+}\de x^{-}$. In this gauge, the non-vanishing components of the connection, Ricci curvature and scalar curvature are:
\begin{equation}
    \Gamma^{\pm}_{\pm\pm}=2\partial_{\pm}\rho,\hspace{2mm}R_{+-}=-2\partial_{+}\partial_{-}\rho,\hspace{2mm}R=8e^{-2\rho}\partial_{+}\partial_{-}\rho.\label{konf_gauge}
\end{equation}
\noindent Equations of motion (\ref{chi}) and (\ref{chi+box}) become:
\begin{align}
    \partial_{\pm}\left(e^{-2\rho}\partial_{\pm}\varphi\right)&=-GT_{\pm\pm}\varphi^{-1}e^{-2\rho},\label{ppmm}\\
    \partial_{+}\partial_{-}\varphi+\varphi\partial_{+}\partial_{-}\rho&=0,\label{pm}\\
    \partial_{+}\partial_{-}\varphi^{2}+\frac{\lambda^{2}}{2}e^{2\rho}&=0.\label{chi2}
\end{align}
\noindent We first solve these equations in the vacuum case $T_{\pm\pm}=0$. By integrating equation (\ref{ppmm}), we get:
\begin{equation}
    \partial_{\pm}\varphi=\frac{\lambda}{2}\mathrm{F}^{\mp}(x^{\mp})e^{2\rho},\label{chi_izvod}
\end{equation}
\noindent where $\mathrm{F}^{\mp}(x^{\mp})$ are arbitrary functions; we can think of them as components of a vector field. Functions $\partial_{\pm}\varphi$ satisfy the following equation $\partial_{+}\partial_{-}\varphi=\partial_{-}\partial_{+}\varphi$, which gives:
\begin{align}
    \nabla_{+}\mathrm{F}^{+}&=\nabla_{-}\mathrm{F}^{-},\label{konzistentnost_1}\\
    \partial_{+}\partial_{-}\varphi&=\frac{\lambda}{2}\nabla_{+}\mathrm{F}^{+}e^{2\rho}.\label{chi_drugi}
\end{align}
\noindent Combining this result with equations (\ref{chi2}) and (\ref{chi_izvod}) we obtain $\varphi$ in terms of $\mathrm{F}^{\pm}$ and $\rho$:
\begin{equation}
    \varphi=-\frac{\lambda}{2}\frac{1+\mathrm{F}^{+}\mathrm{F}^{-}e^{2\rho}}{\nabla_{+}\mathrm{F}^{+}}.\label{chi_calc}
\end{equation}
\noindent Next, we calculate $\partial_{+}\partial_{-}\rho$ by substituting equations (\ref{chi_calc}) and (\ref{chi_drugi}) into equation (\ref{pm}):
\begin{equation}
    \partial_{+}\partial_{-}\rho=\frac{\left(\nabla_{+}\mathrm{F}^{+}\right)^{2}e^{2\rho}}{1+\mathrm{F}^{+}\mathrm{F}^{-}e^{2\rho}}.\label{rhooo}
\end{equation}
\noindent It is easy to see that $\partial_{-}\left(1+\mathrm{F}^{+}\mathrm{F}^{-}e^{2\rho}\right)=\mathrm{F}^{+}\nabla_{+}\mathrm{F}^{+}e^{2\rho}$ and $\partial_{+}\partial_{-}\rho=\frac{\partial_{-}\left(\nabla_{+}\mathrm{F}^{+}\right)}{2\mathrm{F}^{+}}$. This, in combination with equation (\ref{rhooo}) yields:
\begin{equation}
    \nabla_{+}\mathrm{F}^{+}=\mathrm{G}^{+}(x^{+})\left(1+\mathrm{F}^{+}\mathrm{F}^{-}e^{2\rho}\right)^{2},
\end{equation}
\noindent where $\mathrm{G}^{+}(x^{+})$ is an arbitrary function. Differentiating with respect to $x^{+}$ we get the same equation, which means we also have:
\begin{equation}
    \nabla_{+}\mathrm{F}^{+}=\mathrm{G}^{-}(x^{-})\left(1+\mathrm{F}^{+}\mathrm{F}^{-}e^{2\rho}\right)^{2},
\end{equation}
\noindent where $\mathrm{G}^{-}(x^{-})$ is also an arbitrary function. Since these two equations are the same, the functions $G^{\pm}$ must be the same constant: $\mathrm{G}^{+}(x^{+})=\mathrm{G}^{-}(x^{-})=-\frac{1}{2a}$, where the meaning of $a$ will soon become apparent. Thus, with the help of the equation (\ref{chi_calc}) we have the following result:
\begin{equation}
    \nabla_{+}\mathrm{F}^{+}=-\frac{1}{2a}\left(1+\mathrm{F}^{+}\mathrm{F}^{-}e^{2\rho}\right)^{2}=-\frac{\lambda}{2}\frac{\lambda a}{\varphi^{2}}.\label{nablaFjna}
\end{equation}
\noindent We can verify that $\nabla_{+}\mathrm{F}^{+}$ also satisfies: $\nabla_{+}\mathrm{F}^{+}=\mathrm{F}^{\pm}\partial_{\pm}\ln{(\mathrm{F}^{+}\mathrm{F}^{-}e^{2\rho})}$. Together with equations (\ref{nablaFjna}) and (\ref{chi_calc}) the derivatives of $\varphi$ may be expressed as:
\begin{equation}
    \partial_{\pm}\varphi=-\frac{\lambda}{2\mathrm{F}^{\pm}}\frac{\varphi-\lambda a}{\varphi}.\label{chi_izvodi}
\end{equation}
\noindent Now, we can integrate this equation. The final result is given by the following two equations:
\begin{align}
    \varphi+\lambda a\ln{\left(\frac{\varphi}{\lambda a}-1\right)}&=-\frac{\lambda}{2}\left(\int\frac{\de x^{+}}{\mathrm{F}^{+}}+\int\frac{\de x^{-}}{\mathrm{F}^{-}}\right),\label{chi_resenje}\\
    \mathrm{F}^{+}\mathrm{F}^{-}e^{2\rho}&=\frac{\lambda a}{\varphi}-1.\label{rho_resenje}
\end{align}
The solution to the classical vacuum equations is defined by one constant and two functions. We write these as: $(\mathrm{F}^{+},\mathrm{F}^{-},a)$. The meaning of arbitrary functions $\mathrm{F}^{\pm}$ is now obvious; they represent the conformal transformations of the coordinates. We can define new coordinates: 
\begin{equation}
    \frac{\de\sigma^{\pm}}{\de x^{\pm}}=\mp\frac{1}{\mathrm{F}^{\pm}}.\label{koord_transf}
\end{equation}
\noindent These $\sigma^{\pm}$ coordinates are interpreted as Eddington-Finkelstein coordinates. Then equation (\ref{chi_resenje}) represents the formula for the tortoise coordinate:
\begin{equation}
    r_{*}\equiv\sigma=\frac{1}{2}(\sigma^{+}-\sigma^{-}).
\end{equation} 
Equation (\ref{rho_resenje}) is the equation for the metric tensor in the Eddington-Finkelstein coordinates. We can directly check this by calculating the interval: \begin{equation}
    \de s^{2}=-e^{2\rho}\de x^{+}\de x^{-}=\mathrm{F}^{+}\mathrm{F}^{-}e^{2\rho}\de\sigma^{+}\de\sigma^{-}.
\end{equation} 
We conclude that the solution defined by $(-1,1,a)$ amounts to choosing Eddington-Finkelstein coordinates. This choice we call the Eddington-Finkelstein gauge. The constant $a$ is the Schwarzschild radius. The dilation can be chosen to be proportional to the radial coordinate: $\varphi=e^{-\phi}=\lambda r$. So we see that we have reproduced the Schwarzschild solution in the same way as in \cite{DREH1}.


\subsection{Classical gravitational collapse}\label{grav_kolaps}

To formulate the classical gravitational collapse, we need to introduce the matter in the model. The easiest way to do this is to define the energy-momentum tensor as an incoming shockwave: $T_{++}=s\delta(x^{+}-x_{0}^{+})$ and $T_{--}=0$. The corresponding Penrose diagram is shown in Figure \ref{fig1}. There are two parts of space-time (I and II in Figure \ref{fig1}), which we need to connect continuously. Metric and dilaton need to be continuous functions of $x^{+}$ and $x^{-}$ at the $x^{+}=x_{0}^{+}$ hypersurface: 
\begin{align}
    e^{2\rho_{\mathrm{I}}}\bigg{|}_{x^{+}=x_{0}^{+}-\epsilon}&=e^{2\rho_{\mathrm{II}}}\bigg{|}_{x^{+}=x_{0}^{+}+\epsilon},\label{krpljenje_metrika}\\
    \varphi_{I}\bigg{|}_{x^{+}=x_{0}^{+}-\epsilon}&=\varphi_{II}\bigg{|}_{x^{+}=x_{0}^{+}+\epsilon},\label{krpljenje_polje}
\end{align}
\noindent when $\epsilon\to0$, and indices I and II correspond to part I and part II of space-time, respectively (see Figure \ref{fig1}). In part I, we have the vacuum solution defined by $(\mathrm{F}_{\mathrm{I}}^{+},\mathrm{F}_{\mathrm{I}}^{-},a_{\mathrm{I}})$, while in part II, the solution is defined by $(\mathrm{F}_{\mathrm{II}}^{+},\mathrm{F}_{\mathrm{II}}^{-},a_{\mathrm{II}})$.
\begin{figure}[h]
    \begin{center}
    \includegraphics[width=7cm, height=8.5cm]{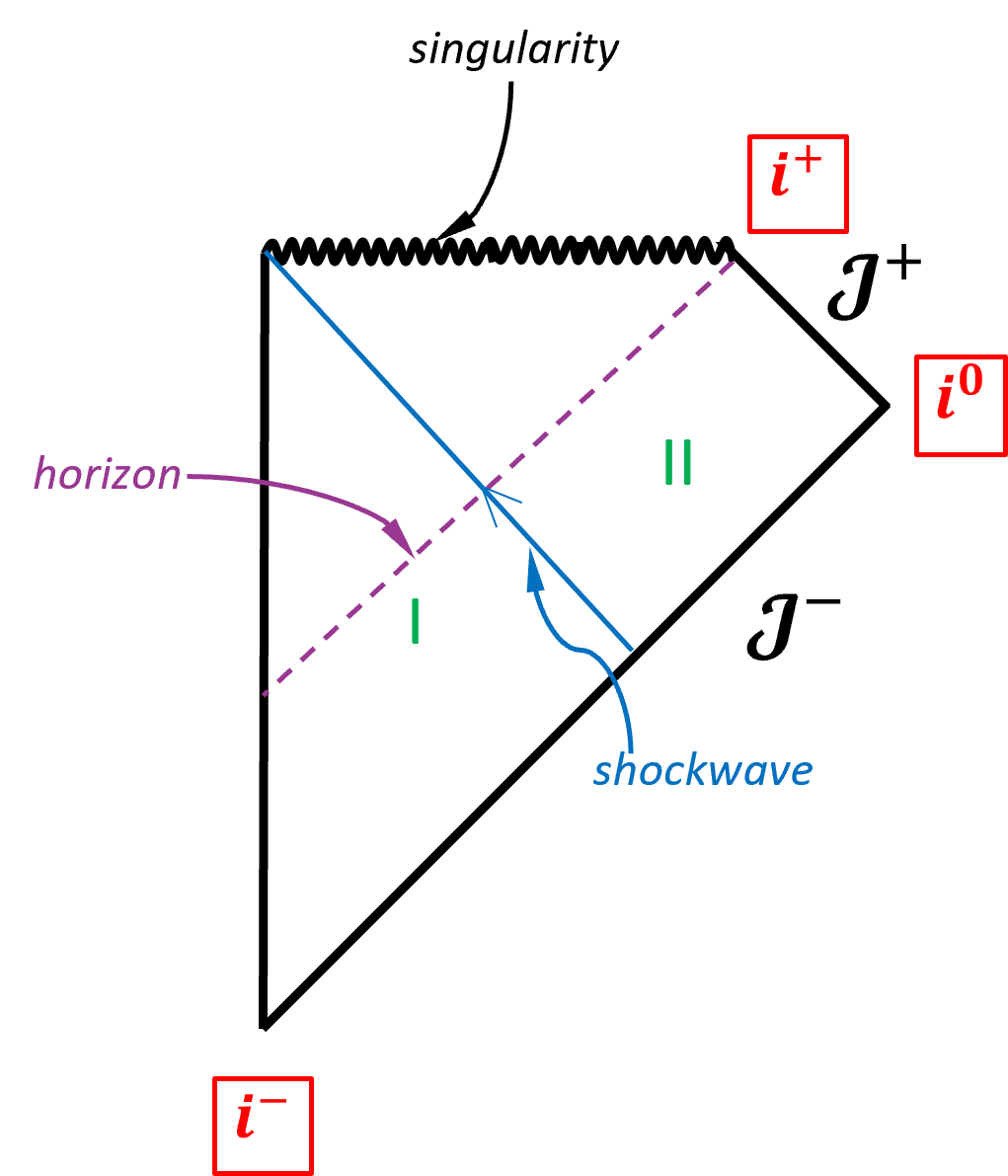}
    \caption{Penrose diagram of a gravitational collapse}\label{fig1}
    \end{center}
\end{figure}
\noindent Integrating equation (\ref{ppmm}), we obtain:
\begin{align}
    e^{-2\rho_{\mathrm{II}}}\partial_{+}\varphi_{II}\bigg{|}_{x^{+}=x^{+}_{0}+\epsilon}&-e^{-2\rho_{\mathrm{I}}}\partial_{+}\varphi_{I}\bigg{|}_{x^{+}=x^{+}_{0}-\epsilon}\nonumber\\
    &=-Gse^{-2\rho}\varphi^{-1}\bigg{|}_{x^{+}=x^{+}_{0}}.\label{krp}
\end{align}
\noindent Equation (\ref{krp}), together with equations (\ref{chi_izvod}) and (\ref{krpljenje_metrika}), implies the jump in the derivatives of the dilaton field, which, in turn, implies the jump in $\mathrm{F}^{-}$ given by:
\begin{equation}
    \mathrm{F}^{-}_{\mathrm{II}}=\mathrm{F}^{-}_{\mathrm{I}}-\frac{2sG}{\lambda}e^{-2\rho}\varphi^{-1}\bigg{|}_{x^{+}=x^{+}_{0}}.\label{skok_Fm}
\end{equation}
\noindent The derivative $\partial_{-}\varphi$ is continuous, implying $\mathrm{F}^{+}_{\mathrm{II}}=\mathrm{F}^{+}_{\mathrm{I}}$. Using equation (\ref{krpljenje_metrika}) and solution (\ref{rho_resenje}) together with equation (\ref{krpljenje_polje}) we get that the constant $a$ changes at $x^{+}=x^{+}_{0}$, as well:
\begin{equation}
    a_{\mathrm{II}}=a_{\mathrm{I}}-\frac{2sG}{\lambda^{2}}\mathrm{F}^{+}_{\mathrm{I}}(x^{+}_{0}).\label{skok_a}
\end{equation}
For the correct interpretation, we have to connect the constant $s$ with the mass of the shockwave. This can be done with the help of the time-translation Killing vector $\xi$ (on part I of space-time):
\begin{equation}
    M=\int_{\Sigma}\de\sigma\sqrt{\gamma}n^{\mu}\xi^{\nu}T_{\mu\nu},\label{masa}
\end{equation}
\noindent where $\Sigma$ is a space-like hypersurface, $n$ is an orthogonal vector to the surface $\Sigma$ normalized as $n_{\mu}n^{\mu}=-1$, $\gamma$ is the induced metric on $\Sigma$ and $T_{\mu\nu}$ is the energy-momentum tensor. Solving the Killing equation, we find that the time-translation Killing vector is given by:
\begin{equation}
    \xi=\mathrm{F}^{-}\partial_{-}-\mathrm{F}^{+}\partial_{+},\label{Kiling}
\end{equation}
\noindent i.e., $\mathrm{F}^{\pm}$ are components of $\xi$. Since $\xi$ is orthogonal to $\Sigma$, we can choose $n$ to be proportional to $\xi$: $n^{\mu}=\frac{\xi^{\mu}}{\sqrt{-\xi\cdot\xi}}$. The induced metric on the hypersurface $\Sigma$ is then: $\de s^{2}=-\mathrm{F}^{+}\mathrm{F}^{-}e^{2\rho}\de\sigma^{2}$. Placing all of this into equation (\ref{masa}) we obtain:
\begin{equation}
    M=-s\mathrm{F}^{+}(x^{+}_{0}).\label{masa_izracunata}
\end{equation}
\noindent If the solution in part I of space-time is given by $(\mathrm{F^{+}_{\mathrm{I}}},\mathrm{F^{-}_{\mathrm{I}}},a_{\mathrm{I}})$, then the solution in  part II of space-time is given by:
\begin{equation}
    \left(\mathrm{F^{+}_{\mathrm{I}}},\mathrm{F^{-}_{\mathrm{I}}}\left(1+\frac{2MG/\lambda}{\lambda a_{\mathrm{I}}-\varphi(x^{+}_{0},x^{-})}\right),a_{\mathrm{I}}+\frac{2MG}{\lambda^{2}}\right).\label{resenje_II}
\end{equation}
\noindent Equation (\ref{resenje_II}) tells us what form the solution in part II of space-time takes if we already have a black hole solution in part I of space-time and that black hole consumes matter of mass $M$. Without loss of generality, we can choose the Eddington-Finkelstein gauge in part I of space-time, which amounts to $\mathrm{F}^{+}_{\mathrm{I}}=-1$ and $\mathrm{F}^{-}_{\mathrm{I}}=1$. To analyze gravitational collapse, we can further choose that the metric of part I is a Minkowski space-time by setting $a_{\mathrm{I}}=0$. The corresponding dilaton is given by $\varphi_{\mathrm{I}}=\frac{\lambda}{2}\left(\sigma^{+}-\sigma^{-}\right)$. Then, the equation (\ref{resenje_II}) implies that part II is a black hole of mass $M$:
\begin{equation}
    \left(-1,1-\frac{4MG/\lambda^{2}}{\sigma_{0}^{+}-\sigma^{-}},\frac{2MG}{\lambda^{2}}\right).
\end{equation} 
Let us define the notation that we will use from now on. First, we assume that the solution in part I of space-time is the Minkowski vacuum solution, defined by: $(-1,1,0)$. Next, we define $\hat{\varphi}(\sigma^{-})=\frac{\lambda}{2}\left(\sigma^{+}_{0}-\sigma^{-}\right)$. Using this function, we can express the coordinate transformations in part II as:
\begin{equation}
    \mathrm{F}^{+}=-1,\hspace{2mm}\hspace{2mm}\mathrm{F}^{-}=\frac{\hat{\varphi}-\lambda a}{\hat{\varphi}}.\label{FpFm}
\end{equation}
\noindent Notice that we have dropped the index II on functions $\mathrm{F}^{\pm}_{\mathrm{II}}$ and constant $a$ as well, so we have $a=\frac{2MG}{\lambda^{2}}$. In part II of space-time the relation (\ref{chi_resenje}) still holds, so we can introduce new Eddington-Finkelstein coordinates $\hat{\sigma}^{\pm}$. In terms of $\sigma^{\pm}$, they are given by:
\begin{align}
    \hat{\sigma}^{+}&=\sigma^{+},\label{koord_smena_p}\\
    \hat{\sigma}^{-}&=\sigma^{+}_{0}-\frac{2}{\lambda}\left[\hat{\varphi}+\lambda a\ln{\left(\frac{\hat{\varphi}}{\lambda a}-1\right)}\right].\label{koord_smena_m}
\end{align}
\noindent Additionally, we define a new coordinate $\delta$: 
\begin{equation}
    \delta(\sigma^{+})=e^{-\frac{\sigma^{+}-\sigma^{+}_{0}}{2a}},\label{delta_def}
\end{equation}
\noindent which satisfies $\delta<1$ in part II of space-time ($\sigma^{+}>\sigma^{+}_{0}$). This property will be used extensively later on. Solutions (\ref{chi_resenje}) and (\ref{rho_resenje}) in part II can now be written in the following form (in terms of functions $\delta$ and $\hat{\varphi}$):
\begin{align}
    &\varphi+\lambda a\ln{\left(\frac{\varphi}{\lambda a}-1\right)}=-\lambda a\ln{\delta}+\hat{\varphi}+\lambda a\ln{\left(\frac{\hat{\varphi}}{\lambda a}-1\right)},\nonumber\\
    &\mathrm{F}^{+}\mathrm{F}^{-}e^{2\rho}=\frac{\lambda a}{\varphi}-1.\label{resenje_II_}
\end{align}
Now we can analyse solution (\ref{resenje_II_}). The equation for the metric tells us that we have a black hole with the singularity at $\varphi=0$ and the horizon at $\varphi=\lambda a$. Keep in mind that $\varphi=\lambda \hat{r}$, where $\hat{r}$ is the radial coordinate in part II of space-time. Looking again at solution (\ref{resenje_II_}), we see that $\varphi=\lambda a$ corresponds to $\hat{\varphi}=\lambda a$. This means that the horizon and the singularity are respectively given by: 
\begin{align}
    \sigma^{-}_{\mathrm{H}}&=\sigma^{+}_{0}-2a\label{horizont},\\
    \ln{\delta}_{S}&=\frac{\hat{\varphi}_{\mathrm{S}}}{\lambda a}+\ln{\left|\frac{\hat{\varphi}_{\mathrm{S}}}{\lambda a}-1\right|}\label{singularnost}
\end{align}
\noindent When $t\to\infty$, we have $\sigma^{+}\to\infty$, which is equivalent to $\delta\to0$. Using equation (\ref{singularnost}) we deduce that the horizon and the singularity intersect at future time infinity $i^{+}$, which is the expected result. Both lines are represented in the conformal diagram (see Figure \ref{fig1}).
\section{Quantum corrections in the DREH model}

Having established the classical DREH model, we consider quantization of matter fields (a single massless scalar field) on the classical background of the Schwarzschild black hole. Quantum corrections come in the form of Polyakov-Liouville (PL) action \cite{Polyakov}, 
\begin{align}\label{SPL}
S_{\text{PL}}=-\frac{\hbar}{96\pi}\int\de^{2}x\int&\de^{2}x'\sqrt{-g(x)}\sqrt{-g(x')}\nonumber \\
&\times R(x)G(x-x')R(x'),
\end{align}
where $G(x-x')$ stands for the Green's function for the massless Klein-Gordon equation in curved (1+1)-dimensional space-time. This action represents the $1$-loop effective action obtained by integrating out fluctuations of the massless scalar field, 
\begin{equation}
e^{\frac{i}{\hbar}S_{\text{PL}}}=\int\mathcal{D}\chi e^{\frac{i}{\hbar}\int\de^{2}x\sqrt{-g}[-\frac{1}{2}(\nabla \chi)^{2}]},   
\end{equation}
and it can be converted into a local form by introducing an auxiliary field $\psi$,
\begin{equation}
S_{\text{PL}}=-\frac{\hbar}{96\pi}\int\de^{2}x\sqrt{-g}\left[2R\psi+\left(\nabla\psi\right)^{2}\right].\label{aa1111} 
\end{equation}
Action (\ref{aa1111}) is on-shell equivalent to $(\ref{SPL})$. The field equation for the auxiliary field is 
\begin{equation}
\Box\psi=R.\label{psi_jna}
\end{equation}
The full action for the $1$-loop quantum DREH model is given by
\begin{equation}
S=S_{\rm{DREH}}+S_{\text{PL}}.    
\end{equation}
Variation of $S_{\text{PL}}$ with respect to $g_{\mu\nu}$ gives the quantum correction to the energy-momentum tensor for the scalar field $f$,
\begin{align} \label{T_1loop}
\langle\Delta T_{\mu\nu}^{(f)}\rangle&=\frac{-2}{\sqrt{-g}}\frac{\delta S_{\text{PL}}}{\delta g^{\mu\nu}}
=\frac{\hbar}{48\pi} \Bigg[-2\nabla_{\mu}\nabla_{\nu}\psi\nonumber\\
&+\nabla_{\mu}\psi\nabla_{\nu}\psi
+g_{\mu\nu}\left(2\Box\psi-\frac{1}{2}(\nabla\psi)^{2}\right)\Bigg].
\end{align}
To define $\langle\Delta T_{\mu\nu}^{(f)}\rangle=\langle\Psi\vert\Delta T_{\mu\nu}^{(f)}\vert\Psi\rangle$, we also need to specify the quantum state $\vert\Psi\rangle$ that we are considering.  
The full energy-momentum tensor consists of the classical part and the $1$-loop quantum correction coming from the PL effective action,
\begin{equation}
T_{\mu\nu}^{(f)}=T_{\mu\nu,class}^{(f)}+\langle\Delta T_{\mu\nu}^{(f)}\rangle.    
\end{equation}
The metric equation (\ref{g-equation}) changes only due to this quantum correction of the energy-momentum tensor and reads:
\begin{widetext}
\begin{align} 
\Big[2\nabla_{\mu}\nabla_{\nu}\phi-2\nabla_{\mu}\phi\nabla_{\nu}\phi +g_{\mu\nu}\Big(3(\nabla\phi)^{2}-2\Box\phi&-\lambda^{2}e^{2\phi}\Big)\Big] e^{-2\phi}=2G_{\text{N}}\Bigg\{\nabla_{\mu}f\nabla_{\nu}f-\frac{1}{2}g_{\mu\nu}(\nabla f)^{2}\nn\\
&+\frac{\hbar}{48\pi} \Bigg[-2\nabla_{\mu}\nabla_{\nu}\psi
+\nabla_{\mu}\psi\nabla_{\nu}\psi
+g_{\mu\nu}\left(2\Box\psi-\frac{1}{2}(\nabla\psi)^{2}\right)\Bigg]\Bigg\}.\label{kvantna_jna}
\end{align}
\end{widetext}

Since the quantum correction to the energy-momentum tensor (\ref{T_1loop}) depends only on the auxiliary field, we need to solve equation (\ref{psi_jna}). In the conformal gauge (\ref{konf_gauge}), this equation reads: $\partial_{+}\partial_{-}\left(\psi+2\rho\right)=0$. Its solution in terms of the metric $\rho$, is given by:
\begin{equation}
    \psi=-2\rho+f_{+}(x^{+})+f_{-}(x^{-}),\label{psi_resenje}
\end{equation}
\noindent where $f_{\pm}(x^{\pm})$ are arbitrary functions. Then, the quantum correction to the energy-momentum tensor defined in equation (\ref{T_1loop}), calculated in the conformal gauge (\ref{konf_gauge}), in terms of the metric $\rho$ reads:
\begin{align}
    \langle\Delta T_{\pm\pm}^{(f)}\rangle&=\frac{\hbar}{12\pi}\left[\partial^{2}_{\pm}\rho-\left(\partial_{\pm}\rho\right)^{2}-t_{\pm}(x^{\pm})\right],\label{TEI_PP_MM}\\
    \langle\Delta T_{+-}^{(f)}\rangle&=-\frac{\hbar}{12\pi}\partial_{+}\partial_{-}\rho,\label{TEI_PM}
\end{align}
\noindent where functions $t_{\pm}(x^{\pm})$ are given by: $t_{\pm}(x^{\pm})=\frac{1}{2}\partial^{2}_{\pm}f_{\pm}-\frac{1}{4}\left(\partial_{\pm}f_{\pm}\right)^{2}$. After adding quantum corrections, the equations of motion become as follows:
\begin{widetext}
\begin{align}
    &\partial_{\pm}\left(e^{-2\rho}\partial_{\pm}\varphi\right)=-Gs\delta(x^{+}-x^{+}_{0})\varphi^{-1}e^{-2\rho}+\frac{\hbar G}{12\pi}\left[\left(\partial_{\pm}\rho\right)^{2}-\partial^{2}_{\pm}\rho+t_{\pm}(x^{\pm})\right]\varphi^{-1}e^{-2\rho},\label{ppmm_q}\\
    &\partial_{+}\partial_{-}\varphi+\varphi\partial_{+}\partial_{-}\rho=0,\label{pm_q}\\
    &\partial_{+}\partial_{-}\varphi^{2}+\frac{\lambda^{2}}{2}e^{2\rho}+\frac{\hbar G}{6\pi}\partial_{+}\partial_{-}\rho=0.\label{chi2_q}
\end{align}
\end{widetext}
\noindent Equation (\ref{pm_q}) is the same as equation (\ref{pm}). This is because the right-hand side of equation (\ref{chi}) is zero, even after adding the quantum correction. The reason is the gravitational trace anomaly. The trace of the quantum correction to the energy-momentum tensor is given by $\langle\Delta T^{(f)}\rangle=\frac{\hbar}{3}e^{-2\rho}\partial_{+}\partial_{-}\rho$. The same reason drives the change of equation (\ref{chi2}) to equation (\ref{chi2_q}). 
\par These quantum-corrected equations cannot be exactly solved. We will solve them perturbatively to the first-order in the perturbative parameter $\varepsilon=\frac{\hbar G}{12\pi}$. We expand the dilaton field as $\varphi\mapsto\varphi+\varepsilon\theta$ and the metric as $\rho\mapsto\rho+\varepsilon\alpha$. The quantum correction to equations (\ref{ppmm_q}-\ref{chi2_q}), in terms of $\theta$ and $\alpha$, reads:
\begin{widetext}
\begin{align}
    &\partial_{+}\left[e^{-2\rho}\partial_{+}\theta-2\alpha e^{-2\rho}\partial_{+}\varphi\right]=Gs\left(\frac{\theta}{\varphi}+2\alpha\right)e^{-2\rho}\varphi^{-1}\delta(x^{+}-x^{+}_{0})+\left[\left(\partial_{+}\rho\right)^{2}-\partial^{2}_{+}\rho+t_{+}(x^{+})\right]\varphi^{-1}e^{-2\rho},\label{jna_pp}\\
    &\partial_{-}\left[e^{-2\rho}\partial_{-}\theta-2\alpha e^{-2\rho}\partial_{-}\varphi\right]=\left[\left(\partial_{-}\rho\right)^{2}-\partial^{2}_{-}\rho+t_{-}(x^{-})\right]\varphi^{-1}e^{-2\rho},\label{jna_mm}\\
    &\partial_{+}\partial_{-}\theta+\varphi\partial_{+}\partial_{-}\alpha+\theta\partial_{+}\partial_{-}\rho=0,\label{jna_pm}\\
    &\partial_{+}\partial_{-}\left(\theta\varphi+\rho\right)+\frac{\lambda^{2}}{2}\alpha e^{2\rho}=0.\label{jna_chi2}
\end{align}
\end{widetext}
\subsection{Vacuum solution to the quantum-corrected equations of motion}

As in the case of the classical gravitational collapse, we will drop the $\delta$-function term in equation (\ref{jna_pp}) and find the general solution to the quantum-corrected equations of motion in vacuum. Then, we will patch up the solution over the $x=x^{+}_{0}$ hypersurface by integrating equation (\ref{jna_pp}). Direct calculation of the term $\left(\partial_{\pm}\rho\right)^{2}-\partial^{2}_{\pm}\rho+t_{+}(x^{\pm})$ yields:
\begin{align}
    \left(\partial_{\pm}\rho\right)^{2}-\partial^{2}_{\pm}\rho&+t_{\pm}(x^{\pm})=\frac{1}{4}\left(\partial_{\pm}\ln{\mathrm{F}^{\pm}}\right)^{2}+\frac{1}{2}\partial_{\pm}^{2}\ln{\mathrm{F}^{\pm}}\nonumber\\
    &+t_{\pm}(x^{\pm})+\left(\frac{\lambda}{2\mathrm{F}^{\pm}}\right)^{2}\lambda a\frac{\varphi-\frac{3}{4}\lambda a}{\varphi^{4}}.\label{medju}
\end{align}
\noindent Using coordinate transformation (\ref{koord_transf})
we can easily see that the first two terms on the right hand side reduce to the Schwarzian derivative (defined in Appendix $A$):
\begin{equation}
    \frac{1}{4}\left(\partial_{\pm}\ln{\mathrm{F}^{\pm}}\right)^{2}+\frac{1}{2}\partial_{\pm}^{2}\ln{\mathrm{F}^{\pm}}(x^{\pm})=-\frac{1}{2}D_{x^{\pm}}\left[\sigma^{\pm}\right].\label{m2}
\end{equation}
\noindent Together with the term $t_{\pm}(x^{\pm})$, we get the transformation of this variable, given by equation (\ref{t_pm_trans}) in Appendix $A$. Placing this back into equations (\ref{jna_pp}) and (\ref{jna_mm}) we arrive at:
\begin{align}
    \partial_{\pm}\left(e^{-2\rho}\partial_{\pm}\theta-\lambda\alpha\mathrm{F}^{\mp}\right)&=\frac{1}{\mathrm{F}^{\pm2}}t_{\pm}(\sigma^{\pm})e^{-2\rho}\varphi^{-1}\nonumber\\
    &+\frac{\lambda}{2}\mathrm{F}^{\mp}\lambda a\frac{\varphi-\frac{3}{4}\lambda a}{\varphi^{3}(\varphi-\lambda a)^{2}}\partial_{\pm}\varphi,\label{m3}
\end{align}
\noindent Next we define the quantity $\mathcal{D}(\varphi)$ through its derivative:
\begin{equation}
    \frac{\de\mathcal{D}}{\de\varphi}=\lambda a\frac{\varphi-\frac{3}{4}\lambda a}{\varphi^{3}(\varphi-\lambda a)^{2}}.\label{def_D}
\end{equation}
\noindent By integrating equation (\ref{m3}), we obtain the derivatives of $\theta$:
\begin{widetext}
\begin{align}
    \partial_{+}\theta&=\frac{\lambda}{2}\left[\mathrm{F}^{-}(\mathcal{D}+2\alpha)+\mathrm{G}^{-}\right]e^{2\rho}-\mathrm{F}^{-}e^{2\rho}\int\frac{\de x^{+}}{\mathrm{F}^{+}}\frac{t_{+}(\sigma^{+})}{\varphi-\lambda a},\label{theta_izvod_plus}\\
    \partial_{-}\theta&=\frac{\lambda}{2}\left[\mathrm{F}^{+}(\mathcal{D}+2\alpha)+\mathrm{G}^{+}\right]e^{2\rho}-\mathrm{F}^{+}e^{2\rho}\int\frac{\de x^{-}}{\mathrm{F}^{-}}\frac{t_{-}(\sigma^{-})}{\varphi-\lambda a}\label{theta_izvod_minus},
\end{align}
\end{widetext}
\noindent where $G^{\pm}(x^{\pm})$ are a new set of arbitrary functions. Later we will see that these functions represent quantum correction to the coordinate transformations. We have not yet found the exact form of the derivatives $\partial_{\pm}\theta$ since equations (\ref{theta_izvod_plus}) and (\ref{theta_izvod_minus}) have unevaluated integrals. The problem lies in the functional dependence $\varphi=\varphi(x^{+},x^{-})$ given through the transcendental equation (\ref{chi_resenje}). For now, the solution will be written in terms of these integrals. To integrate equation (\ref{theta_izvod_plus}) once more, we need to eliminate the unknown function $\alpha$. This can be done by employing equation (\ref{jna_chi2}). In conjunction with equations (\ref{chi_izvod}) and (\ref{rho_resenje}) we find:
\begin{align}
    \theta&+\frac{2\mathrm{F}^{-}}{\lambda}\left(\varphi\partial_{-}\theta+\theta\partial_{-}\varphi+\partial_{-}\rho\right)=\int\mathcal{D}\de\varphi+\frac{\mathrm{G}^{-}}{\mathrm{F}^{-}}\varphi\nonumber\\
    &\hspace{5mm}-\int\frac{\de x^{+}}{\mathrm{F}^{+}}\left(\frac{\lambda a}{\varphi}-1\right)\int\frac{\de x^{+}}{\mathrm{F}^{+}}\frac{t_{+}(\sigma^{+})}{\varphi-\lambda a}+\mathrm{H}^{-},
\end{align}
\noindent where we have introduced another arbitrary function $\mathrm{H}^{-}(x^{-})$. Eliminating the derivatives $\partial_{-}\varphi$ and $\partial_{-}\theta$ via equations (\ref{chi_izvodi}) and (\ref{theta_izvod_minus}) respectively, and applying the same procedure to equation (\ref{theta_izvod_minus}) we obtain the following equations:
\begin{widetext}
\begin{align}
    \theta&=\frac{\varphi(\varphi-\lambda a)}{\lambda a}\left(\mathcal{D}+2\alpha+\frac{\mathrm{G}^{+}}{\mathrm{F}^{+}}\right)+\frac{\varphi^{2}}{\lambda a}\frac{\mathrm{G}^{-}}{\mathrm{F}^{-}}+\frac{1}{2\varphi}\nonumber\\
    &+\frac{\varphi}{\lambda a}\left[\int\mathcal{D}\de\varphi-\frac{2}{\lambda}(\varphi-\lambda a)\int\frac{\de x^{-}}{\mathrm{F}^{-}}\frac{t_{-}(\sigma^{-})}{\varphi-\lambda a}-\int\frac{\de x^{+}}{\mathrm{F}^{+}}\left(\frac{\lambda a}{\varphi}-1\right)\int\frac{\de x^{+}}{\mathrm{F}^{+}}\frac{t_{+}(\sigma^{+})}{\varphi-\lambda a}-\mathrm{H}^{-}\right],\label{a1}\\
    \theta&=\frac{\varphi(\varphi-\lambda a)}{\lambda a}\left(\mathcal{D}+2\alpha+\frac{\mathrm{G}^{-}}{\mathrm{F}^{-}}\right)+\frac{\varphi^{2}}{\lambda a}\frac{\mathrm{G}^{+}}{\mathrm{F}^{+}}+\frac{1}{2\varphi}\nonumber\\
    &+\frac{\varphi}{\lambda a}\left[\int\mathcal{D}\de\varphi-\frac{2}{\lambda}(\varphi-\lambda a)\int\frac{\de x^{+}}{\mathrm{F}^{+}}\frac{t_{+}(\sigma^{+})}{\varphi-\lambda a}-\int\frac{\de x^{-}}{\mathrm{F}^{-}}\left(\frac{\lambda a}{\varphi}-1\right)\int\frac{\de x^{-}}{\mathrm{F}^{-}}\frac{t_{-}(\sigma^{-})}{\varphi-\lambda a}-\mathrm{H}^{+}\right]\label{a2}.
\end{align}
\end{widetext}
\noindent Equating expressions (\ref{a1}) and (\ref{a2}) we get the following equation:
\begin{align}
    &\mathrm{H}^{-}+\lambda a\frac{\mathrm{G^{-}}}{\mathrm{F}^{-}}-\frac{2}{\lambda}\int\frac{\de x^{-}}{\mathrm{F}^{-}}t_{-}(\sigma^{-})\nonumber\\
    &=\mathrm{H}^{+}+\lambda a\frac{\mathrm{G^{+}}}{\mathrm{F}^{+}}-\frac{2}{\lambda}\int\frac{\de x^{+}}{\mathrm{F}^{+}}t_{+}.(\sigma^{+})=C,\label{a3}
\end{align}
\noindent where $C$ is an arbitrary constant. Adding equations (\ref{a1}) and (\ref{a2}), and using equation (\ref{a3}) we get the following expression for $\theta$ in terms of $\varphi$ and $\alpha$:
\begin{widetext}
\begin{align}
    \theta=\frac{\varphi(\varphi-\lambda a)}{\lambda a}&\left[\mathcal{D}+2\alpha+\frac{\mathrm{G}^{+}}{\mathrm{F}^{+}}+\frac{\mathrm{G}^{-}}{\mathrm{F}^{-}}\right]+\frac{1}{2\varphi}+\frac{\varphi}{\lambda a}\int\mathcal{D}\de\varphi\nonumber\\
    -\frac{\varphi}{\lambda a}&\left[\int\frac{\de x^{+}}{\mathrm{F}^{+}}\left(\frac{\lambda a}{\varphi}-1\right)\int\frac{\de x^{+}}{\mathrm{F}^{+}}\frac{t_{+}(\sigma^{+})}{\varphi-\lambda a}+\int\frac{\de x^{-}}{\mathrm{F}^{-}}\left(\frac{\lambda a}{\varphi}-1\right)\int\frac{\de x^{-}}{\mathrm{F}^{-}}\frac{t_{-}(\sigma^{-})}{\varphi-\lambda a}\right].\label{theta_sol_1}
\end{align}
\end{widetext}
Once again, this is not the final expression for the field $\theta$, since it still depends on the unknown function $\alpha$. We need to somehow eliminate it. An easy way to do this is to use equation (\ref{theta_izvod_plus}) or (\ref{theta_izvod_minus}). In this way, we get another differential equation of first-order in $\theta$, which means that we need to perform another integration. Note that this constant $C$, defined through equation (\ref{a3}) does not appear in equation (\ref{theta_sol_1}); it is the integration constant of the integral $\int\mathcal{D}\de\varphi$. Later, we will see that it is a quantum correction to the constant $a$. We can arrive at the same result using either of the two equations (\ref{theta_izvod_minus}) or (\ref{theta_izvod_plus}). This will result in two expressions for $\theta$, with another set of arbitrary functions, which can be eliminated by equating these expressions. The final result is then given by:
\begin{widetext}
\begin{align}
    \theta=\frac{\varphi-\lambda a}{\varphi}&\left[\frac{\lambda}{2}\left(\int\de x^{+}\frac{\mathrm{G}^{+}}{\mathrm{F}^{+2}}+\int\de x^{-}\frac{\mathrm{G}^{-}}{\mathrm{F}^{-2}}\right)-\int d\varphi\left(\frac{\lambda a}{2\varphi(\varphi-\lambda a)^{2}}+\frac{\varphi}{(\varphi-\lambda a)^{2}}\int\mathcal{D}\de\varphi\right)\right.\nonumber\\
    &\hspace{4.9cm}\left.+\int\frac{\de x^{+}}{\mathrm{F}^{+}}\frac{1}{\varphi-\lambda a}\int\frac{\de x^{+}}{\mathrm{F}^{+}}t_{+}(\sigma^{+})+\int\frac{\de x^{-}}{\mathrm{F}^{-}}\frac{1}{\varphi-\lambda a}\int\frac{\de x^{-}}{\mathrm{F}^{-}}t_{-}(\sigma^{-})\right],\label{theta_sol}
\end{align}
\end{widetext}
\noindent in terms of $x^{+}$, $x^{-}$, $t_{+}(\sigma^{+}(x^{+}))$, $t_{-}(\sigma^{-}(x^{-}))$ and $\varphi(x^{+},x^{-})$. The next step would be to find $\alpha$. One way to do this would be by using equation (\ref{theta_izvod_plus}), it is easy to derive an expression for $\alpha$ in terms of $x^{+}$ and $x^{-}$. We will not present that calculation here. Rather, we will calculate it in a special case of the quantum-corrected collapse scenario.
\par Before evaluating the integrals in (\ref{theta_sol}), we simplify the expressions by introducing the reduced fields $\mathrm{x}=\frac{\varphi}{\lambda a}$ and $\hat{\mathrm{x}}=\frac{\hat{\varphi}}{\lambda a}$. In terms of $\mathrm{x}$, the function $\mathcal{D}(\mathrm{x})$ and its integral are given by:

\begin{align}
    &\mathcal{D}=\frac{1}{4(\lambda a)^{2}}\left[\ln{\left(1-\frac{1}{\mathrm{x}}\right)}-\frac{1}{\mathrm{x}-1}+\frac{3}{2}\frac{1}{\mathrm{x}^{2}}+\frac{2}{\mathrm{x}}\right],\label{D}\\
    &\int\mathcal{D}\de\varphi=\frac{1}{4\lambda a}\left[(\mathrm{x}-2)\ln{\left(1-\frac{1}{\mathrm{x}}\right)}-\frac{3}{2}\frac{1}{\mathrm{x}}\right]+C\label{D_int},
\end{align}
\noindent where $C$ is a constant defined in equation (\ref{a3}). For future simplifications, we introduce the following function:
\begin{equation}
    \theta_{0}(\varphi)=-\int\de\varphi\left(\frac{\lambda a}{2\varphi(\varphi-\lambda a)^{2}}+\frac{\varphi}{(\varphi-\lambda a)^{2}}\int\mathcal{D}\de\varphi\right),
\end{equation}
\noindent which, by utilizing equation (\ref{D_int}), becomes:
\begin{widetext}
\begin{equation}
    \theta_{0}(\mathrm{x})=-\frac{1}{4\lambda a}\frac{1}{\mathrm{x}-1}\left[\frac{1}{2}+(\mathrm{x}^{2}-3\mathrm{x}+3)\ln{(\mathrm{x}-1)}-(\mathrm{x}^{2}-2\mathrm{x}+2)\ln{\mathrm{x}}\right]+C\left(\frac{1}{\mathrm{x}-1}-\ln{\left(\mathrm{x}-1\right)}\right).\label{theta0_def}
\end{equation}
\end{widetext}
\noindent We also define integrals: 
\begin{equation}
    \mathrm{I}_{\pm}=\int\frac{\de x^{\pm}}{\mathrm{F}^{\pm}}\frac{1}{\varphi-\lambda a}\int\frac{\de x^{\pm}}{\mathrm{F}^{\pm}}t_{\pm}(\sigma^{\pm}).\label{I_pm}
\end{equation}
\noindent With the help of all this newly defined notation, equation (\ref{theta_sol}) can be written in a much simpler form:
\begin{equation}
    \theta=\frac{\mathrm{x}-1}{\mathrm{x}}\left[\frac{\lambda}{2}\left(\int\de x^{+}\frac{\mathrm{G}^{+}}{\mathrm{F}^{+2}}+\int\de x^{-}\frac{\mathrm{G}^{-}}{\mathrm{F}^{-2}}\right)+\theta_{0}+\mathrm{I}\right],\label{theta_sol_simp}
\end{equation}
\noindent where $\mathrm{I}=\mathrm{I}_{+}+\mathrm{I}_{-}$. Let us now show that $\mathrm{G}^{\pm}$ are the quantum corrections of the coordinate transformations, by showing how equation (\ref{chi_resenje}) changes when the quantum corrections are added. From equation (\ref{theta_sol_simp}) we extract the part connected to the functions $\mathrm{G}^{\pm}$, multiply it by $\varepsilon$ and add it to the right hand side of equation (\ref{chi_resenje}):
\begin{align}
    &\varphi+\lambda a\ln{\left(\frac{\varphi}{\lambda a}-1\right)}+\varepsilon\frac{\varphi\theta}{\varphi-\lambda a}-\varepsilon\theta_{0}-\varepsilon\mathrm{I}=\nonumber\\
    &-\frac{\lambda}{2}\left[\int\frac{\de x^{+}}{\mathrm{F}^{+}}\left(1-\varepsilon\frac{\mathrm{G}^{+}}{\mathrm{F}^{+}}\right)+\int\frac{\de x^{-}}{\mathrm{F}^{-}}\left(1-\varepsilon\frac{\mathrm{G}^{-}}{\mathrm{F}^{-}}\right)\right].\label{a4}
\end{align}
Now, we revert to the quantum-corrected fields $\varphi+\varepsilon\theta\mapsto\varphi$. If we expand equation (\ref{chi_resenje}) to the first-order in $\varepsilon$, we find exactly the term linear in $\theta$ appearing in equation (\ref{a4}). On the right-hand side of equation (\ref{a4}) we have the expansion of $(\mathrm{F}^{\pm}+\varepsilon\mathrm{G}^{\pm})^{-1}$ to the first-order in $\varepsilon$, which means that $\mathrm{F}^{\pm}+\varepsilon\mathrm{G}^{\pm}\mapsto\mathrm{F}^{\pm}$ is the quantum-corrected coordinate transformation. Also, note that the constant $C$ appearing in equation (\ref{theta0_def}) is the quantum correction $\lambda a+\varepsilon C\mapsto\lambda a$. This is easy to check by expanding the term $(\lambda a+\varepsilon C)\ln{(\varphi/(\lambda a+\varepsilon C)-1)}$ to the first-order in $\varepsilon$. We are also interested in the effect of quantum corrections on equation (\ref{rho_resenje}). Let us use the newly defined quantum-corrected coordinate transformations: 
\begin{equation}
    \mathrm{F}^{+}\mathrm{F}^{-}e^{2\rho}\mapsto\mathrm{F}^{+}\mathrm{F}^{-}e^{2\rho}\left[1+\varepsilon\left(2\alpha+\frac{\mathrm{G}^{+}}{\mathrm{F}^{+}}+\frac{\mathrm{G}^{-}}{\mathrm{F}^{-}}\right)\right].\label{a5}
\end{equation}
\noindent The quantity in brackets of equation (\ref{a5}) appears in equation (\ref{theta_sol_1}), which means that we can extract it from that equation. Before we do that, we define another set of integrals:
\begin{equation}
    \mathrm{I}_{\pm}'=\int\frac{\de x^{\pm}}{\mathrm{F}^{\pm}}\left(\frac{\lambda a}{\varphi}-1\right)\int\frac{\de x^{\pm}}{\mathrm{F}^{\pm}}\frac{t_{\pm}(\sigma^{\pm})}{\varphi-\lambda a}.\label{I'_pm}
\end{equation}
\noindent Also, we define $\mathrm{I}'=\mathrm{I}_{+}'+\mathrm{I}_{-}'$, substitute it into equation (\ref{a5}) and use equation (\ref{rho_resenje}) to arrive at the following expression:
\begin{align}
    &\mathrm{F}^{+}\mathrm{F}^{-}e^{2\rho}\mapsto\frac{\lambda a}{\varphi}-1-\varepsilon\frac{\lambda a\theta}{\varphi^{2}}\nonumber\\
    &+\varepsilon\left[\frac{\varphi-\lambda a}{\varphi}\mathcal{D}+\frac{\lambda a}{2\varphi^{3}}+\frac{1}{\varphi}\int\mathcal{D}\de\varphi+\frac{1}{\varphi}\mathrm{I}'\right].\label{a6}
\end{align}
\noindent The first three terms represent the expansion of $\lambda a(\varphi+\varepsilon\theta)^{-1}-1$. Since $C$ is the integration constant of $\int\mathcal{D}\de\varphi$, the only term containing this constant is $\varepsilon C/\varphi$, which demonstrates that it is in fact the quantum correction of the constant $\lambda a$. Rewriting equations (\ref{a4}) and (\ref{a6}) in terms of the quantum-corrected reduced field $\mathrm{x}$ and the quantum-corrected constant $\lambda a$, we get the following final result:
\begin{widetext}
\begin{align}
    &\mathrm{x}+\ln{(\mathrm{x}-1)}+\frac{\varepsilon}{4(\lambda a)^{2}}\frac{1}{\mathrm{x}-1}\left[\frac{1}{2}+(\mathrm{x}^{2}-3\mathrm{x}+3)\ln{(\mathrm{x}-1)}-(\mathrm{x}^{2}-2\mathrm{x}+2)\ln{\mathrm{x}}\right]-\frac{\varepsilon}{\lambda a}\mathrm{I}=-\frac{1}{2a}\left(\int\frac{\de x^{+}}{\mathrm{F}^{+}}+\int\frac{\de x^{-}}{\mathrm{F}^{-}}\right),\label{chi_resenje_q_final}\\
    &\mathrm{F}^{+}\mathrm{F}^{-}e^{2\rho}=\frac{1}{\mathrm{x}}-1+\frac{\varepsilon}{4(\lambda a)^{2}}\left[\frac{2\mathrm{x}-3}{\mathrm{x}}\ln{\left(1-\frac{1}{\mathrm{x}}\right)}+\frac{1}{\mathrm{x}}-\frac{2}{\mathrm{x}^{2}}+\frac{1}{2}\frac{1}{\mathrm{x}^{3}}\right]+\frac{\varepsilon}{\lambda a}\frac{\mathrm{I}'}{\mathrm{x}}.\label{rho_resenje_q_final}
\end{align}
\end{widetext}
Equations (\ref{chi_resenje_q_final}) and (\ref{rho_resenje_q_final}) are the general solution to the vacuum equations (\ref{jna_pp})-(\ref{jna_chi2}). So, the solution of the quantum-corrected equations of motion is given by $(\mathrm{F}^{+},\mathrm{F}^{-},a,t_{+},t_{-})$. Keep in mind that $\mathrm{F}^{\pm}$ and $a$ are quantum-corrected quantities. A concrete solution can be written when the integrals $\mathrm{I}$ and $\mathrm{I}'$ are evaluated. These integrals are dependent on each other, connected through equations (\ref{theta_izvod_plus}) and (\ref{theta_izvod_minus}), which means that only one of them needs to be calculated. Since $t_{\pm}(x^{\pm})=0$ defines the vacuum we are choosing (see Appendix $A$), the integrals $\mathrm{I}$ and $\mathrm{I}'$ are closely related to the choice of the vacuum state.

\subsection{Solution without energy flux at the asymptotic infinity}

Here, we consider the solution that does not have energy flux in either of the two asymptotic infinities of space-time. Note that in our case it is different from \cite{BPP}. There exists an exact Minkowski vacuum, which is the reason we are calling this a solution without the energy flux at infinity. Since the Eddington-Finkelstein coordinates $\sigma^{\pm}$ are asymptotically flat, we define the vacuum state by demanding: $t_{\pm}(\sigma^{\pm})=0$ (see Appendix $A$). This corresponds to the Boulware vacuum state $\ket{B}$. Since the integrals $\mathrm{I}_{\pm}$ (equation (\ref{I_pm})) and $\mathrm{I}_{\pm}'$ (equation (\ref{I'_pm})) are proportional to $t_{\pm}(\sigma^{\pm})$ they vanish, that is, $\mathrm{I}=0$ and $\mathrm{I}'=0$.
\par The solution is given by equations (\ref{chi_resenje_q_final}) and (\ref{rho_resenje_q_final}), without the integrals $\mathrm{I}$ and $\mathrm{I}'$. Note that the limit $\lambda a\to0$ leads to the Minkowski vacuum, which is expected. This limit is apparently divergent, but one should not forget that the form of the $\lambda a$ constant is given by $\lambda a\mapsto\lambda a+\varepsilon C(\lambda a)$. The divergence can be canceled with an appropriate choice of the constant $C$ as a function of $\lambda a$. This choice is given by $C(\lambda a)=\frac{1}{4\lambda a}$.

\section{Evaporating black hole scenario}

The evaporating black hole scenario encapsulates the process of gravitational collapse into a black hole, and its subsequent evaporation through Hawking radiation. The classical gravitational collapse has been studied in section II. Without quantum back-reaction to the metric, the process of evaporation cannot be studied. We have already seen that the black hole does not start to evaporate in the classical case, since the horizon and the singularity do not intersect in finite time. At the beginning, the metric of space-time is that of Minkowski space-time. This is true even when the quantum corrections are included, since the Minkowski metric solves equations (\ref{ppmm_q})-(\ref{chi2_q}), as we have commented at the end of the previous section. This means that the quantum state is the Minkowski vacuum state. But because the black hole is created later on, it is actually the Unruh state $\ket{U}=\ket{\sigma,0}$, defined as $t_{\pm}(\sigma^{\pm})=0$. After the creation of the black hole, the asymptotically flat coordinates change to $\hat{\sigma}^{\pm}$, whose classical parts are defined by equations (\ref{koord_smena_p}) and (\ref{koord_smena_m}). In these new coordinates, the energy flux at future infinity is nonzero, since the vacuum state is not $\ket{\hat{\sigma},0}$, so the black hole starts to evaporate.
\begin{figure}[h]
    \begin{center}
    \includegraphics[width=7cm, height=7cm]{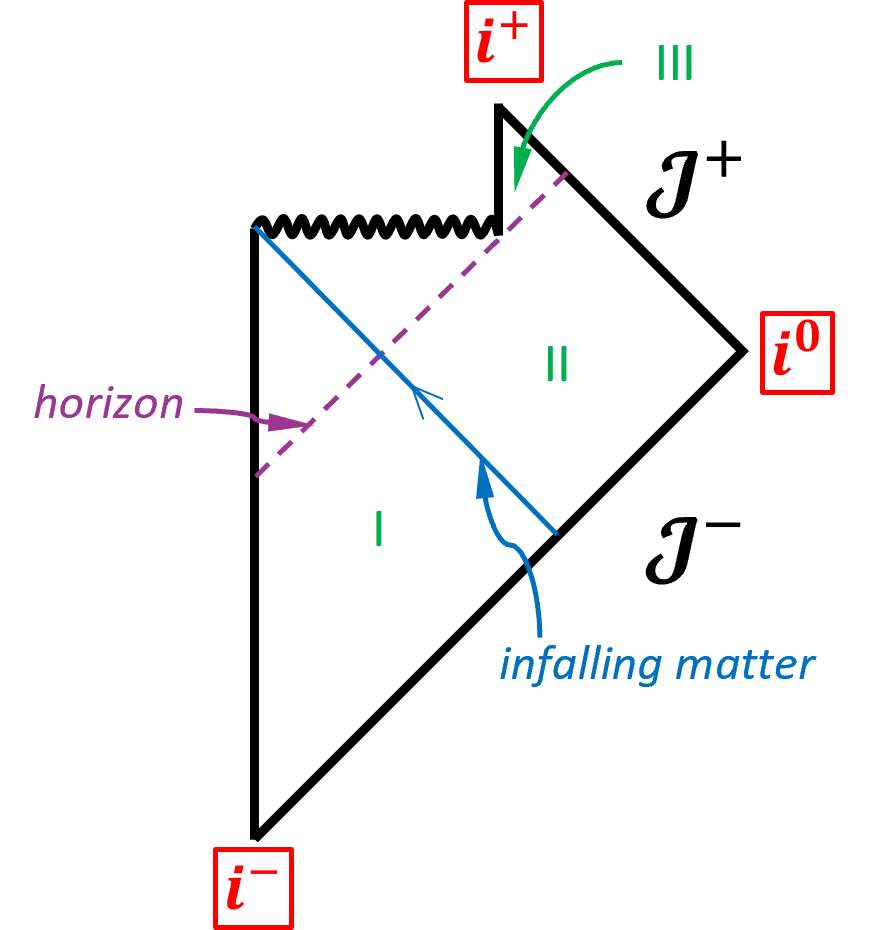}
    \caption{Penrose diagram of an evaporating black hole}\label{fig2}
    \end{center}
\end{figure}
\par Before the creation of the black hole (part I of the space-time shown in Figure \ref{fig2}), we choose to use the Eddington-Finkelstein gauge, which implies that coordinates $x^{\pm}$ in solutions (\ref{chi_resenje_q_final}) and (\ref{rho_resenje_q_final}) are actually $\sigma^{\pm}$ coordinates. The solution is then given by $(-1,1,a,0,0)$. In part II of space-time there exists an evaporating black hole. The Eddington-Finkelstein coordinates in this part of space-time are $\hat{\sigma}^{\pm}$, which implies that we need to calculate $t_{\pm}(\hat{\sigma}^{\pm})$. This can be easily done using the transformation law (\ref{t_pm_trans}) of the functions $t_{\pm}$, from the coordinates $\sigma^{\pm}$ to the coordinates $\hat{\sigma}^{\pm}$. The result is given in terms of the reduced field $\hat{\mathrm{x}}$:
\begin{equation}
    t_{+}(\hat{\sigma}^{+})=0\hspace{2mm}\text{and}\hspace{2mm}t_{-}(\hat{\sigma}^{-})=-\frac{\lambda^{2}}{4(\lambda a)^{2}}\frac{\hat{\mathrm{x}}-\frac{3}{4}}{\hat{\mathrm{x}}^{4}}.\label{t_pm_hat}
\end{equation}
\noindent Notice that space-time in Figure \ref{fig2} differs from that of Figure \ref{fig1} by the existence of a third region (III), which is space-time left after the black hole evaporates completely. The next step would be to calculate $\theta$ (from equation (\ref{theta_sol_simp})) in part II of space-time, since it holds all the necessary information regarding the metric and dilaton field. The solution is defined up to the integral $\mathrm{I}$ defined by equation (\ref{I_pm}). Since $t_{+}=0$, the integral $\mathrm{I}_{+}$ vanishes, and we are only left with the integral $\mathrm{I}_{-}$:
\begin{equation}
    \mathrm{I}=\int\frac{\de\sigma^{-}}{\mathrm{F}^{-}}\frac{1}{\varphi-\lambda a}\int\frac{\de\sigma^{-}}{\mathrm{F}^{-}}t_{-}(\hat{\sigma}^{-}).
\end{equation}
\noindent The inner integral is given in terms of reduced field $\hat{\mathrm{x}}(\sigma^{-})$ by the following expression:
\begin{align}
    \int\frac{\de\sigma^{-}}{\mathrm{F}^{-}}t_{-}(\hat{\sigma}^{-})&=\frac{1}{8a}\mathrm{F}(\hat{\mathrm{x}})\nonumber\\
    &=\frac{1}{8a}\left[\ln{\left(1-\frac{1}{\hat{\mathrm{x}}}\right)}+\frac{1}{\hat{\mathrm{x}}}-\frac{3}{2}\frac{1}{\hat{\mathrm{x}}^{2}}\right],\label{F_def}
\end{align}
\noindent where we have defined another function $\mathrm{F}$. Now, the integral $\mathrm{I}$ can be rewritten as:
\begin{equation}
    \mathrm{I}=-\frac{1}{4\lambda a}\int\frac{\hat{\mathrm{x}}\de\hat{\mathrm{x}}}{\hat{\mathrm{x}}-1}\frac{\mathrm{F}(\hat{\mathrm{x}})}{\mathrm{x}-1}.\label{I_sracunato}
\end{equation}
\subsection{Calculating the integral $\mathrm{I}$}
\noindent The main problem arising here is the fact that we have to calculate the integral with respect to $\hat{\mathrm{x}}$, but the integrand also depends on the field $\mathrm{x}$, and those two are connected to each other through the transcendental equation (\ref{resenje_II_}), which can be rewritten in the following form:
\begin{equation}
    e^{\hat{\mathrm{x}}}(\hat{\mathrm{x}}-1)=\delta e^{\mathrm{x}}(\mathrm{x}-1).\label{transcedentna}
\end{equation}
\noindent To proceed, we need to solve equation (\ref{transcedentna}) with respect to the variable $\hat{\mathrm{x}}=\hat{\mathrm{x}}(\mathrm{x},\delta)$. Using the fact that $\delta<1$ in part II of space-time, equation (\ref{transcedentna}) can be solved analytically. The detailed solution is given in the Appendix $B$ (theorems \ref{T1} and \ref{T4}). Here we give the final answer:
\begin{equation}
    \hat{\mathrm{x}}=\begin{cases}
                        1-\sum\limits_{n=1}^{\infty}\frac{\delta^{n}}{n!}e^{n(\mathrm{x}-1)}(1-\mathrm{x})^{n}\mathrm{P}_{n-1}(1),&\mathrm{x}\leqslant1\\
                        \mathrm{x}-\sum\limits_{n=1}^{\infty}\frac{(1-\delta)^{n}}{n!}\left(1-\frac{1}{\mathrm{x}}\right)^{n}\mathrm{P}_{n-1}\left(\frac{1}{\mathrm{x}}\right),&\mathrm{x}>1
                     \end{cases},\label{x_het_od_x_delta}
\end{equation}
\noindent where $\mathrm{P}_{n}(x)$ are the polynomials defined in theorem \ref{T1}. For this calculation, the function $\ln{\hat{\mathrm{x}}}=\ln{\hat{\mathrm{x}}}(\mathrm{x},\delta)$ will also be needed. It is given by (theorems \ref{T7} and \ref{T9}):
\begin{equation}
    \ln{\hat{\mathrm{x}}}=\begin{cases}
                        \hspace{6.5mm}-\sum\limits_{n=1}^{\infty}\frac{\delta^{n}}{n!}e^{n(\mathrm{x}-1)}(1-\mathrm{x})^{n}\mathrm{Q}_{n-1}(1),&\mathrm{x}\leqslant1\\
                        \ln{\mathrm{x}}-\sum\limits_{n=1}^{\infty}\frac{(1-\delta)^{n}}{n!}\left(1-\frac{1}{\mathrm{x}}\right)^{n}\mathrm{Q}_{n-1}\left(\frac{1}{\mathrm{x}}\right),&\mathrm{x}>1
                     \end{cases},\label{lnx_het_od_x_delta}
\end{equation}
\noindent where $\mathrm{Q}_{n}(x)$ are the polynomials defined in corollary \ref{T7.1}. Now, using the expansions (\ref{x_het_od_x_delta}) and (\ref{lnx_het_od_x_delta}) and the uniform convergence (see Appendix B),
it is easy to show that $\mathrm{I}$ takes the following form:
\begin{equation}
    \mathrm{I}(\mathrm{x},\delta)=\frac{1}{8\lambda a}\begin{cases}
        \mathrm{I}_{0}^{<}(\mathrm{x})+\mathrm{I}_{-1}^{<}(\mathrm{x})\ln{\delta}-\sum\limits_{n=1}^{\infty}\frac{\delta^{n}}{n!}\mathrm{I}_{n}^{<}(\mathrm{x})\\
        \mathrm{I}_{0}^{>}(\mathrm{x})+\mathrm{I}_{-1}^{>}(\mathrm{x})\ln{\delta}-\sum\limits_{n=1}^{\infty}\frac{(1-\delta)^{n}}{n!}\mathrm{I}_{n}^{>}(\mathrm{x})\label{I_resenje}
    \end{cases},
\end{equation}
\noindent where the newly defined functions $\mathrm{I}_{n}$ are given by:
\begin{align}
    \mathrm{I}^{<}_{0}(\mathrm{x})&=\frac{1}{\mathrm{x}-1}-\frac{\mathrm{x}-3}{\mathrm{x}-1}\ln{\left(1-\mathrm{x}\right)}-\ln^{2}{\left(1-\mathrm{x}\right)}\nonumber\\
    &-2\mathrm{x}+c_{0},\nonumber\\
    \mathrm{I}^{>}_{0}(\mathrm{x})&=\frac{1}{\mathrm{x}-1}-\frac{\mathrm{x}-3}{\mathrm{x}-1}\ln{\left(1-\frac{1}{\mathrm{x}}\right)}-\ln^{2}{\left(\mathrm{x}-1\right)}\nonumber\\
    &-2\mathrm{L}(\mathrm{x})+c_{0},\nonumber\\
    \mathrm{I}^{>}_{-1}(\mathrm{x})&=\mathrm{I}^{<}_{-1}(\mathrm{x})=2\left(\frac{1}{\mathrm{x}-1}-\ln{|\mathrm{x}-1|}+c_{-1}\right),\nonumber\\
    \mathrm{I}_{n}^{>}(\mathrm{x})&=\int\frac{(2\mathrm{x}-1)\de\mathrm{x}}{(\mathrm{x}-1)^{2}}\left(1-\frac{1}{\mathrm{x}}\right)^{n}\mathrm{P}_{n-1}\left(\frac{1}{\mathrm{x}}\right)\nonumber\\
    &+\int\frac{(2\mathrm{x}-3)\de\mathrm{x}}{(\mathrm{x}-1)^{2}}\left(1-\frac{1}{\mathrm{x}}\right)^{n}\mathrm{Q}_{n-1}\left(\frac{1}{\mathrm{x}}\right)\nonumber\\
    &-\left(1-\frac{1}{\mathrm{x}}\right)^{n}\frac{\mathrm{P}_{n-1}\left(\frac{1}{\mathrm{x}}\right)+3\mathrm{Q}_{n-1}\left(\frac{1}{\mathrm{x}}\right)}{\mathrm{x}-1}+c_{n},\nonumber\\
    \mathrm{I}_{n}^{<}(\mathrm{x})&=\mathrm{P}_{n-1}(1)\int\frac{(2\mathrm{x}-1)\de\mathrm{x}}{(\mathrm{x}-1)^{2}}e^{n(\mathrm{x}-1)}(1-\mathrm{x})^{n}\nonumber\\
    &+\mathrm{Q}_{n-1}(1)\int\frac{(2\mathrm{x}-3)\de\mathrm{x}}{(\mathrm{x}-1)^{2}}e^{n(\mathrm{x}-1)}(1-\mathrm{x})^{n}\nonumber\\
    &+e^{n(\mathrm{x}-1)}(1-\mathrm{x})^{n-1}(\mathrm{P}_{n-1}(1)+3\mathrm{Q}_{n-1}(1))+\Tilde{c}_{n},\label{I_clanovi}
\end{align}
\noindent where function $\mathrm{L}(\mathrm{x})=\frac{\pi^{2}}{6}-\int_{0}^{\mathrm{x}}\frac{ds}{s-1}\ln{s}$, so that $\mathrm{L}(1)=0$. The constants $c_{n}$ are arbitrary; comparing them with (\ref{chi_resenje_q_final}), we can interpret them as the choice of coordinate $\sigma^{+}$, since the part of $\mathrm{I}$ that depends on them is only a function of $\delta=\delta(\sigma^{+})$. The constants $\Tilde{c}_{n}$ are not arbitrary since we have to impose the condition of continuity on the functions $\mathrm{I}_{n}$, meaning: $\mathrm{I}_{n}^{>}(1)=\mathrm{I}_{n}^{<}(1)$.
\par The integral (\ref{I_sracunato}) has been calculated up to the simple integrals appearing in (\ref{I_clanovi}). This means that we can write an expression for the function $\theta(\mathrm{x},\delta)$ from equation (\ref{theta_sol_simp}):
\begin{widetext}
\begin{equation}
    \theta(\mathrm{x},\delta)=\frac{\mathrm{x}-1}{\mathrm{x}}\left[\frac{\lambda}{2}\left(\int\de\sigma^{+}\frac{\mathrm{G}^{+}}{\mathrm{F}^{+2}}+\int\de\sigma^{-}\frac{\mathrm{G}^{-}}{\mathrm{F}^{-2}}\right)+\frac{1}{8\lambda a}\left(\mathrm{S}_{0}(\mathrm{x})+\mathrm{S}_{-1}(\mathrm{x})\ln{\delta}-\sum_{n=1}^{\infty}\frac{\Tilde{\delta}^{n}}{n!}\mathrm{S}_{n}(\mathrm{x})\right)\right],\label{theta_resenje}
\end{equation}
\end{widetext}
\noindent where $\mathrm{S}_{0}(\mathrm{x})=\mathrm{I}_{0}(\mathrm{x})+(8\lambda a)\theta_{0}(\mathrm{x})$, $\mathrm{S}_{-1}=\mathrm{I}_{-1}$ and $\mathrm{S}_{n}=\mathrm{I}_{n}$. The function $\theta_{0}(\mathrm{x})$ is defined in equation (\ref{theta0_def}). Note that we have given both cases $\mathrm{x}<1$ and $\mathrm{x}>1$ in equation (\ref{theta_resenje}), using $\Tilde{\delta}=\delta$ for the case of $x<1$ and $\Tilde{\delta}=1-\delta$ for the case of $x>1$. $\mathrm{S}_{0}(\mathrm{x})$ is given by equation (\ref{S0_definicija}). Other functions $\mathrm{S}_{n}(\mathrm{x})$ are the same as the corresponding functions $\mathrm{I}_{n}(\mathrm{x})$, so they will not be given here again. It is simple to check that the functions $\mathrm{S}_{n}(\mathrm{x})$ satisfy a recurrence relation:
\begin{equation}
    n\frac{\de\mathrm{S}_{n}}{\de\mathrm{x}}-\frac{1}{\mathrm{x}-1}\frac{\de}{\de\mathrm{x}}\left(\frac{(\mathrm{x}-1)^{2}}{\mathrm{x}}\frac{\de\mathrm{S}_{n}}{\de\mathrm{x}}\right)=\begin{cases}\frac{\de\mathrm{S}_{n+1}}{\de\mathrm{x}} &\mathrm{x}>1\\
    \hspace{3mm}0&\mathrm{x}<1\end{cases}.\label{S_n_rekurentna}
\end{equation}
\begin{widetext}
\begin{equation}
    \mathrm{S}_{0}(\mathrm{x})=\begin{cases}
        2\frac{\mathrm{x^{2}}-2\mathrm{x}+2}{\mathrm{x}-1}\ln{\mathrm{x}}-(2\mathrm{x}-3)\ln{(1-\mathrm{x})}-\ln^{2}(1-\mathrm{x})-2\mathrm{x}+8\lambda aC\left(\frac{1}{\mathrm{x}-1}-\ln{(1-\mathrm{x})}\right)+c_{0}&;\mathrm{x}\leqslant1\\
        (2\mathrm{x}-1)\ln{\mathrm{x}}-(2\mathrm{x}-3)\ln{(\mathrm{x}-1)}-\ln^{2}(\mathrm{x}-1)-2\mathrm{L}(\mathrm{x})+8\lambda aC\left(\frac{1}{\mathrm{x}-1}-\ln{(\mathrm{x}-1)}\right)+c_{0}&;\mathrm{x}\geqslant1\\
    \end{cases}.\label{S0_definicija}
\end{equation}
\end{widetext}
Using equation (\ref{theta_izvod_plus}), $\alpha(\mathrm{x},\delta)$ can be calculated. Substituting the solution for $\theta$ and $\alpha$ in the equations of motion (\ref{jna_pp})-(\ref{jna_chi2}) we see that (\ref{theta_resenje}) is indeed the solution.
\subsection{Evaporating black hole solution}
\par In the previous sections we found a solution in part II of space-time (see Figure \ref{fig2}) where the black hole has already been created. It depends on many constants: $c_{0}$, $c_{-1}$, $c_{n}$, $\Tilde{c}_{n}$ and $C$. Since we can choose constants $c_{n}$ and $\Tilde{c}_{n}$ as we want (they correspond to the choice of the function $\mathrm{G}^{+}$) we define them such that $\mathrm{S}_{n}(1)=0$, and $c_{0}=0$. With the help of equation (\ref{theta_izvod_plus}) it is easy to derive $\alpha(\mathrm{x},\delta)$. Then, applying the same procedure as before: $\rho+\varepsilon\alpha\mapsto\rho$, $\mathrm{x}+\frac{\varepsilon}{\lambda a}\theta\mapsto\mathrm{x}$ and $\mathrm{F}^{\pm}+\varepsilon\mathrm{G}^{\pm}\mapsto\mathrm{F}^{\pm}$ (see equations (\ref{a4}) and (\ref{a5})) we arrive at the following equations:
\begin{widetext}
\begin{align}
    \mathrm{x}+\ln{|\mathrm{x}-1|}&-\frac{\varepsilon}{8(\lambda a)^{2}}\left[\mathrm{S}_{0}(\mathrm{x})+\mathrm{S}_{-1}(\mathrm{x})\ln{\delta}-\sum\frac{\Tilde{\delta}^{n}}{n!}\mathrm{S}_{n}(\mathrm{x})\right]=-\frac{\lambda}{2}\frac{1}{\lambda a}\left[\int\frac{\de\sigma^{+}}{\mathrm{F}^{+}}+\int\frac{\de\sigma^{-}}{\mathrm{F}^{-}}\right]\equiv\frac{\hat{\sigma}^{+}-\hat{\sigma}^{-}}{2a},\label{xx_jednacina}\\
    \mathrm{F}^{+}\mathrm{F}^{-}e^{2\rho}&=\left(\frac{1}{\mathrm{x}}-1\right)\left\{1+\frac{\varepsilon}{8(\lambda a)^{2}}\left[\frac{\mathrm{x}-1}{\mathrm{x}}\frac{\de \mathrm{S}_{0}(\mathrm{x})}{\de\mathrm{x}}-2c_{-1}-c_{1}+\frac{\mathrm{x}-1}{\mathrm{x}}\frac{\de \mathrm{S}_{-1}(\mathrm{x})}{\de\mathrm{x}}\ln{\delta}\right.\right.\nonumber\\
    &\hspace{4cm}\left.\left.-\sum_{n=1}^{\infty}\frac{(1-\delta)^{n}}{n!}\left(\frac{\mathrm{x}-1}{\mathrm{x}}\frac{\de \mathrm{S}_{n}(\mathrm{x})}{\de\mathrm{x}}+\mathrm{S}_{n+1}(\mathrm{x})-n\mathrm{S}_{n}(\mathrm{x})\right)\right]\right\},\hspace{2mm}\text{when }x\geqslant1,\nonumber\\
    \mathrm{F}^{+}\mathrm{F}^{-}e^{2\rho}&=\left(\frac{1}{\mathrm{x}}-1\right)\left\{1+\frac{\varepsilon}{8(\lambda a)^{2}}\left[\frac{\mathrm{x}-1}{\mathrm{x}}\frac{\de \mathrm{S}_{0}(\mathrm{x})}{\de\mathrm{x}}-\mathrm{S}_{-1}(\mathrm{x})-8(\lambda a)^{2}\mathcal{D}(\mathrm{x})+\frac{\mathrm{x}-1}{\mathrm{x}}\frac{\de \mathrm{S}_{-1}(\mathrm{x})}{\de\mathrm{x}}\ln{\delta}\right.\right.\nonumber\\
    &\hspace{4cm}\left.\left.-\sum_{n=1}^{\infty}\frac{\delta^{n}}{n!}\left(\frac{\mathrm{x}-1}{\mathrm{x}}\frac{\de \mathrm{S}_{n}(\mathrm{x})}{\de\mathrm{x}}-n\mathrm{S}_{n}(\mathrm{x})\right)\right]\right\},\hspace{2mm}\text{when }x\leqslant1.\label{rhoo_jednacina}
\end{align}
\end{widetext}
\par The only remaining arbitrary constants are $c_{-1}$ and $C$. The solution in part II of space-time needs to be continuously connected to the solution in part I of space-time over the hypersurface $\sigma^{+}=\sigma^{+}_{0}$, or $\delta=1$ (Figure \ref{fig1}). First, let us examine the "++" equation from (\ref{ppmm_q}):
\begin{align}
    \partial_{+}(e^{-2\rho}\partial_{+}\varphi)&=-\frac{\lambda}{2}\lambda a\delta(\sigma^{+}-\sigma^{+}_{0})\varphi^{-1}e^{-2\rho}\nonumber\\
    &+\varepsilon((\partial_{+}\rho)^{2}-\partial_{+}^{2}\rho)\varphi^{-1}e^{-2\rho}.\label{qq1}
\end{align}
Integrating this equation, we get:
\begin{align}
    e^{-2\rho}\partial_{+}\varphi_{>}-e^{-2\rho}\partial_{+}\varphi_{<}&=-\frac{\lambda}{2}\lambda a\varphi^{-1}e^{-2\rho}\label{qq2}\\
    &-\varepsilon(\partial_{+}\rho_{>}-\partial_{+}\rho_{<})\varphi^{-1}e^{-2\rho},\nonumber
\end{align}
where "$>$" stands for the region where $\sigma^{+}>\sigma^{+}_{0}$ and "$<$" stands for the region where $\sigma^{+}<\sigma^{+}_{0}$. The last term in equation (\ref{qq2}) warrants explanation; it comes from the last term of equation (\ref{qq1}). Since $\partial_{+}\rho=\frac{\lambda}{4}\frac{\lambda a}{\varphi^{2}}\eta(\sigma^{+}-\sigma^{+}_{0})+o(\varepsilon)$, where $\eta(\sigma^{+}-\sigma^{+}_{0})$ represents the step function, we can deduce that the second derivative $\partial^{2}_{+}\rho$ has a term that behaves as a Dirac $\delta$-function. From equation (\ref{theta_izvod_plus}) one can show that:
\begin{equation}
    \partial_{+}\varphi_{>}=\frac{\lambda}{2}\mathrm{F}^{-}_{>}e^{2\rho}\left(1+\varepsilon\mathcal{D}(\mathrm{x})\right).\label{qq3}
\end{equation}
\noindent Since when $\delta=1$, we have $\mathrm{x}=\hat{\mathrm{x}}$, $e^{2\rho}=1$, and $\partial_{+}\varphi_{<}=\lambda/2$. After a brief calculation, we arrive at the following expression for the function $\mathrm{F}^{-}(\hat{\mathrm{x}})$:
\begin{equation}
    \mathrm{F}^{-}=\left(1-\frac{1}{\hat{\mathrm{x}}}\right)\left[1-\varepsilon\mathcal{D}(\hat{\mathrm{x}})-\frac{\varepsilon}{2(\lambda a)^{2}}\frac{1}{\hat{\mathrm{x}}^{2}(\hat{\mathrm{x}}-1)}\right].\label{F-_11}
\end{equation}
\par We may now impose the continuity condition on $\rho$ and $\theta$. This will be done in two separate cases $\mathrm{x}>1$ and $\mathrm{x}<1$. First, we will consider the case where $\mathrm{x}>1$. Inserting $\delta=1$ into equation (\ref{rhoo_jednacina}) and using equation (\ref{F-_11}) we arrive at:
\begin{equation}
    \frac{8\lambda aC-2}{\hat{\mathrm{x}}-1}+2c_{-1}+c_{1}=0.\label{qq4}
\end{equation}
Equation (\ref{qq4}) implies that $C=\frac{1}{4\lambda a}$ and $c_{-1}=-c_{1}/2$. Since $c_{1}=7$ has already been determined from $\mathrm{S}_{1}(1)=0$, where $S_{1}(\mathrm{x})=2\ln{\mathrm{x}}-\frac{4}{\mathrm{x}}-\frac{3}{\mathrm{x}^{2}}+c_{1}$, it follows $c_{-1}=-\frac{7}{2}$. The other continuity condition is $\theta(\hat{\mathrm{x}},1)=0$. Using equation (\ref{xx_jednacina}), this condition becomes: 
\begin{equation}
    \frac{\lambda}{2}\int\de\sigma^{+}\frac{\mathrm{G}^{+}}{\mathrm{F}^{+2}}\bigg{|}_{\delta=1}+\frac{\lambda}{2}\int\de\sigma^{-}\frac{\mathrm{G}^{-}}{\mathrm{F}^{-2}}=-\frac{\mathrm{S}_{0}(\hat{\mathrm{x}})}{8\lambda a}.
\end{equation}
For a completely determined solution, we choose $\mathrm{G}^{+}=0$. This gives the following coordinate transformations:
\begin{align}
    \mathrm{F}^{-}&=\left(1-\frac{1}{\hat{\mathrm{x}}}\right)\left[1+\frac{\varepsilon}{8(\lambda a)^{2}}\frac{\hat{\mathrm{x}}-1}{\hat{\mathrm{x}}}\frac{\de\mathrm{S}_{0}(\hat{\mathrm{x}})}{\de\hat{\mathrm{x}}}\right]\label{F-_resenje},\\
    \mathrm{F}^{+}&=-1\label{F+_resenje},\\
    \hat{\sigma}^{-}&=\sigma^{+}_{0}-2a\left[\hat{\mathrm{x}}+\ln{|\hat{\mathrm{x}}-1|}-\frac{\varepsilon}{8(\lambda a)^{2}}\mathrm{S}_{0}(\hat{\mathrm{x}})\right]\label{sigma-_hat_resenje},\\
    \hat{\sigma}^{+}&=\sigma^{+}.\label{sigma+_hat_resenje}
\end{align}
\par Next, we consider the $x<1$ case. When $\delta=1$ is inserted into equation (\ref{rhoo_jednacina}), it reduces to the following condition:
\begin{align}
    2\ln{\hat{\mathrm{x}}}&-2\frac{\ln{\hat{\mathrm{x}}}}{\hat{\mathrm{x}}-1}-\frac{3}{\hat{\mathrm{x}}}-2-2c_{-1}+2\frac{1-4\lambda aC}{\hat{\mathrm{x}}-1}\nonumber\\
    &=\sum_{n=1}^{\infty}\frac{1}{n!}\left[\frac{\hat{\mathrm{x}}-1}{\hat{\mathrm{x}}}\frac{\de\mathrm{S}_{n}(\hat{\mathrm{x}})}{\de\hat{\mathrm{x}}}-n\mathrm{S}_{n}(\hat{\mathrm{x}})\right].\label{qqq1}
\end{align}
\noindent Inserting $\delta=1$ into equations (\ref{x_het_od_x_delta}) and (\ref{lnx_het_od_x_delta}), we find:
\begin{align}
    \hat{\mathrm{x}}&=1-\sum_{n=1}^{\infty}\frac{\mathrm{P}_{n-1}(1)}{n!}e^{n(\hat{\mathrm{x}}-1)}(1-\hat{\mathrm{x}})^{n},\label{qqq2}\\
    \ln{\hat{\mathrm{x}}}&=-\sum_{n=1}^{\infty}\frac{\mathrm{Q}_{n-1}(1)}{n!}e^{n(\hat{\mathrm{x}}-1)}(1-\hat{\mathrm{x}})^{n}.\label{qqq3}
\end{align}
\noindent Using equation (\ref{I_clanovi}) we obtain the following derivatives of $\mathrm{S}_{n}(\mathrm{x})$:
\begin{align}
    \frac{\de\mathrm{S}_{n}^{<}(\hat{\mathrm{x}})}{\de\hat{\mathrm{x}}}&=\hat{\mathrm{x}}(1-\hat{\mathrm{x}})^{n-2}e^{n(\hat{\mathrm{x}}-1)}\nonumber\\
    &\times\left[(2-n)\mathrm{P}_{n-1}(1)+(2-3n)\mathrm{Q}_{n-1}(1)\right].\label{qqq4}
\end{align}
\noindent Changing equation (\ref{qqq4}) into the first sum in equation (\ref{qqq1}) and using (\ref{qqq2}) and (\ref{qqq3}) it is easy to calculate the corresponding sum. We get:
\begin{equation}
    \frac{\hat{\mathrm{x}}-1}{\hat{\mathrm{x}}}\sum_{n=0}^{\infty}\frac{1}{n!}\frac{\de\mathrm{S}_{n}^{<}(\hat{\mathrm{x}})}{\de\hat{\mathrm{x}}}=-2-2\frac{\ln{\hat{\mathrm{x}}}}{\hat{\mathrm{x}}-1}+\frac{1}{\hat{\mathrm{x}}}+\frac{3}{\hat{\mathrm{x}}^{2}}.\label{sum_1}
\end{equation}
The second sum appearing in equation (\ref{qqq1}) can be easily calculated after taking the derivative, using the same method as in the case of the previous sum, and integrating the result with the initial condition that the sum is equal to zero when $\hat{\mathrm{x}}=1$. The final result is given by:
\begin{equation}
    -\sum_{n=1}^{\infty}\frac{\mathrm{S}_{n}^{<}(\hat{\mathrm{x}})}{(n-1)!}=2\ln{\hat{\mathrm{x}}}-\frac{4}{\hat{\mathrm{x}}}-\frac{3}{\hat{\mathrm{x}}^{2}}+7\equiv\mathrm{S}_{1}^{>}(\hat{\mathrm{x}}).\label{sum_2}
\end{equation}
Substituting the sums (\ref{sum_1}) and (\ref{sum_2}) into equation (\ref{qqq1}) returns condition (\ref{qq4}), as expected. The last condition we need to impose is the continuity of the function $\theta$. Using equation (\ref{xx_jednacina}), this continuity condition may be recast in the following form:
\begin{equation}
    \mathrm{S}_{0}^{>}(\hat{\mathrm{x}})=\mathrm{S}_{0}^{<}(\hat{\mathrm{x}})-\sum_{n=1}^{\infty}\frac{\mathrm{S}_{n}^{<}(\hat{\mathrm{x}})}{n!}.\label{qqq5}
\end{equation}
Again, with the help of equations (\ref{qqq2}) and (\ref{qqq3}) this sum can be calculated. The result is given by:
\begin{equation}
    \sum_{n=1}^{\infty}\frac{\mathrm{S}^{<}_{n}(\hat{\mathrm{x}})}{n!}=2\mathrm{L}(\hat{\mathrm{x}})-2\hat{\mathrm{x}}-\frac{\hat{\mathrm{x}}-3}{\hat{\mathrm{x}}-1}\ln{\hat{\mathrm{x}}}.\label{sum_3}
\end{equation}
Substituting equation (\ref{sum_3}) into equation (\ref{qqq5}), one can show that the continuity condition is met. 
\subsection{Asymptotically flat solution}
Our final task is to check if the solution given by equations (\ref{xx_jednacina}) and (\ref{rhoo_jednacina}) is asymptotically flat. To do this, it is instructive to better understand the form of $\mathrm{S}^{>}_{n}(\mathrm{x})$ as $\mathrm{x}\to\infty$. We show that the following equation holds when $n>1$:
\begin{equation}
    \mathrm{S}^{>}_{n}(\mathrm{x})=(n-1)!\mathrm{S}^{>}_{1}(\mathrm{x})+\frac{6}{\mathrm{x}^{4}}\mathrm{Z}_{2(n-2)}\left(\frac{1}{\mathrm{x}}\right)-6\mathrm{Z}_{2(n-2)}(1),\label{S_n_uslov}
\end{equation}
where $\mathrm{Z}_{2(n-2)}(x)$ are polynomials of degree $2(n-2)$. The expression (\ref{S_n_uslov}) is written with the condition $\mathrm{S}_{n}(1)=0$ in mind, since $\mathrm{S}_{1}(1)=0$. The last two terms in equation (\ref{S_n_uslov}) will be denoted by $\mathrm{R}_{n}(\mathrm{x})$. Using the recurrence relation (\ref{S_n_rekurentna}), the following recurrence relation for the functions $\mathrm{R}_{n}$ can be derived:
\begin{equation}
    \frac{\de\mathrm{R}_{n+1}}{\de\mathrm{x}}-n\frac{\de\mathrm{R}_{n}}{\de\mathrm{x}}=-(n-1)!\frac{24}{\mathrm{x}^{5}}-\frac{1}{\mathrm{x}-1}\frac{\de}{\de\mathrm{x}}\left(\frac{(\mathrm{x}-1)^{2}}{\mathrm{x}}\frac{\de\mathrm{R}_{n}}{\de\mathrm{x}}\right).\label{zzz1}
\end{equation}
Now we can substitute $\frac{\de\mathrm{R}_{n}}{\de\mathrm{x}}=-\frac{24} {\mathrm{x}^{5}}\mathrm{M}_{n}(\mathrm{x})$ into equation (\ref{zzz1}), which yields:
\begin{equation}
    \mathrm{M}_{n+1}-n\mathrm{M}_{n}=(n-1)!+2\left(\frac{2}{\mathrm{x}}-\frac{3}{\mathrm{x}^{2}}\right)\mathrm{M}_{n}-\frac{\mathrm{x}-1}{\mathrm{x}}\frac{\de\mathrm{M}_{n}}{\de\mathrm{x}}.\label{zzz2}
\end{equation}
Since $\mathrm{M}_{2}(\mathrm{x})=1$, equation (\ref{zzz2}) implies that the functions $\mathrm{M}_{n}(\mathrm{x})$ are polynomials of $1/\mathrm{x}$ and that their degree increases by two when $n$ increases by one. This means that the degree of $n$-th polynomial is $2(n-2)$. After integration of the functions $\mathrm{M}_{n}/\mathrm{x}^{5}$, we end up with the form (\ref{S_n_uslov}) for the functions $\mathrm{S}^{>}_{n}(\mathrm{x})$. To simplify future expressions, we will denote $6\mathrm{Z}_{2(n-2)}(1)=z_{n}$. To see if space-time is asymptotically flat, let us check the behavior of the metric when $\mathrm{x}\to\infty$. This will be done term by term using equation (\ref{rhoo_jednacina}). The term that does not depend on $\delta$ is given by:
\begin{equation}
    \left(\frac{\mathrm{x}-1}{\mathrm{x}}\right)^{2}\frac{\de\mathrm{S}^{>}_{0}}{\de\mathrm{x}}=-2\frac{\mathrm{x}-1}{\mathrm{x}}\ln{\left(1-\frac{1}{\mathrm{x}}\right)}-\frac{2}{\mathrm{x}}+\frac{1}{\mathrm{x}^{2}}-\frac{1}{\mathrm{x}^{3}},\label{o1}
\end{equation}
which implies that this term tends to zero when $\mathrm{x}\to\infty$. The $\ln{\delta}/\mathrm{x}$ term also tends to zero when $\mathrm{x}\to\infty$. Since $\frac{\de\mathrm{S}^{>}_{n}}{\de\mathrm{x}}\sim(n-1)!\frac{2}{\mathrm{x}}$ when $\mathrm{x}\to\infty$, this term also vanishes. The last remaining term in the sum is: 
\begin{equation}
    \mathrm{S}^{>}_{n+1}-n\mathrm{S}^{>}_{n}=-(z_{n+1}-nz_{n})+\mathcal{O}\left(\frac{1}{\mathrm{x}^{4}}\right).
\end{equation}
This implies that we are left with the following expression:
\begin{equation}
    \lim_{\mathrm{x}\to\infty}\mathrm{F}^{+}\mathrm{F}^{-}e^{2\rho}=-1-\frac{\varepsilon}{8(\lambda a)^{2}}\sum_{n=1}^{\infty}\frac{(1-\delta)^{n}}{n!}(z_{n+1}-nz_{n}).\label{rho_beskonacno}
\end{equation}
From equation (\ref{rho_beskonacno}) we conclude that in $\hat{\sigma}^{\pm}$ coordinates, defined in (\ref{sigma-_hat_resenje}) and (\ref{sigma+_hat_resenje}), the metric is not asymptotically flat. Since the metric at infinity depends only on $\sigma^{+}$, we can try to find a coordinate transformation that makes the metric flat at infinity: $\de s^{2}=e^{2\rho}\de\sigma^{+}\de\sigma^{-}=-\mathrm{F}^{+}\mathrm{F}^{-}e^{2\rho}\de\hat{\sigma}^{+}\de\hat{\sigma}^{-}=-\mathrm{F}^{+}\mathrm{F}'^{+}\mathrm{F}^{-}e^{2\rho}\de\hat{\sigma}'^{+}\de\hat{\sigma}^{-}$. Demanding that the metric be flat at infinity in the coordinates $\hat{\sigma}^{-}$ and $\hat{\sigma}'^{+}$ boils down to choosing:
\begin{equation}
    \mathrm{F}'^{+}=1-\frac{\varepsilon}{8(\lambda a)^{2}}\sum_{n=1}^{\infty}\frac{(1-\delta)^{n}}{n!}(z_{n+1}-nz_{n}).\label{F'+}
\end{equation}
To obtain the coordinate $\hat{\sigma}'^{+}$, the integral $\int\frac{\de\hat{\sigma}^{+}}{\mathrm{\mathrm{F}'^{+}}}$ needs to be calculated. The result is given by:
\begin{equation}
    \hat{\sigma}'^{+}=-2a\ln{\delta}+\frac{a\varepsilon}{4(\lambda a)^{2}}\sum_{n=1}^{\infty}\frac{(1-\delta)^{n}}{n!}z_{n}+2aD,\label{sigma_hat_+'}
\end{equation}
where $D$ is an undetermined constant. Next, we check what happens with equation (\ref{xx_jednacina}) for the dilaton field, at infinity. Expressing $\ln{\delta}$ from equation (\ref{sigma_hat_+'}) and substituting it into equation (\ref{xx_jednacina}) yields the following:
\begin{align}
    \mathrm{x}+&\ln{(\mathrm{x}-1)}-\frac{\varepsilon}{8(\lambda a)^{2}}\bigg{[}\mathrm{S}_{0}^{>}(\mathrm{x})+\mathrm{S}_{-1}(\mathrm{x})\ln{\delta}\label{o2}\\
    &-\sum_{n=1}^{\infty}\frac{(1-\delta)^{n}}{n!}(\mathrm{S}_{n}^{>}(\mathrm{x})+z_{n})\bigg{]}=\frac{\hat{\sigma}'^{+}-\hat{\sigma}^{-}}{2a}-D.\nonumber
\end{align}
First we will work out the $\delta$ dependent terms in equation (\ref{o2}). At infinity $\mathrm{S}_{n}(\mathrm{x})\sim(n-1)!(2\ln{\mathrm{x}}+7)-z_{n}$, allowing us to express the sum appearing in equation (\ref{o2}) as:
\begin{equation}
    -\sum_{n=1}^{\infty}\frac{(1-\delta)^{n}}{n!}(\mathrm{S}^{>}_{n}(\mathrm{x})+z_{n})\sim\ln{\delta}\bigg{(}2\ln{\mathrm{x}}+7\bigg{)}.\label{o3}
\end{equation}
This term (\ref{o3}) cancels exactly with the term $\mathrm{S}_{-1}(\mathrm{x})\ln{\delta}$ when $\mathrm{x}\to\infty$. The $\mathrm{S}^{>}_{0}(\mathrm{x})$ term in (\ref{o2}) behaves as:
\begin{equation}
    \lim_{\mathrm{x}\to\infty}\mathrm{S}^{>}_{0}(\mathrm{x})=2+\frac{\pi^{2}}{3},\label{o4}
\end{equation}
as $\mathrm{x}\to\infty$. The important step in the derivation of the formula (\ref{o4}) is that the behavior of $\mathrm{L}(\mathrm{x})$ at infinity is given by $\mathrm{L}(\mathrm{x})\sim-\frac{1}{2}\ln^{2}{(\mathrm{x}-1)}-\frac{\pi^{2}}{6}$. At infinity, equation (\ref{o2}) reduces to:
\begin{equation}
    \mathrm{x}+\ln{(\mathrm{x}-1)}-\frac{\varepsilon}{4(\lambda a)^{2}}\left(\frac{\pi^{2}}{6}+1\right)=\frac{\hat{\sigma}'^{+}-\hat{\sigma}^{-}}{2a}-D.
\end{equation}
Demanding that the asymptotic behavior of the dilaton field (at infinity) be unaffected by quantum corrections leads to the choice of the constant $D=\frac{\varepsilon}{4(\lambda a)^{2}}\left(\frac{\pi^{2}}{6}+1\right)$. For simplicity, we may redefine the coordinate transformation of $\sigma^{+}$ by choosing $\mathrm{F}^{+}\mathrm{F}'^{+}\mapsto\mathrm{F}^{+}$ and $\hat{\sigma}'^{+}\mapsto\hat{\sigma}^{+}$. 
In this new notation, the coordinate transformations are given by:
\begin{align}
    \mathrm{F}^{-}&=\left(1-\frac{1}{\hat{\mathrm{x}}}\right)\left[1+\frac{\varepsilon}{8(\lambda a)^{2}}\frac{\hat{\mathrm{x}}-1}{\hat{\mathrm{x}}}\frac{\de\mathrm{S}_{0}(\hat{\mathrm{x}})}{\de\hat{\mathrm{x}}}\right],\label{F-_resenje_final}\\
    \mathrm{F}^{+}&=-1+\frac{\varepsilon}{8(\lambda a)^{2}}\sum_{n=1}^{\infty}\frac{(1-\delta)^{n}}{n!}(z_{n+1}-nz_{n}).\label{F+_resenje_final}
\end{align}
The new asymptotically flat coordinates are given by:
\begin{align}
    \hat{\sigma}^{-}&=\sigma^{+}_{0}-2a\left[\hat{\mathrm{x}}+\ln{|\hat{\mathrm{x}}-1|}-\frac{\varepsilon}{8(\lambda a)^{2}}\mathrm{S}_{0}(\hat{\mathrm{x}})\right],\label{sigma-_hat_resenje_final}\\
    \hat{\sigma}^{+}&=\sigma^{+}+\frac{2a\varepsilon}{8(\lambda a)^{2}}\left[\frac{\pi^{2}}{3}+2+\sum_{n=1}^{\infty}\frac{(1-\delta)^{n}}{n!}z_{n}\right].\label{sigma+_hat_resenje_final}
\end{align}
\noindent We define a new $\hat{\delta}$ coordinate in terms of $\hat{\sigma}^{+}$: 
\begin{equation}
    \hat{\delta}=\exp{\left(-\frac{\hat{\sigma}^{+}-\sigma^{+}_{0}}{2a}\right)}.\label{delta_hat_def}
\end{equation}
Now we can write down the final expression of the metric and the dilaton field in part II of the space-time (see Figure \ref{fig1}),
\begin{widetext}
\begin{align}
    &\mathrm{F}^{+}\mathrm{F}^{-}e^{2\rho}=\left(\frac{1}{\mathrm{x}}-1\right)\left\{1+\frac{\varepsilon}{8(\lambda a)^{2}}\left[\frac{\mathrm{x}-1}{\mathrm{x}}\frac{\de\mathrm{S}^{>}_{0}(\mathrm{x})}{\de\mathrm{x}}-\frac{2\ln{\delta}}{\mathrm{x}-1}-\sum_{n=1}^{\infty}\frac{(1-\delta)^{n}}{n!}\left(\frac{\mathrm{x}-1}{\mathrm{x}}\frac{\de\mathrm{S}^{>}_{n}(\mathrm{x})}{\de\mathrm{x}}+\mathrm{S}^{>}_{n+1}(\mathrm{x})-n\mathrm{S}^{>}_{n}(\mathrm{x})\right)\right]\right\}\nonumber\\
    &\hspace{7cm}\times\left[1-\frac{\varepsilon}{8(\lambda a)^{2}}\sum_{n=1}^{\infty}\frac{(1-\delta)^{n}}{n!}(z_{n+1}-nz_{n})\right]\label{rho_resenje_>}\\
    &\mathrm{F}^{+}\mathrm{F}^{-}e^{2\rho}=\left(\frac{1}{\mathrm{x}}-1\right)\left\{1+\frac{\varepsilon}{8(\lambda a)^{2}}\left[\frac{\mathrm{x}-1}{\mathrm{x}}\frac{\de\mathrm{S}^{<}_{0}(\mathrm{x})}{\de\mathrm{x}}+\mathrm{S}^{>}_{1}(\mathrm{x})-\frac{2\ln{\delta}}{\mathrm{x}-1}-\sum_{n=1}^{\infty}\frac{\delta^{n}}{n!}\left(\frac{\mathrm{x}-1}{\mathrm{x}}\frac{\de\mathrm{S}^{<}_{n}(\mathrm{x})}{\de\mathrm{x}}-n\mathrm{S}^{<}_{n}(\mathrm{x})\right)\right]\right\}\nonumber\\
    &\hspace{7cm}\times\left[1-\frac{\varepsilon}{8(\lambda a)^{2}}\sum_{n=1}^{\infty}\frac{(1-\delta)^{n}}{n!}(z_{n+1}-nz_{n})\right]\label{rho_resenje_<}\\
    &\mathrm{x}+\ln{(\mathrm{x}-1)}-\frac{\varepsilon}{8(\lambda a)^{2}}\left[\mathrm{S}_{0}^{>}(\mathrm{x})-\frac{\pi^{2}}{3}-2+\ln{\delta}\mathrm{S}_{-1}(\mathrm{x})-\sum_{n=1}^{\infty}\frac{(1-\delta)^{n}}{n!}(\mathrm{S}_{n}^{>}(\mathrm{x})+z_{n})\right]\nonumber\\
    &\hspace{7cm}=-\ln{\hat{\delta}}+\hat{\mathrm{x}}+\ln{(\hat{\mathrm{x}}-1)}-\frac{\varepsilon}{8(\lambda a)^{2}}\mathrm{S}^{>}_{0}(\hat{\mathrm{x}})\equiv\frac{\hat{\sigma}^{+}-\hat{\sigma}^{-}}{2a}\equiv\frac{\hat{\sigma}}{a},\label{x_resenje_>}\\
    &\mathrm{x}+\ln{(1-\mathrm{x})}-\frac{\varepsilon}{8(\lambda a)^{2}}\left[\mathrm{S}_{0}^{<}(\mathrm{x})-\frac{\pi^{2}}{3}-2+\ln{\delta}\mathrm{S}_{-1}(\mathrm{x})-\sum_{n=1}^{\infty}\frac{\delta^{n}}{n!}\mathrm{S}_{n}^{<}(\mathrm{x})-\sum_{n=1}^{\infty}\frac{(1-\delta)^{n}}{n!}z_{n}\right]\nonumber\\
    &\hspace{7cm}=-\ln{\hat{\delta}}+\hat{\mathrm{x}}+\ln{(1-\hat{\mathrm{x}})}-\frac{\varepsilon}{8(\lambda a)^{2}}\mathrm{S}^{>}_{0}(\hat{\mathrm{x}})\equiv\frac{\hat{\sigma}^{+}-\hat{\sigma}^{-}}{2a}\equiv\frac{\hat{\sigma}}{a}.\label{x_resenje_<}
\end{align}
\end{widetext}
\noindent Equation (\ref{rho_resenje_>}) represents the solution for the metric when $\mathrm{x}\geqslant1$, while equation (\ref{rho_resenje_<}) represents the solution for the metric when $\mathrm{x}\leqslant1$. Equations (\ref{x_resenje_>}) and (\ref{x_resenje_<}) represent the transcendental equations for the dilaton field $\mathrm{x}=\frac{1}{\lambda a}e^{-\phi}$ when $\mathrm{x}\geqslant1$ and $\mathrm{x}\leqslant1$, respectively, in terms of the coordinates $\sigma^{+}$ (through the dependence of $\delta$) and $\sigma^{-}$ (through the dependence of $\hat{\mathrm{x}}$). The functions $\mathrm{S}_{0}^{>/<}$ and $\mathrm{S}_{n}^{>/<}(\mathrm{x})$ that appear in equations (\ref{rho_resenje_>})-(\ref{x_resenje_<}) are given by equations (\ref{S0_definicija}) and (\ref{I_clanovi}), respectively. Notice that when $\delta=1$ both equations (\ref{rho_resenje_>}) and (\ref{rho_resenje_<}) are reduced to equation (\ref{F-_resenje_final}) for the coordinate transformation $\mathrm{F}^{-}$. This implies that $e^{2\rho}=1$ when $\delta=1$, which is the continuity condition for the metric. In addition, equations (\ref{x_resenje_>}) and (\ref{x_resenje_<}) imply that $\mathrm{x}=\hat{\mathrm{x}}$ when $\delta=1$, which is the continuity condition of the dilaton field. In equations (\ref{x_resenje_>}) and (\ref{x_resenje_<}) we can see new asymptotically flat coordinates $\hat{\sigma}^{\pm}$, defined in equations (\ref{sigma-_hat_resenje_final}) and (\ref{sigma+_hat_resenje_final}).
\medskip
\section{The end-state of black hole evolution}
\par In the previous section we have found an evaporating black hole solution created by collapsing matter. In this section, we investigate the final state of black hole evolution. First, let us calculate the apparent horizon. It is defined by $\partial_{+}e^{-2\phi}=0$, which is equivalent to $\mathrm{x}\partial_{+}\mathrm{x}=0$. Using formula (\ref{qq3}), the direct calculation yields the following:
\begin{widetext}
\begin{equation}
    \mathrm{x}=1+\frac{\varepsilon}{8(\lambda a)^{2}}\left[2\ln{\delta}+8-2(\mathrm{x}-2)\ln{\mathrm{x}}-\frac{3}{\mathrm{x}}+\frac{4}{\mathrm{x}^{2}}-5\mathrm{x}+(\mathrm{x}-1)\sum_{n=1}^{\infty}\frac{\delta^{n}}{n!}\left(\frac{\mathrm{x}-1}{\mathrm{x}}\frac{\de\mathrm{S}_{n}(\mathrm{x})}{\de \mathrm{x}}-n\mathrm{S}_{n}(\mathrm{x})\right)\right].\label{AH_opste}
\end{equation}
\end{widetext}
Solving equation (\ref{AH_opste}) perturbatively, by choosing ansatz $\mathrm{x}_{AH}=\mathrm{x}^{(0)}_{AH}+\varepsilon\mathrm{x}_{AH}^{(1)}$, we obtain:
\begin{equation}
    \mathrm{x}_{AH}=1+\frac{\varepsilon}{4(\lambda a)^{2}}\bigg{(}2+\ln{\delta}\bigg{)}.\label{AH}
\end{equation}
As time elapses, $\delta$ gets smaller and smaller, which in turn reduces the apparent horizon. That is, in fact, the expected result. Note that we have used expression (\ref{rho_resenje_<}) in calculating $\partial_{+}\mathrm{x}$. The same result for the apparent horizon can be derived using the formula (\ref{rho_resenje_>}) since $\mathrm{x}_{AH}^{(0)}=1$. It is not possible to find the dependence $\hat{\mathrm{x}}_{AH}(\delta)$ using equations (\ref{x_resenje_>}) or (\ref{x_resenje_<}) since they diverge when $\mathrm{x}=1+\mathcal{O}(\varepsilon)$. This is an artifact of the perturbative expansion. We need to rewrite these equations in a more suitable form. The details are given in Appendix \ref{app_C}. Using the new notation defined in Appendix \ref{app_C}, and equating expression (\ref{d_+xy}) to zero, the apparent horizon is given by:
\begin{equation}
    \mathrm{y}_{AH}=\sqrt{1+\frac{\varepsilon}{(\lambda a)^{2}}}.\label{y_AH}
\end{equation}
It is easy to check that this equation reproduces the same result (\ref{AH}). Using formula (\ref{resenje_oko_1_opste}) the $\hat{\mathrm{x}}(\delta)$ dependence of the apparent horizon is given by:
\begin{equation}
    \hat{\mathrm{x}}_{AH}=1+\frac{\varepsilon}{4(\lambda a)^{2}}\left[1+\left(1+\frac{\varepsilon}{4(\lambda a)^{2}}\ln{\delta}\right)^{\frac{4(\lambda a)^{2}}{\varepsilon}}\right].\label{x_hat(delta)_AH}
\end{equation}
After taking the limit $\varepsilon\to0$ of the term within the square brackets, since that term is already of $\mathcal{O}(\varepsilon)$ order, we get the following simplified expression for the apparent horizon in $(\hat{\mathrm{x}},\delta)$ coordinates:
\begin{equation}
    \hat{\mathrm{x}}_{AH}(\delta)=1+\frac{\varepsilon}{4(\lambda a)^{2}}(1+\delta).\label{x_hat_AH_early}
\end{equation}
\par The next step is to calculate the line of the singularity. In order to do that, we need to calculate the scalar curvature. It is given by: $R=8e^{-2\rho}\partial_{+}\partial_{-}\rho$. With the help of the equations of motion (\ref{pm_q}-\ref{chi2_q}) the curvature can be expressed in terms of the first derivatives of the dilaton field:
\begin{equation}
    R=\frac{2}{a^{2}}\frac{1+4a^{2}e^{-2\rho}\partial_{+}\mathrm{x}\partial_{-}\mathrm{x}}{\mathrm{x}^{2}-\frac{\varepsilon}{(\lambda a)^{2}}}.\label{skalarna_krivina_formula}
\end{equation} 
Note that equation (\ref{skalarna_krivina_formula}) naively implies a change in the position of the singularity from $\mathrm{x}_{S}=0$ to $\mathrm{x}_{S}=\frac{\sqrt{\varepsilon}}{\lambda a}$ in the presence of quantum corrections. If this holds true, it might lead to an issue within the perturbative approach. Using the expression (\ref{skalarna_krivina_formula}) and the recurrence formula (\ref{S_n_rekurentna}), the scalar curvature becomes:
\begin{widetext}
    \begin{equation}
        R=\frac{2}{a^{2}\mathrm{x}^{3}}\left\{1+\frac{\varepsilon}{8(\lambda a)^{2}}\left[2\ln{\mathrm{x}}-3-\frac{2}{\mathrm{x}}+\frac{15}{\mathrm{x}^{2}}+2\mathrm{x}+2\ln{\delta}+\frac{(\mathrm{x}-1)^{2}}{\mathrm{x}}\sum_{n=1}^{\infty}\frac{\delta^{n}}{n!}\frac{\de\mathrm{S}_{n}(\mathrm{x})}{\de\mathrm{x}}\right]\right\}.\label{Skalarna_krivina}
    \end{equation}
\end{widetext}
Let us first consider the beginning of the evaporation $\delta=1$. Using the formula (\ref{sum_1}), we arrive at the following:
\begin{equation}
    R=\frac{2}{a^{2}\mathrm{x}^{3}}\left(1+\frac{3\varepsilon}{2(\lambda a)^{2}}\frac{1}{\mathrm{x}^{2}}\right).\label{Krivin_delta=1}
\end{equation}
Equation (\ref{Krivin_delta=1}) yields two solutions for the line of singularity:
\begin{equation}
    \mathrm{x}_{S}=0\hspace{2mm}\text{and}\hspace{2mm}\mathrm{x}_{S}\approx\sqrt{\frac{3\varepsilon}{2(\lambda a)^{2}}}.\label{Singularity}
\end{equation}
Neither of these two solutions corresponds to the expected result $\mathrm{x}_{S}=\frac{\sqrt{\varepsilon}}{\lambda a}$. This could be a consequence of the perturbative approach. Fortunately, it is possible to find a non-perturbative expression for the scalar curvature at the beginning of the evaporation. We will solve the equations of motion along the $\sigma^{+}=\sigma^{+}_{0}+\epsilon,\hspace{2mm}\epsilon\to0$ hypersurface. Here, the metric is fixed: $e^{2\rho}=1$, and for the dilaton field we have $\mathrm{x}=\frac{\sigma^{+}_{0}-\sigma^{-}}{2a}$. Equation (\ref{qq2}) yields:
\begin{equation}
    \mathrm{x}\partial_{+}\mathrm{x}+\frac{\varepsilon}{(\lambda a)^{2}}\partial_{+}\rho=\frac{\mathrm{x}-1}{2a}.\label{zs1}
\end{equation}
Using (\ref{zs1}) it is easy to show that equation (\ref{chi2_q}) is satisfied. Combining equations (\ref{pm_q}) and (\ref{zs1}) results in the following differential equation for $\partial_{+}\rho=f(\mathrm{x(\sigma^{-})})$:
\begin{equation}
    \frac{\mathrm{d}f}{\mathrm{dx}}=-\frac{1}{2a}\frac{1+2a\frac{\varepsilon}{(\lambda a)^{2}}f}{\mathrm{x}\left(\mathrm{x}^{2}-\frac{\varepsilon}{(\lambda a)^{2}}\right)}.\label{zs2}
\end{equation}
The solution to equation (\ref{zs2}) is given by:
\begin{equation}
    f(\mathbf{x})=\frac{\lambda^{2}a}{2\varepsilon}\left(\frac{Q\mathrm{x}}{\sqrt{\mathrm{x}^2-\frac{\varepsilon}{(\lambda a)^{2}}}}-1\right),\label{zs3}
\end{equation}
where $Q$ is a constant. By direct comparison between equations (\ref{zs3}) and the solution for $\partial_{+}\rho$ derived from (\ref{rho_resenje_<}) when $\delta=1$ we conclude that $Q=1$. Then, the solution for $\partial_{\pm}\rho$ and $\partial_{\pm}\mathrm{x}$, along with $\sigma^{+}=\sigma^{+}_{0}$, is given by:
\begin{align}
    \partial_{+}\mathrm{x}&=\frac{1}{2a}\left(1-\frac{1}{\sqrt{\mathrm{x}^{2}-\frac{\varepsilon}{(\lambda a)^{2}}}}\right),\hspace{5mm}\partial_{-}\mathrm{x}=-\frac{1}{2a},\label{x_izvod_delta=1}\\
    \partial_{+}\rho&=\frac{\lambda^{2}a}{2\varepsilon}\left(\frac{\mathrm{x}}{\sqrt{\mathrm{x}^2-\frac{\varepsilon}{(\lambda a)^{2}}}}-1\right),\hspace{5mm}\partial_{-}\rho=0.\label{rho_izvod_delta=1}
\end{align}
Finally, the scalar curvature reads:
\begin{equation}
    R=\frac{2}{a^{2}}\frac{1}{\left(\mathrm{x}^{2}-\frac{\varepsilon}{(\lambda a)^{2}}\right)^{\frac{3}{2}}}.\label{Krivina_delta=1_neperturbativno}
\end{equation}
We deduce that at the beginning of the evaporation, the singularity appears at:
\begin{equation}
    \mathrm{x}_{S}=\frac{\sqrt{\varepsilon}}{\lambda a},\label{Singularnost_kvantna korekcija}
\end{equation}
as expected. When all the constants are returned, we get $\varphi_{S}=\sqrt{\frac{\hbar G}{12\pi c^{3}}}$, which is of Planck length order of magnitude.
\par Next, we find the quantum-corrected formula for the line of singularity (\ref{singularnost}). In the $\mathrm{y}$ coordinates (introduced in Appendix \ref{app_C}), the singularity is given by:
\begin{equation}
    \mathrm{y}_{S}=\frac{\sqrt{\varepsilon}}{\lambda a}\frac{1}{1+\frac{\varepsilon}{4(\lambda a)^{2}}\ln{\delta_{S}}}.\label{sing_1}
\end{equation}
Placing the expression (\ref{sing_1}) in equation (\ref{y_<_resenje}) and expanding to the $\mathcal{O}(\varepsilon)$ term leads to the following equation for the singularity:
\begin{widetext}
    \begin{equation}
        \left(1+\frac{9\varepsilon}{8(\lambda a)^{2}}\right)\ln{\delta}_{S}+\frac{\varepsilon}{8(\lambda a)^{2}}\sum_{n=1}^{\infty}\frac{\delta_{S}^{n}}{n!}\mathrm{S}_{n}^{<}(0)=\hat{\mathrm{x}}_{S}+\ln{|\hat{\mathrm{x}}_{S}-1|}-\frac{\varepsilon}{8(\lambda a)^{2}}\mathrm{S}_{0}^{>}(\hat{\mathrm{x}_{S}}).\label{Singularnost_qm}
    \end{equation}
\end{widetext}
Equation (\ref{Singularnost_qm}) tells us that the singularity (\ref{singularnost}) gets a quantum correction of $\mathcal{O}(\varepsilon)$ order.
\subsection{The end-point of the evaporation}
\par By the cosmic censorship conjecture, a naked singularity cannot exist. On the other hand, the two hypersurfaces $\mathrm{y}_{AH}(\delta)$ (given by equation (\ref{y_AH})) and $\mathrm{y}_{S}(\delta)$ (\ref{sing_1}) cross at a finite value of the $\sigma^{+}$ coordinate (or $\delta$ coordinate), that is $\tilde{a}_{E}=\frac{\sqrt{\varepsilon}}{\lambda a}$. The point defined by this value is the end-point of the evaporation. It is given by:
\begin{equation}
    \sigma^{+}_{E}=\sigma^{+}_{0}+\frac{8a(\lambda a)^{2}}{\varepsilon}\left(1-\frac{\sqrt{\varepsilon}}{\lambda a}\right).\label{sigma_+_End}
\end{equation}
Replacing all the constants, equation (\ref{sigma_+_End}) becomes:
\begin{equation}
    \sigma^{+}_{E}=\sigma^{+}_{0}+\frac{768\pi M^{3}G^{2}}{\hbar\lambda^{4}}\left(1-\frac{\lambda}{4M}\sqrt{\frac{\hbar}{3\pi G}}\right).\label{t_end}
\end{equation}
\noindent The formula (\ref{t_end}) tells us that the dominant term for the time of evaporation of the black hole scales as $M^{3}$, which is the expected result from a thermodynamics standpoint, as well as the result derived in \cite{Dimred}. It will be useful to express the end-point of the evaporation in terms of the $\delta$ coordinate:
\begin{equation}
    \delta_{E}=e^{-\frac{4(\lambda a)^{2}}{\varepsilon}\left(1-\frac{\sqrt{\varepsilon}}{\lambda a}\right)}.\label{delta_End}
\end{equation}
The quantity (\ref{delta_End}) is very small $\delta_{E}\ll\frac{\varepsilon}{4(\lambda a)^{2}}$. On the other hand $\ln{\delta_{E}}$ is very large, as it scales as $1/\varepsilon$. This shows that the perturbative approach breaks down at this point. We can still get some useful new insights if we tread carefully. To obtain the exact position of the end-point of the evaporation, in addition to $\delta_{E}$, we need to calculate $\hat{\mathrm{x}}_{E}$. From $\mathrm{y}_{E}=\mathrm{y}_{AH}=1+\frac{\varepsilon}{2(\lambda a)^{2}}$, it follows that $\hat{\mathrm{x}}_{E}=1+\mathcal{O}(\varepsilon)$. Using the expression (\ref{resenje_oko_1}) the other coordinate of the end-point of the evaporation is:
\begin{equation}
    \hat{\mathrm{x}}_{E}=1+\frac{\varepsilon}{4(\lambda a)^{2}}\left[1+e^{-\frac{4(\lambda a)^{2}}{\varepsilon}\left(1-\frac{\sqrt{\varepsilon}}{\lambda a}\right)}\right].\label{x_hat_E}
\end{equation}
Note that formula (\ref{x_hat_E}) can be directly derived using (\ref{x_hat_AH_early}) by placing $\delta=\delta_{E}$.
\par Finally, we also comment on the position of the event horizon of the black hole. It is given by $\hat{\mathrm{x}}=\hat{\mathrm{x}}_{E}$. We expect it to appear when $\mathrm{x}<1$ but very close to the $\mathrm{x}=1$ hypersurface at the beginning of the evaporation. Then, it would gradually decrease until it intersects with the singularity. To find the line of the horizon $\mathrm{y}_{H}(\delta)$ let us revisit equation (\ref{resenje_oko_1_opste}). The expression is given by:
\begin{equation}
    \mathrm{y}_{H}=1+\frac{\varepsilon}{4(\lambda a)^{2}}\left[1+\left(\frac{\tilde{a}_{E}}{\tilde{a}_{H}}\right)^{\frac{4(\lambda a)^{2}}{\varepsilon}}\right].\label{y_H}
\end{equation}
Notice that in the limit $\tilde{a}_{H}\to \tilde{a}_{E}$ the expression (\ref{y_H}) reduces to $\mathrm{y}_{H}=\mathrm{y}_{AH}$, which is the expected result. Now, it is easy to calculate $\mathrm{x}_{H}$. It is given by the following equation:
\begin{equation}
    \mathrm{x}_{H}=\left(1+\frac{\varepsilon}{4(\lambda a)^{2}}\ln{\delta_{H}}\right)\left[1+\frac{\varepsilon}{4(\lambda a)^{2}}\left(1+\frac{\delta_{E}}{\delta_{H}}\right)\right].\label{x_H}
\end{equation}
From equation (\ref{skalarna_krivina_formula}) along the apparent horizon hypersurface, the scalar curvature is the following:
\begin{equation}
    R_{AH}=\frac{2}{a^{2}}\frac{1}{\mathrm{x}^{2}-\frac{\varepsilon}{(\lambda a)^{2}}}.\label{R_AH}
\end{equation}
This expression can also be derived using (\ref{Skalarna_krivina}), which can be rewritten as:
\begin{widetext}
    \begin{equation}
        R=\frac{2}{a^{2}}\frac{1}{\mathrm{x}^{2}-\frac{\varepsilon}{(\lambda a)^{2}}}\frac{1}{\sqrt{\mathrm{y}^{2}-\frac{\varepsilon}{(\lambda a)^{2}}}}\left\{1+\frac{\varepsilon}{8(\lambda a)^{2}}\left[2\ln{\mathrm{y}}-3-\frac{2}{\mathrm{y}}+\frac{3}{\mathrm{y}^{2}}+2\mathrm{y}+\frac{(\mathrm{y}-1)^{2}}{\mathrm{y}}\sum_{n=1}^{\infty}\frac{\delta^{n}}{n!}\frac{\mathrm{d}\mathrm{S}^{<}_{n}}{\mathrm{dy}}\right]\right\},\label{krivina_y}
    \end{equation}
\end{widetext}
then, by applying (\ref{y_AH}) it becomes equivalent to (\ref{R_AH}), as expected.
\subsection{The radiated energy}
In this part, we will determine the amount of energy that was emitted by the black hole throughout its lifespan. It will be calculated as the energy flux at future null infinity from the beginning of the evaporation, $\hat{\sigma}^{-}\to-\infty$, until the end-point of the evaporation, $\hat{\sigma}^{-}=\hat{\sigma}^{-}_{E}$,
\begin{equation}
    E_{rad}=\int_{-\infty}^{\hat{\sigma}^{-}_{E}}\de\hat{\sigma}^{-}\langle\Delta T_{--}^{(f)}(\hat{\sigma}^{-})\rangle\bigg{|}_{\mathcal{I}^{+}}.\label{uuu1}
\end{equation}
Considering the fact that the metric is flat along the $\mathcal{I}^{+}$ hypersurface, the previous integral (\ref{uuu1}) reduces to:
\begin{equation}
    E_{rad}=-\frac{\varepsilon}{G}\int_{-\infty}^{\sigma^{-}_{E}}\frac{\de\sigma^{-}}{\mathrm{F}^{-}}t_{-}(\hat{\sigma}^{-})=-\frac{\varepsilon}{8aG}\mathrm{F}(\hat{\mathrm{x}}_{E}).\label{E_rad}
\end{equation}
\noindent Since the Schwarzschild radius $\lambda a$ is directly proportional to the mass, we can use it as a measure of energy. In this regard, we can associate a value of $\lambda a_{rad}=\frac{2E_{rad}G}{\lambda}$ to the radiated energy. Then the value of $\lambda a_{rad}$ will be given by:
\begin{equation}
    \lambda a_{rad}=-\frac{\varepsilon}{4\lambda a}\mathrm{F}(\hat{\mathrm{x}}_{E}).\label{a_rad}
\end{equation}
The function $\mathrm{F}(\hat{\mathrm{x}})$ is not well defined when $\hat{\mathrm{x}}\to1$. This is the same type of problem that we address in Appendix \ref{app_C}. One should calculate the first-order correction to the function $\mathrm{F}(\hat{\mathrm{x}})$:
\begin{equation}
    \mathrm{F}(\hat{\mathrm{x}})=-2\frac{\mathrm{d}\mathrm{F}^{-}}{\mathrm{d}\hat{\mathrm{x}}}+\int\frac{\mathrm{d}\hat{\mathrm{x}}}{\mathrm{F}^{-}}\left(\frac{\mathrm{d}\mathrm{F}^{-}}{\mathrm{d}\hat{\mathrm{x}}}\right)^{2}.
\end{equation}
Using equation (\ref{y_<_resenje}) we arrive at the following equations:
\begin{align}
    \frac{1}{\mathrm{F}^{-}}&=\frac{\mathrm{d}\mathcal{J}}{\mathrm{d}\hat{\mathrm{x}}}-\frac{\varepsilon}{8(\lambda a)^{2}}\left(\frac{\hat{2\mathrm{x}}}{\hat{\mathrm{x}}-1}\ln{\hat{\mathrm{x}}}+\frac{3}{\hat{\mathrm{x}}}+7\right),\label{1/F^-}\\
    \frac{\mathrm{d}\mathrm{F}^{-}}{\mathrm{d}\hat{\mathrm{x}}}&=-\mathrm{F}^{-2}\left[\frac{\mathrm{d}^{2}\mathcal{J}}{\mathrm{d}\hat{\mathrm{x}}^{2}}+\frac{\varepsilon}{8(\lambda a)^{2}}\left(\frac{3}{\hat{\mathrm{x}}^{2}}+2\frac{\ln{\hat{\mathrm{x}}}-1}{\hat{\mathrm{x}}-1}\right)\right].\label{1/F^-_izovd}
\end{align}
Since terms of order $\mathcal{O}(\varepsilon)$ behave well when $\hat{\mathrm{x}}\to1$, they can be neglected. After preforming substitution to $\mathrm{z}$ (defined in equation (\ref{integral_z})), upper limit of integration becomes 1, while the lower limit of integration changes to $\mathrm{z}_{E}=\mathrm{z}_{1}^{+}+\left(\frac{\sqrt{\varepsilon}}{\lambda a}\right)^{3}\frac{\delta_{E}}{4}$, where $\mathrm{z}_{1}^{+}$ is given by equation (\ref{z1+}). Now, the integral is performed over a narrow interval around the point $\mathrm{z}=1$ which allows us to substitute $\mathrm{z}=1$ wherever it is well defined within the integrand function. The integral becomes:
\begin{equation}
    \int_{\infty}^{\hat{\mathrm{x}}_{E}}\frac{\mathrm{d}\hat{\mathrm{x}}}{\mathrm{F}^{-}}\left(\frac{\mathrm{d}\mathrm{F}^{-}}{\mathrm{d}\hat{\mathrm{x}}}\right)^{2}=-\frac{4(\lambda a)^{4}}{\varepsilon^{2}}\int_{\mathrm{z}_{E}}^{1}\frac{(1-\mathrm{z})^{4}\mathrm{dz}}{(\mathrm{z}-\mathrm{z_{1}^{+}})(\mathrm{z}-\mathrm{z}_{2}^{+})^{2}}.
\end{equation}
Direct calculation up to $\mathcal{O}(1)$ order leads to the following expression for the function $\mathrm{F}(\hat{\mathrm{x}}_{E})$:
\begin{equation}
    \mathrm{F}(\hat{\mathrm{x}}_{E})=-\frac{4(\lambda a)^{2}}{\varepsilon}+\frac{8(\lambda a)}{\sqrt{\varepsilon}}+\ln{\frac{\varepsilon}{4(\lambda a)^{2}}+\frac{3}{2}}.\label{F_x_hat_E}
\end{equation}
The radiated energy (\ref{E_rad}) is given by:
\begin{widetext}
    \begin{equation}
        E_{rad}=Mc^{2}\left[1-\sqrt{\frac{\hbar c}{12\pi M^{2}G_{N}^{(4)}}}-\frac{\hbar c}{192\pi M^{2}G_{N}^{(4)}}\ln{\frac{\hbar c}{192\pi M^{2}G_{N}^{(4)}}}-\frac{3\hbar c}{384\pi M^{2}G_{N}^{(4)}}+o\left(\frac{\hbar c}{M^{2}G_{N}^{(4)}}\right)\right],\label{izraceno}
    \end{equation}
\end{widetext}
where all of the constants have been returned. Equation (\ref{izraceno}) implies that during the evaporation process, almost the entire initial mass of the black hole evaporates, up to a small portion proportional to $\sqrt{\varepsilon}$. This result coincides with the results within the RST and BPP models of dilaton gravity \cite{RST1,BPP}.
\subsection{The end-state geometry}
Previously, we have calculated the end-point of the evaporation. Now we may ask: "what is the metric of space-time after the black hole evaporates?" In the perturbative framework, it may be hard to answer this question since it breaks near the end-point of the evaporation. However, we may try to find some useful information.
\par First, let us rewrite the equations for the metric (\ref{rho_resenje_>}) and the dilaton field (\ref{x_resenje_>}) in a more suitable form. We begin by calculating $\mathrm{F}(\hat{\mathrm{x}})$ defined by formula (\ref{F_def}) in terms of $\mathrm{x}$ and $\delta$. Using the same idea as in theorem \ref{T1}.
\begin{equation}
    \mathrm{F}(\hat{\mathrm{x}})=\mathrm{F}(\mathrm{x})+\sum_{n=1}^{\infty}\frac{(\delta-1)^{n}}{n!}\frac{\partial^{n}\mathrm{F}(\hat{\mathrm{x}})}{\partial\delta^{n}}\bigg{|}_{\delta=1}.\label{iii1}
\end{equation}
Mirroring the method used in the theorem \ref{T7}, we find the following expression:
\begin{equation}
    \frac{\partial^{n}\mathrm{F}(\hat{\mathrm{x}})}{\partial\delta^{n}}\bigg{|}_{\delta=1}=e^{n\mathrm{x}}(\mathrm{x}-1)^{n}\left(\frac{e^{-\mathrm{x}}}{\mathrm{x}}\frac{\de}{\de\mathrm{x}}\right)^{n-1}\frac{e^{-\mathrm{x}}}{\mathrm{x}}\frac{\de\mathrm{F}(\mathrm{x})}{\de\mathrm{x}}.\label{iii2}
\end{equation}
On the other hand, one can demonstrate, by induction, that the following expression is satisfied:
\begin{equation}
    \frac{\partial^{n}\mathrm{F}(\hat{\mathrm{x}})}{\partial\delta^{n}}\bigg{|}_{\delta=1}=(-1)^{n}\left[\frac{1}{2}\frac{(\mathrm{x}-1)^{2}}{\mathrm{x}}\frac{\de\mathrm{S}^{>}_{n}(\mathrm{x})}{\de\mathrm{x}}-(n-1)!\right].\label{iii3}
\end{equation}
The inductive step reduces to the recurrence formula for the $\mathrm{S}^{>}_{n}(\mathrm{x})$ functions (\ref{S_n_rekurentna}). Upon replacing this expression for $\frac{\partial^{n}\mathrm{F}(\hat{\mathrm{x}})}{\partial\delta^{n}}|_{\delta=1}$ into expansion (\ref{iii1}), we get the following equation:
\begin{equation}
    \mathrm{F}(\hat{\mathrm{x}})=\mathrm{F}(\mathrm{x})+\ln{\delta}+\frac{(\mathrm{x}-1)^{2}}{2\mathrm{x}}\sum_{n=1}^{\infty}\frac{(1-\delta)^{2}}{n!}\frac{\de\mathrm{S}^{>}_{n}(\mathrm{x})}{\de\mathrm{x}}.\label{F(x,delta)}
\end{equation}
Eliminating the $\ln{\delta}$ term in the formulas for the metric (\ref{rho_resenje_>}) and the tortoise coordinate $\hat{\sigma}$ (\ref{x_resenje_>}), with the help of equation (\ref{F(x,delta)}), we arrive at the following.
\begin{widetext}
\begin{align}
    &\mathrm{F}^{+}\mathrm{F}^{-}e^{2\rho}=\frac{1+\frac{\varepsilon}{4(\lambda a)^{2}}\mathrm{F}(\hat{\mathrm{x}})}{\mathrm{x}}-1+\frac{\varepsilon}{4(\lambda a)^{2}}\left[\frac{\mathrm{x}-2}{\mathrm{x}}\ln{\frac{\mathrm{x}-1}{\mathrm{x}}}+\frac{1}{\mathrm{x}}-\frac{3}{2}\frac{1}{\mathrm{x}^{2}}+\frac{2}{\mathrm{x}^{3}}\right.\nonumber\\
    &\hspace{5cm}\left.+\frac{\mathrm{x}-1}{2\mathrm{x}}\sum_{n=1}^{\infty}\frac{(1-\delta)^{n}}{n!}\bigg{(}\mathrm{S}^{>}_{n+1}(\mathrm{x})-n\mathrm{S}^{>}_{n}(\mathrm{x})+z_{n+1}-nz_{n}\bigg{)}\right],\label{rho_prilagodjena}\\
    &\frac{\hat{\sigma}}{a}=\mathrm{x}+\left(1+\frac{\varepsilon}{4(\lambda a)^{2}}\mathrm{F}(\hat{\mathrm{x}})\right)\ln{\left(\frac{\mathrm{x}}{1+\frac{\varepsilon}{4(\lambda a)^{2}}\mathrm{F}(\hat{\mathrm{x}})}-1\right)}-\frac{\varepsilon}{8(\lambda a)^{2}}\left[\mathrm{S}^{>}_{0}(\mathrm{x})-\frac{\pi^{2}}{3}-2-2\left(\frac{\mathrm{x}}{\mathrm{x}-1}-\ln{(\mathrm{x}-1)}\right)\mathrm{F}(\mathrm{x})\right.\nonumber\\
    &\hspace{4cm}\left.-9\ln{\delta}-\sum_{n=1}^{\infty}\frac{(1-\delta)^{n}}{n!}\left(\mathrm{S}^{>}_{n}(\mathrm{x})+z_{n}+\frac{(\mathrm{x}-1)^{2}}{\mathrm{x}}\left(\frac{\mathrm{x}}{\mathrm{x}-1}-\ln{(\mathrm{x}-1)}\right)\frac{\de\mathrm{S}^{>}_{n}(\mathrm{x})}{\de\mathrm{x}}\right)\right].\label{x_prilagodjena}
\end{align}
\end{widetext}
In equations (\ref{rho_prilagodjena}) and (\ref{x_prilagodjena}) the term $\varphi/\lambda a$ has been replaced by the term $\varphi/\left(\lambda a+\frac{\varepsilon}{4\lambda a}\mathrm{F}(\hat{\mathrm{x}})\right)$, in zeroth order in $\varepsilon$. Let us define a new, running, Schwarzschild radius, dependent on the $\sigma^{-}$ coordinate:
\begin{equation}
    \lambda\hat{a}(\hat{\mathrm{x}})=\lambda a+\frac{\varepsilon}{4\lambda a}\mathrm{F}(\hat{\mathrm{x}}).\label{a_tilda}
\end{equation}
From the perspective of the first-order terms in $\varepsilon$, $\lambda\hat{a}$ is equal to $\lambda a$, so we can exchange $\lambda a$ in the equations for the metric and the dilaton field for $\lambda\hat{a}$.
\par To understand the final state of evaporation, we first must examine what happens along the $\hat{\mathrm{x}}=\hat{\mathrm{x}}_{E}$ hypersurface when $\delta<\delta_{E}$. Along this hypersurface $\lambda\hat{a}(\hat{\mathrm{x}}_{E})$ is constant. Since $\mathrm{F}(\hat{\mathrm{x}}_{E})\sim\ln{\delta_{E}}$, equation (\ref{F_x_hat_E}) implies $\lambda\hat{a}=\mathcal{O}(\sqrt{\varepsilon})$, which within the first-order of perturbation theory gives $\lambda\hat{a}=0$. Taking this limit is akin to taking the $\mathrm{x}\to\infty$ limit discussed in the previous section. The form of $\mathrm{S}^{>}_{n}(\mathrm{x})$ from equation (\ref{S_n_uslov}) implies that:
\begin{equation}
    \mathrm{S}^{>}_{n}(\mathrm{x})=(n-1)!\mathrm{S}^{>}_{1}(\mathrm{x})-z_{n}+\mathcal{O}((\lambda\hat{a})^{4}).\label{razvojj}
\end{equation}
Taking the limit $\lambda\hat{a}\to0$ in the first-order term of the metric equation (\ref{rho_prilagodjena}) reduces that term to zero, since the sum behaves as $\mathcal{O}((\lambda\hat{a})^{2})$ and the $\delta$ independent term is $\mathcal{O}(\lambda\hat{a})$. Next, we take the limit $\lambda\hat{a}\to0$ of the first-order term in $\varepsilon$ of the dilaton equation (\ref{x_prilagodjena}). In this limit, the $\delta$ dependent term becomes:
\begin{equation}
    \frac{\lambda\hat{a}\ln{\delta}}{(\lambda\hat{a})^{2}}\left[\mathrm{S}^{>}_{1}+\frac{(\mathrm{x}-1)^{2}}{\mathrm{x}}\left(\frac{\mathrm{x}}{\mathrm{x}-1}-\ln{(\mathrm{x}-1)}\right)\frac{\de\mathrm{S}^{>}_{1}}{\de\mathrm{x}}-9\right].\label{jj1}
\end{equation}
This term vanishes as can be seen by expanding in terms of $\lambda\hat{a}$ and then taking the $\lambda\hat{a}\to0$ limit. An important observation to make is the fact that $\ln{\delta}$ behaves as $1/\lambda\hat{a}$, which means that $\lambda\hat{a}\ln{\delta}$ is constant. The $\delta$ independent terms vanish almost in the same way that they vanish when taking the $\mathrm{x}\to\infty$ limit. By taking this limit, we get the following results:
\begin{align}
    &\mathrm{F}^{+}\mathrm{F}^{-}e^{2\rho}=\frac{\lambda\hat{a}_{E}}{\sqrt{\varphi^{2}-\varepsilon}}-1,\label{Metrika_E}\\
    &\varphi+\lambda\hat{a}_{E}\ln{\left(\frac{\varphi}{\lambda\hat{a}_{E}}\right)}=\lambda\hat{\sigma}.\label{dilaton_E}
\end{align}
Since $\lambda\hat{a}_{E}=\lambda\hat{a}(\hat{\mathrm{x}}_{E})$ is of order $\mathcal{O}(\sqrt{\varepsilon})$, equations (\ref{Metrika_E}) and (\ref{dilaton_E}) define a quantum-corrected Minkowski space-time. Comparing to BPP or RST models of dilaton gravity \cite{RST1,RST2,BPP}, one could expect a shockwave at the end-point of the evaporation known as a thunderpop. To check if this occurs in the DREH model, one needs to continuously connect regions II and III (Figure \ref{fig2}) along the $\sigma^{-}=\sigma^{-}_{E}$ hypersurface by integrating the "$--$" equation from (\ref{ppmm_q}):
\begin{align}
    e^{-2\rho}\partial_{-}\varphi_{>}-e^{-2\rho}\partial_{-}\varphi_{<}&=\frac{\lambda}{2\mathrm{F}^{-}}\lambda a_{0}\varphi^{-1}e^{-2\rho}\label{--qq}\\
    &-\varepsilon(\partial_{-}\rho_{>}-\partial_{-}\rho_{<})\varphi^{-1}e^{-2\rho},\nonumber
\end{align}
where $\lambda a_{0}$ represents the mass of the shockwave at the end-point of the evaporation, "$<$" represents the part of space-time where $\sigma^{-}<\sigma^{-}_{E}$ and "$>$" represents the part of space-time where $\sigma^{-}>\sigma^{-}_{E}$. The derivatives appearing in equation (\ref{--qq}) can be easily calculated using equations (\ref{rho_resenje_>}) and (\ref{x_resenje_>}) in a $\sigma^{+}<\sigma^{+}_{0}$ region of space-time. The results are given by:
\begin{align}
    \partial_{-}\varphi_{<}&=-\frac{\lambda}{2\mathrm{F}^{-}}\bigg{\{}1-\frac{\lambda\hat{a}(\hat{\mathrm{x}})}{\sqrt{\varphi^{2}-\varepsilon}}\bigg{[}1\nonumber\\
    &\hspace{0.5cm}-\frac{\varepsilon}{4(\lambda a)^{2}}\left((2-\mathrm{x})\mathrm{F}(\mathrm{x})-\frac{2}{\mathrm{x}}+\frac{3}{\mathrm{x}^{2}}\right)\bigg{]}\bigg{\}},\label{ll1}\\
    \partial_{-}\rho_{<}&=-\frac{1}{2}\partial_{-}\ln{\mathrm{F}^{-}}-\frac{\lambda}{4\mathrm{F}^{-}}\frac{\lambda\hat{a}(\hat{\mathrm{x}})}{\varphi^{2}}\left[1-\frac{\varepsilon}{2(\lambda a)^{2}}\mathrm{F}(\mathrm{x})\right].\label{ll2}
\end{align}
Analogues to the case of equations (\ref{x_prilagodjena}-\ref{rho_prilagodjena}) when $\hat{\mathrm{x}}=\hat{\mathrm{x}}_{E}$ we take the limit $\lambda\hat{a}\to0$ within the first-order of the perturbation theory. Now, equations (\ref{ll1}-\ref{ll2}) reduce to:
\begin{align}
    \partial_{-}\varphi_{<}&=-\frac{\lambda}{2\mathrm{F}^{-}}\left(1-\frac{\lambda\hat{a}_{E}}{\sqrt{\varphi^{2}-\varepsilon}}\right),\label{ll1_lim}\\
    \partial_{-}\rho_{<}&=-\frac{1}{2}\partial_{-}\ln{\mathrm{F}^{-}}-\frac{\lambda}{4\mathrm{F}^{-}}\frac{\lambda\hat{a}_{E}}{\varphi^{2}}.\label{ll2_lim}
\end{align}
Let us check if the final state of the evaporation is a Minkowski vacuum. Similarly to the case of the initial gravitational collapse, we allow one final coordinate change $\hat{\sigma}^{+}\mapsto\hat{\sigma}^{+}_{f}$. In these new coordinates, the metric and the dilaton field take the following form:
\begin{equation}
    \mathrm{F}^{+}_{>}\mathrm{F}^{-}e^{2\rho}=-1\hspace{5mm}\land\hspace{5mm}\varphi=\frac{\lambda}{2}(\hat{\sigma}^{+}_{f}-\hat{\sigma}^{-}).\label{metrika_dilaton_fin}
\end{equation}
Now we simultaneously solve equations (\ref{pm_q}) and (\ref{--qq}) for $\partial_{-}\varphi_{<}$ and $\partial_{-}\rho_{<}$, the same way as in the case of initial collapse (\ref{zs1}-\ref{rho_izvod_delta=1}) using expressions (\ref{metrika_dilaton_fin}) for the metric and the dilaton field. The solution is given by:
\begin{align}
    \partial_{-}\varphi_{<}&=-\frac{\lambda}{2\mathrm{F}^{-}}\left(1-\frac{Q'}{\sqrt{\varphi^{2}-\varepsilon}}\right),\label{ll1_exp}\\
    \partial_{-}\rho_{<}&=-\frac{1}{2}\partial_{-}\ln{\mathrm{F}^{-}}-\frac{\lambda}{4\mathrm{F}^{-}}\left(\frac{Q'\varphi}{\sqrt{\varphi^{2}-\varepsilon}}+\lambda a_{0}\right).\label{ll2_exp}
\end{align}
Equating equations (\ref{ll1_lim}) and (\ref{ll1_exp}) gives $Q'=\lambda\hat{a}_{E}$; while the comparison between equations (\ref{ll2_lim}) and (\ref{ll2_exp}) results in $Q'+\lambda a_{0}=0$. This successful comparison indicates that, at least to the first order in perturbation theory, the final geometry corresponds to that of the Minkowski vacuum.
\par Balance of energy states that the difference between the energies of space-time at the beginning of the evaporation, given by $\lambda a$, and space-time after the black hole evaporates, given by $\lambda a_{f}$, together with the energy of the shockwave, should be equal to the radiated energy:
\begin{equation}
    \lambda a_{rad}=\lambda a_{0}+\lambda(a-a_{f}).\label{Bilans_energije}
\end{equation}
Upon examining equation (\ref{Bilans_energije}) alongside equation (\ref{E_rad}) concerning the radiated energy, and equation (\ref{a_tilda}) where $Q'=\lambda\hat{a}_{E}$, we find consistency with the condition $a_{f}=0$ (Minkowski vacuum), given that $Q'+\lambda a_{0}=0$. This also implies that the mass of the shockwave $\lambda a_{0}$ is exactly the difference between the initial mass of the black hole and the radiated energy.
\par Finally we need to check if the metric and the dilaton field can be continuously connected along the $\hat{\mathrm{x}}=\hat{\mathrm{x}}_{E}$ hypersurface by the change in the $\hat{\sigma}^{+}$ coordinate. The continuity of the metric (\ref{Metrika_E}) implies:
\begin{equation}
    \mathrm{F}^{+}_{>}=\mathrm{F}^{+}_{<}\frac{\sqrt{\varphi^{2}-\varepsilon}}{\sqrt{\varphi^{2}-\varepsilon}-\lambda\hat{a}_{E}}.\label{F^+_>_fin}
\end{equation}
Using equation (\ref{metrika_dilaton_fin}) for the dilaton field, direct integration of equation (\ref{F^+_>_fin}) yields:
\begin{equation}
    \frac{\lambda}{2}\hat{\sigma}^{+}=\frac{\lambda}{2}\hat{\sigma}^{+}_{f}+\lambda\hat{a}_{E}\ln{\left(\frac{\lambda}{2}\frac{\hat{\sigma}^{+}_{f}-\hat{\sigma}^{-}_{E}}{\lambda\hat{a}_{E}}\right)},\label{sigma^+_fin_deff}
\end{equation}
after expanding with respect to $\varepsilon$. Note that equation (\ref{sigma^+_fin_deff}) exactly reproduces the dilaton equation (\ref{dilaton_E}), which implies the continuity of the dilaton field. The boundary of space-time in the final part of space-time is defined by $\varphi=0$, which implies:
\begin{equation}
    \frac{\lambda}{2}\hat{\sigma}^{+}=\frac{\lambda}{2}\hat{\sigma}^{-}+\lambda\hat{a}_{E}\ln{\left(\frac{\lambda}{2}\frac{\hat{\sigma}^{-}-\hat{\sigma}^{-}_{E}}{\lambda\hat{a}_{E}}\right)}
\end{equation}
\par Finally, we can conclude that the end-state of the radiation is $\ket{0,\hat{\sigma}_{f}}$, the vacuum state in $\hat{\sigma}_{f}$ coordinates, since $t_{\pm}(\hat{\sigma}^{\pm}_{f})=0$. This is implied from equations (\ref{ppmm_q}) since the metric is flat in part III of space-time (Figure \ref{fig2}). In addition, our analysis has shown the appearance of the thunderpop at the end-point of evaporation in the same way as in the case of the BPP and RST models of dilaton gravity \cite{RST1,RST2,BPP}.
\par It is important to note that we cannot be sure of the exact value of the $\varepsilon$ correction to the metric within the $\lambda\hat{a}_{E}$ term, since at this point the perturbation theory breaks and higher-order terms may influence the first-order terms. It is possible that there is no correction at all. To check this, a non-perturbative solution for the end-point of evaporation is needed, or to try to solve the equations of motion (\ref{ppmm_q})-(\ref{chi2_q}) numerically.
\section{Conclusion} 

In this paper, we studied a model of two-dimensional dilaton gravity related to the four-dimensional Einstein-Hilbert action by dimensional reduction. The gravitational collapse scenario has been constructed within this model. Initially, we explored the classical scenario, after which we analyzed the back-reaction of quantum matter fields on the geometry by incorporating the Polyakov-Liouville action into the classical DREH framework. At past time-infinity the metric was that of Minkowski space-time. Then, collapsing matter was formed as an infalling shockwave, which created the singularity when it touched the boundary of space-time. The metric of space-time with an evaporating black hole has been calculated within a perturbative framework. It has been found that, due to quantum corrections, the singularity changes position form $r=0$ to $r=\sqrt{\frac{\hbar G^{(4)}_{N}}{12\pi c^{3}}}$, which is of Planck length order of magnitude. The end-point of evaporation has been identified and it coincides with the thermodynamic result. It has been shown that, at this point, the perturbative approach breaks, but still there have been insights that suggest that the final state of the evaporation is the Minkowski vacuum.
\par The groundwork for this project was laid in our previous work \cite{DREH1}, in which we explored the case of an eternal black hole scenario and successfully derived the corresponding Page curve. The following stage focuses on generating the Page curve for a black hole undergoing evaporation \cite{DREH3}, utilizing the solution presented in the current study. It is important to note that there are relatively few instances where a Page curve has been reproduced for an evaporating black hole, particularly in the context of a Schwarzschild-like solution.

{\bf Acknowledgments}
The authors acknowledge the funding provided by the Faculty of Physics of the University of Belgrade, through a grant number 451-03-47/2023-01/200162 from the Ministry of Education, Science, and Technological Development of the Republic of Serbia. Also S.\Dj. is very grateful to Ana \Dj or\dj evi\'c (University of Belgrade) and Hana Schiff Toman (UC Irvine) for many useful insights throughout the writing of this paper, and once again Ana for the creation of the figures for this paper. 
\bibliography{reference}
\appendix
\setcounter{figure}{0}
\renewcommand{\thefigure}{A\arabic{figure}}
\section*{APPENDIX}

\section{Transformation laws of the energy-momentum tensor}\label{app_A}

The quantum correction to the energy-momentum tensor is defined by equation (\ref{T_1loop}). In the conformal gauge $\de s^{2}=e^{2\rho}\de x^{+}\de x^{-}$, it can be rewritten in the following form:
\begin{align} \label{T_Kruskal}
    \langle\Delta T^{(f)}_{\pm\pm}\rangle&=\frac{\varepsilon}{G}\left[\partial^{2}_{\pm}\rho-\left(\partial_{\pm}\rho\right)^{2}-t_{\pm}(x^{\pm})\right],\\
    \langle\Delta T^{(f)}_{+-}\rangle&=-\frac{\varepsilon}{G}\partial_{+}\partial_{-}\rho,
\end{align} 
where $t_{\pm}(x^{\pm})=\frac{1}{2}\partial_{\pm}^{2}f_{\pm}-\frac{1}{4}(\partial_{\pm}f_{\pm})^{2}$ is a function related to the state of the quantum fields (equation (\ref{psi_resenje})). 
Under a conformal coordinate transformation $y^{\pm}=y^{\pm}(x^{\pm})$ the energy-momentum tensor changes according to (we only consider the quantum correction):
\begin{equation}
\langle\Delta T^{(f)}_{\pm\pm}(y)\rangle=\left(\frac{d x^{\pm}}{d y^{\pm}}\right)^{2}\langle\Delta T^{(f)}_{\pm\pm}(x)\rangle.    
\end{equation}
For the conformal factor $\rho$, on the other hand, we have the following transformation law:
\begin{equation}
\rho(y)=\rho(x)+\frac{1}{2}\ln\frac{d y^{+}}{d x^{+}}\frac{d y^{-}}{d x^{-}}. \end{equation}
These give the transformation law for $t_{\pm}$:
\begin{equation}\label{t_change}
t_{\pm}(y^{\pm})=\left(\frac{d x^{\pm}}{d y^{\pm}}\right)^{2}\left[t_{\pm}(x^{\pm})-\frac{1}{2}D_{x^{\pm}}[y^{\pm}]\right],    
\end{equation}
with the Schwartz derivative defined by: \begin{equation}
D_{x^{\pm}}[y^{\pm}]=\frac{(y^{\pm})'''}{(y^{\pm})'}-\frac{3}{2}\left(\frac{(y^{\pm})''}{(y^{\pm})'}\right)^{2},\label{t_pm_trans}    
\end{equation}
where the derivatives are with respect to $x^{\pm}$.

The vacuum state of the quantum fields and the corresponding set of creation/anihilation operators depend on the reference frame, i.e. on the coordinate system. If we introduce normal ordering of the energy-momentum operator for one choice of the vacuum state $\vert 0; x\rangle$, say in the coordinate system $x^{\pm}$, the energy-momentum operator can be decomposed as:
\begin{equation}\label{separation}
\hat{T}^{(f)}_{\pm\pm}(x^{\pm})=:\hat{T}^{(f)}_{\pm\pm}(x^{\pm}):+\langle 0;x\vert\hat{T}^{(f)}_{\pm\pm}(x^{\pm})\vert 0;x\rangle.     
\end{equation}
If we make a transition to another coordinate system $y^{\pm}$, the energy momentum tensor will not be ordered normally in general. The transformation law for the normally ordered part is given by
\begin{equation}\label{normal}
:\hat{T}^{(f)}_{\pm\pm}(y^{\pm}):=\left(\frac{d x^{\pm}}{d y^{\pm}}\right)^{2}\left[:\hat{T}^{(f)}_{\pm\pm}(x^{\pm}):+\frac{\hbar}{24\pi}D_{x^{\pm}}[y^{\pm}]\right].     
\end{equation}
Comparing (\ref{t_change}), (\ref{separation}) and (\ref{normal}), we can establish the following relationship:
\begin{equation}
\bra{0,x}:\hat{T}^{(f)}_{\pm\pm}(x^{\pm}):\ket{0,x}=-\frac{\varepsilon}{G} t_{\pm}(x^{\pm}), \end{equation}
which provides an interpretation for the quantity $t_{\pm}$. If $t_{\pm}(x^{\pm})=0$ in some coordinates $x^{\pm}$, the energy-momentum tensor is normally ordered in those coordinates, i.e. it's expectation value is zero in the vacuum defined in the coordinates $x^{\pm}$. 

\section{Transcendental equation solution}\label{app_B}

Here we solve the transcendental equation (\ref{transcedentna}), that is:
\begin{equation}
    e^{\hat{\mathrm{x}}}(\hat{\mathrm{x}}-1)=\delta e^{\mathrm{x}}(\mathrm{x}-1),\label{transc}
\end{equation}
\noindent when $\delta\leqslant1$, $\hat{\mathrm{x}}\geqslant0$ and $\mathrm{x}\geqslant0$. Notice that when $\delta=1$, we have a simple solution $\hat{\mathrm{x}}=\mathrm{x}$. This is the only solution, since both functions are monotonically increasing, which implies that they intersect at most once. This is true simply because when $\sigma^{+}=\sigma^{+}_{0}$, the metric becomes that of the Minkowski space-time (see section \ref{grav_kolaps}). This fact will be extensively used as well as the fact that $\delta<1$ when $\sigma^{+}>\sigma^{+}_{0}$. Two cases will be studied, since the function on the right-hand side of equation (\ref{transc}) behaves differently depending on whether $\mathrm{x}\geqslant1$ or $\mathrm{x}\leqslant1$. 
\par It is important to further discuss the condition $\hat{\mathrm{x}}\geqslant0$. Figure \ref{fig1} shows space-time. In part I of space-time the boundary is defined by $\mathrm{x}=0$. The formation of the singularity happens when the collapsing matter's world-line $\sigma^{+}=\sigma^{+}_{0}$ hits the boundary of space-time. This corresponds to $\hat{\mathrm{x}}=0$, which implies that $\hat{\mathrm{x}}\geqslant0$.
\begin{theorem}\label{T1}
    The solution to equation (\ref{transc}) $\hat{\mathrm{x}}=\hat{\mathrm{x}}(\mathrm{x},\delta)$, when $\delta\leqslant1$, $\hat{\mathrm{x}}\geqslant0$ and $\mathrm{x}\geqslant1$ is given by the following functional series:
    \begin{equation}
    \hat{\mathrm{x}}(\mathrm{x},\delta)=\mathrm{x}-\sum_{n=1}^{\infty}\frac{(1-\delta)^{n}}{n!}\left(1-\frac{1}{\mathrm{x}}\right)^{n}\mathrm{P}_{n-1}\left(\frac{1}{\mathrm{x}}\right),\label{x>1resenje}
\end{equation}
    where $\mathrm{P}_{n}(x)$ are polynomials of the $n$-th degree, defined by the following recurrence relation,
\begin{equation}
    \mathrm{P}_{n}(x)=\left[n(x+1)+x^{2}\frac{\de}{\de x}\right]\mathrm{P}_{n-1}(x),\label{P_rekurentna}
\end{equation}
    and the initial condition $\mathrm{P}_{0}(x)=1$.
\end{theorem}
\begin{proof}
We expand the function $\hat{\mathrm{x}}=\hat{\mathrm{x}}(\mathrm{x},\delta)$ in a Taylor series around $\delta=1$:
\begin{equation}
    \hat{\mathrm{x}}=\mathrm{x}+\sum_{n=1}^{\infty}\frac{1}{n!}\frac{\partial^{n}\hat{\mathrm{x}}}{\partial\delta^{n}}\bigg{|}_{\delta=1}\left(\delta-1\right)^{n},\label{razvoj}
\end{equation}
\noindent where we have used the fact that $\hat{\mathrm{x}}(\mathrm{x},1)=\mathrm{x}$. The differential of the function $\hat{\mathrm{x}}=\hat{\mathrm{x}}(x,\delta)$ is given by:
\begin{equation}
    \de\hat{\mathrm{x}}=\frac{e^{\mathrm{x}}(\mathrm{x}-1)}{\hat{\mathrm{x}}e^{\hat{\mathrm{x}}}}\de\delta+\delta\frac{\mathrm{x}e^{\mathrm{x}}}{\hat{\mathrm{x}}e^{\hat{\mathrm{x}}}}\de\mathrm{x}.\label{diferencijal}
\end{equation}
\noindent The $n$-th derivative can be written in the following form:
\begin{equation}
    \frac{\partial^{n}\hat{\mathrm{x}}}{\partial\delta^{n}}=e^{n\mathrm{x}}(\mathrm{x}-1)^{n}\left(\frac{e^{-\hat{\mathrm{x}}}}{\hat{\mathrm{x}}}\frac{\de}{\de\hat{\mathrm{x}}}\right)^{n-1}\frac{e^{-\hat{\mathrm{x}}}}{\hat{\mathrm{x}}}.\label{izvod}
\end{equation}
\noindent Demanding that $\delta=1$ is equivalent to replacing $\hat{\mathrm{x}}=\mathrm{x}$ on the right hand side of equation (\ref{izvod}). We define the array of polynomials $\mathrm{P}_{n}\left(\frac{1}{\mathrm{x}}\right)$ by:
\begin{equation}
    \mathrm{P}_{n-1}\left(\frac{1}{\mathrm{x}}\right)=(-1)^{n-1}\mathrm{x}^{n}e^{n\mathrm{x}}\left(\frac{e^{-\mathrm{x}}}{\mathrm{x}}\frac{\de}{\de\mathrm{x}}\right)^{n-1}\frac{e^{-\mathrm{x}}}{\mathrm{x}}\label{P_poli}.
\end{equation}
\noindent Using mathematical induction, it is easy to show that the Rodrigues formula (equation (\ref{P_poli})) really defines polynomials of degree $n-1$. With the help of equation (\ref{P_poli}) it is straight forward to derive a recurrence relation for the $\mathrm{P}_{n}(1/\mathrm{x})$ polynomials, which is given by:
\begin{equation}
    \mathrm{P}_{n}\left(\frac{1}{\mathrm{x}}\right)=\left[n\left(1+\frac{1}{\mathrm{x}}\right)-\frac{\de}{\de\mathrm{x}}\right]\mathrm{P}_{n-1}\left(\frac{1}{\mathrm{x}}\right)\label{P_rekrentna_1/x}.
\end{equation}
From this relation, the recurrence relation (\ref{P_rekurentna}) is directly derived by the substitution $x=1/\mathrm{x}$. Replacing $n=1$ in equation (\ref{P_poli}), the initial condition reads $\mathrm{P}_{0}(x)=1$.
\end{proof}
\par To properly use equation (\ref{x>1resenje}), it is important to show that this series is uniformly convergent. If so, then the solution of equation (\ref{transc}) when $\mathrm{x}\geqslant1$ is given by equation (\ref{x>1resenje}).
\begin{theorem}\label{T2}
    The series $\sum_{n=1}^{\infty}\frac{(1-\delta)^{n}}{n!}f_{n}(\mathrm{x})$ is uniformly convergent when $\mathrm{x}\geqslant1$. The functional array $f_{n}(x)$ is defined by the following equation:
    \begin{equation}
        f_{n}(x)=\left(1-\frac{1}{x}\right)^{n}\mathrm{P}_{n-1}\left(\frac{1}{x}\right).
    \end{equation}
\end{theorem}
\begin{proof}
    To prove this theorem, we need to prove the following lemma first.
    \begin{lemma}\label{T3}
        The functions $f_{n}(x)$ are monotonically increasing functions of $x$ when $x\geqslant1$.
    \end{lemma}
    \begin{proof}
        We will calculate the derivative of the function $f_{n}(x)$, and show that it is positive when $x>1$. Using equation (\ref{P_rekrentna_1/x}) the derivative of the polynomial $\mathrm{P_{n-1}}(x)$ can be eliminated, resulting in:
        \begin{align}
            \frac{\de f_{n}(x)}{\de x}&=\left(1-\frac{1}{x}\right)^{n-1}\times\nonumber\\
            &\times\left[n\mathrm{P}_{n-1}\left(\frac{1}{x}\right)-\left(1-\frac{1}{x}\right)\mathrm{P}_{n}\left(\frac{1}{x}\right)\right].
        \end{align}
        Since $x>1$, $f'_{n}(x)>0$ is equivalent to proving that
        \begin{equation}
            n\mathrm{P}_{n-1}\left(\frac{1}{x}\right)>\left(1-\frac{1}{x}\right)\mathrm{P}_{n}\left(\frac{1}{x}\right),\hspace{1mm}\forall n\geqslant1.\label{j1}
        \end{equation}
        We write the polynomials in terms of their coefficients as $\mathrm{P}_{n}(x)=\sum_{k=0}^{n}a_{n}^{(k)}x^{k}$. Equation (\ref{j1}) now takes the following form:
        \begin{align}
            na_{n-1}^{(0)}&-a_{n}^{(0)}+\sum_{k=1}^{n-1}\left[na_{n-1}^{(k)}+a_{n}^{(k-1)}-a_{n}^{(k)}\right]\frac{1}{x^{k}}\nonumber\\&
            +\left[a_{n}^{(n-1)}-a_{n}^{(n)}\right]\frac{1}{x^{n}}+a_{n}^{(n)}\frac{1}{x^{n+1}}>0.\label{j2}
        \end{align}
        We will show that coefficient of every order in equation (\ref{j2}) is greater than 0. To do this we need the recurrence relation between the coefficients $a_{n}^{(k)}$, which is easily derived using recurrence relation (\ref{P_rekurentna}). The result is given by:
        \begin{align}
            k=0:\hspace{2mm}&a_{n}^{(0)}=na_{n-1}^{(0)},\label{rel1}\\
            k=1:\hspace{2mm}&a_{n}^{(1)}=na_{n-1}^{(1)}+na_{n-1}^{(0)},\label{rel2}\\
            1<k<n:\hspace{2mm}&a_{n}^{(k)}=na_{n-1}^{(k)}+(n+k-1)a_{n-1}^{(k-1)},\label{rel3}\\
            k=n:\hspace{2mm}&a_{n}^{(n)}=(2n-1)a_{n-1}^{(n-1)}.\label{rel4}
        \end{align}
        Together with the initial condition $a_{0}^{(0)}=1$, equations (\ref{rel1}) through (\ref{rel4}) give $a_{n}^{(0)}=n!$, $a_{n}^{(1)}=nn!$ and $a_{n}^{(n)}=(2n-1)!!$ respectively.
        \noindent Now, we return to proving the inequality (\ref{j2}). From equation (\ref{rel1}) it follows that the $k=0$ term in the (\ref{j2}) vanishes. The first-order term vanishes because $a_{n}^{(1)}=nn!$. Now, the inequality (\ref{j2}) nay be recast as:
        \begin{equation}
            (\forall k\in\{2,3,...,n-1\})\hspace{1mm}na_{n-1}^{(k)}+a_{n}^{(k-1)}-a_{n}^{(k)}>0.\label{j3}
        \end{equation}
        With the help of the relation (\ref{rel3}), the inequality (\ref{j3}) can be rewritten in the following form:
        \begin{equation}
            (n+k-2)a_{n-1}^{(k-2)}>(k-1)a_{n-1}^{(k-1)}.
        \end{equation}
        Shifting $k-1$ to $k$, and $n-1$ to $n$, we arrive at the following expression:
        \begin{equation}
            (\forall k\in\{1,2,...,n-1\})\hspace{1mm}(n+k)a_{n}^{(k-1)}>ka_{n}^{(k)}\label{j4}.
        \end{equation}
        Now, we apply mathematical induction to prove this inequality. First, we prove the relation for $n=3$ (the base case of the induction). The coefficients of the polynomial $\mathrm{P}_{3}(x)=15x^{3}+25x^{2}+18x+6$ clearly satisfy relations (\ref{j4}). We need to prove that the relation holds for $n+1$ if relation (\ref{j4}) is satisfied for $n$ (induction hypothesis):
        \begin{equation}
            (\forall k\in\{1,2,...,n\})\hspace{1mm}(n+k+1)a_{n+1}^{(k-1)}>ka_{n+1}^{(k)}\label{j5}.
        \end{equation}
        Using relation (\ref{rel3}), the inequality (\ref{j5}) can be rewritten as:
        \begin{align}
            &(\forall k\in\{2,...,n\})\hspace{1mm}(n+1)\left[(n+k)a_{n}^{(k-1)}-ka_{n}^{(k)}\right]\nonumber\\
            &+(n+k+1)\left[(n+k-1)a_{n}^{(k-2)}-(k-1)a_{n}^{(k-1)}\right]>0\label{j6}
        \end{align}
        Following from the hypothesis, both terms are greater then 0. Only two cases remain to be proven: $k=1$ and $k=n$. In the case of $k=1$, the relation is given by $(n+2)a_{n+1}^{(0)}>a_{n+1}^{(1)}$. Using the closed forms for these two coefficients, $(n+2)(n+1)!>(n+1)(n+1)!$, which is clearly satisfied. In the case of $k=n$, we are allowed to use formula (\ref{j6}), which takes the following form: $(n+1)n\left[2a_{n}^{(n-1)}-a_{n}^{(n)}\right]+(2n+1)\left[(2n-1)a_{n}^{(n-2)}-(n-1)a^{(n-1)}_{n}\right]>0$. The second term is greater then zero (from the hypothesis). Using relation (\ref{rel3}) $2a_{n}^{(n-1)}>2na_{n-1}^{(n-1)}=2n(2n-3)!!>(2n-1)!!=a_{n}^{(n)}$. With this, the mathematical induction is finished and we have proven inequality (\ref{j4}), and subsequently the inequality (\ref{j3}).
        \par The last remaining part of inequality (\ref{j2}) that needs to be proven is the relation $(\forall n\geqslant1)a_{n}^{(n-1)}>a_{n}^{(n)}$. Using the relation (\ref{rel3}) (and the result $a_{n}^{(n)}=(2n-1)!!$) repeatedly, it is easy to arrive at the following expression for the closed form of $a_{n}^{(n-1)}$:
        \begin{equation}
            a_{n}^{(n-1)}=\sum_{k=0}^{n-1}\frac{2^{k}(n-k)(n-1)!(2(n-k)-3)!!}{(n-k-1)!}.
        \end{equation}
        The zeroth term in this sum is equal to $n(2n-3)!!$, and every other term is greater then $(2n-3)!!$, which means that: $a_{n}^{(n-1)}>(2n-1)(2n-3)!!=a_{n}^{(n)}$. This is what we needed to prove. 
        \par With this, we have proven inequality (\ref{j2}) and subsequently inequality (\ref{j1}), thus proving the lemma.
    \end{proof}
To prove theorem \ref{T2}, the Weierstrass criterion will be used. Lemma \ref{T3} implies that
\begin{equation}
    (\forall x\geqslant1)f_{n}(x)<\lim_{x\to\infty}f_{n}(x)=a_{n-1}^{(0)}=(n-1)!,\label{j7}
\end{equation}
\noindent $\forall n\geqslant1$. The inequality (\ref{j7}) implies the following inequality for the coefficients of the series $\sum_{n=1}^{\infty}\frac{(1-\delta)^{n}}{n!}f_{n}(x)$:
\begin{equation}
    (\forall n\geqslant1)(\forall x\geqslant1)\hspace{1mm}\bigg{|}\frac{(1-\delta)^{n}}{n!}f_{n}(x)\bigg{|}<\frac{(1-\delta)^{n}}{n}.
\end{equation}
The next step is to prove that the series $\sum_{n=1}^{\infty}\frac{(1-\delta)^{n}}{n}$ converges. By the comparison test, this is true for $|1-\delta|<1$, which is satisfied. Then, according to the Weierstrass criterion, the series $\sum_{n=1}^{\infty}\frac{(1-\delta)^{n}}{n!}f_{n}(x)$ converges uniformly.
\end{proof}
This result is important because the infinite sum may exchange places with integrals and limits with respect to $x$, which will be used in the main text. Now we move on to the case when $x<1$.
\begin{theorem}\label{T4}
    The solution to equation (\ref{transc}) $\hat{\mathrm{x}}=\hat{\mathrm{x}}(\mathrm{x},\delta)$, when $\delta\leqslant1$, $\hat{\mathrm{x}}\geqslant0$ and $0\leqslant\mathrm{x}\leqslant1$ is given by the following functional series:
    \begin{equation}
        \hat{\mathrm{x}}(\mathrm{x},\delta)=1-\sum_{n=1}^{\infty}\frac{\delta^{n}}{n!}\mathrm{P}_{n-1}(1)e^{n(\mathrm{x}-1)}(1-\mathrm{x})^{n},\label{x<1resenje}
    \end{equation}
    where $\mathrm{P}_{n}(x)$ are the polynomials defined in theorem \ref{T1}.
\end{theorem}
\begin{proof}
    If $0\leqslant\mathrm{x}\leqslant1$, then $0\leqslant e^{\mathrm{x}}(1-\mathrm{x})\leqslant1$, which means that we can choose $\Tilde{\delta}=\delta e^{\mathrm{x}}(1-\mathrm{x})$ as a small parameter, and expand the solution as a Taylor series with respect to $\Tilde{\delta}$ around $\Tilde{\delta}=0$, 
    \begin{equation}
        \hat{\mathrm{x}}(\Tilde{\delta})=\hat{\mathrm{x}}(0)+\sum_{n=1}^{\infty}\frac{\Tilde{\delta}^{n}}{n!}\frac{\de^{n}\hat{\mathrm{x}}}{\de\Tilde{\delta}^{n}}\bigg{|}_{\Tilde{\delta}=0}.\label{i0}
    \end{equation} The equation we are solving, in terms of $\hat{\mathrm{x}}=\hat{\mathrm{x}}(\Tilde{\delta})$, is given by:
    \begin{equation}
        e^{\hat{\mathrm{x}}(\Tilde{\delta})}(1-\hat{\mathrm{x}}(\Tilde{\delta}))=\Tilde{\delta}.\label{i1}
    \end{equation}
    With the help of the equation (\ref{i1}) the $n$-th derivative of the function $\hat{\mathrm{x}}(\Tilde{\delta})$ is given by
    \begin{equation}
        \frac{\de^{n}\hat{\mathrm{x}}}{\de\Tilde{\delta}^{n}}=(-1)^{n}\left(\frac{e^{-\hat{\mathrm{x}}}}{\hat{\mathrm{x}}}\frac{\de}{\de\hat{\mathrm{x}}}\right)^{n-1}\frac{e^{-\hat{\mathrm{x}}}}{\hat{\mathrm{x}}}=-\frac{e^{-n\hat{\mathrm{x}}}}{\hat{\mathrm{x}}^{n}}\mathrm{P}_{n-1}\left(\frac{1}{\hat{\mathrm{x}}}\right),
    \end{equation}
    where equation (\ref{P_poli}) has been used. Using the fact that $\hat{\mathrm{x}}=1$ when $\Tilde{\delta}=0$ we arrive at the following expression:
    \begin{equation}
        \frac{\de^{n}\hat{\mathrm{x}}}{\de\Tilde{\delta}^{n}}=-e^{-n}\mathrm{P}_{n-1}(1).
    \end{equation}
    Changing this result in formula (\ref{i0}) we find $\Tilde{\delta}=\delta e^{\mathrm{x}}(1-\mathrm{x})$ and we have derived formula (\ref{x<1resenje}).
\end{proof}
To be able to use this result (\ref{x<1resenje}), once again we need to prove that the series appearing in equation (\ref{x<1resenje}) is uniformly convergent.
\begin{theorem}\label{T5}
    The series $\sum_{n=1}^{\infty}\frac{\delta^{n}}{n!}\mathrm{P}_{n-1}(1)e^{n(x-1)}(1-x)^{n}$ is uniformly convergent when $0\leqslant x\leqslant1$. 
\end{theorem}
\begin{proof}
    Since $0\leqslant x\leqslant1$, we have $e^{nx}(1-x)^{n}<1$, which means
    \begin{equation}
        \bigg{|}\frac{\delta^{n}}{n!}\mathrm{P}_{n-1}(1)e^{n(x-1)}(1-x)^{n}\bigg{|}<\frac{\delta^{n}}{n!}\mathrm{P}_{n-1}(1)e^{-n}\label{i2},
    \end{equation}
    $\forall 0\leqslant x\leqslant1$ and $\forall n\geqslant1$. The next step is to prove that the series $\sum_{n=1}^{\infty}\frac{\delta^{n}}{n!}\mathrm{P}_{n-1}(1)e^{-n}$ is convergent. To be able to do this, we need a closed form for $\mathrm{P}_{n-1}(1)$, which will be derived in the following lemma, without proof.
    \begin{lemma}\label{T6}
        The polynomials defined in theorem \ref{T1} have the following property: $\mathrm{P}_{n}(1)=(n+1)^{n},\hspace{1mm}\forall n\geqslant0$.
    \end{lemma}
    Now we can use the comparison test to check if the series $\sum_{n=1}^{\infty}\frac{\delta^{n}}{n!}\mathrm{P}_{n-1}(1)e^{-n}$ is convergent. Using lemma \ref{T6} we get
    \begin{equation}
        \frac{\frac{\delta^{n+1}}{(n+1)!}(n+1)^{n}e^{-(n+1)}}{\frac{\delta^{n}}{n!}n^{n-1}e^{-n}}=\frac{\delta}{e}\left(1+\frac{1}{n}\right)^{n-1}\underset{n\to\infty}{\longrightarrow}\delta,
    \end{equation}
    which means that the series is convergent when $\delta<1$. This requirement is satisfied by the theorem statement. Then, according to the Weierstrass criterion, the series $\sum_{n=1}^{\infty}\frac{\delta^{n}}{n!}\mathrm{P}_{n-1}(1)e^{n(x-1)}(1-x)^{n}$ converges uniformly.
\end{proof}
Sometimes functions of $\hat{\mathrm{x}}$ will appear in our calculations (for example, $\ln{\hat{\mathrm{x}}}$). Similar expressions to those of theorems \ref{T1} and \ref{T4} are needed for these functions.
\begin{theorem}\label{T7}
    The differentiable function $\mathrm{F}(\hat{\mathrm{x}}(\mathrm{x},\delta))$, where $\hat{\mathrm{x}}$ is a solution of the equation (\ref{transc}), when $\delta\leqslant1$, $\hat{\mathrm{x}}\geqslant0$ and $\mathrm{x}\geqslant1$, is given by the following functional series:
    \begin{equation}
    \mathrm{F}(\hat{\mathrm{x}})=\mathrm{F}(\mathrm{x})-\sum_{n=1}^{\infty}\frac{(1-\delta)^{n}}{n!}\left(1-\frac{1}{\mathrm{x}}\right)^{n}\mathrm{F}_{n-1}\left(\frac{1}{\mathrm{x}}\right),
\end{equation}
    where $\mathrm{F}_{n}(x)$ is an array of functions, defined by the following recurrence relation,
    \begin{equation}
    \mathrm{F}_{n}(x)=\left[n(x+1)+x^{2}\frac{\rm{d}}{\rm{d} x}\right]\mathrm{F}_{n-1}(x),
\end{equation}
    and the initial condition $\mathrm{F}_{0}(x)=-x^{2}\frac{\de \mathrm{F}(1/x)}{\de x}$.
\end{theorem}
\begin{proof}
    The proof is analogous to that of theorem \ref{T1}. We start by writing the Taylor series of the function $\mathrm{F}(\mathrm{x},\delta)$ around the point $\delta=1$:
    \begin{equation}
        \mathrm{F}(\hat{\mathrm{x}})=\mathrm{F}(\mathrm{x})+\sum_{n=1}^{\infty}\frac{1}{n!}\frac{\partial^{n}\mathrm{F}(\hat{\mathrm{x}})}{\partial\delta^{n}}\bigg{|}_{\delta=1}(\delta-1)^{n}.
    \end{equation}
    The $n$-th derivative is now given by the following expression:
    \begin{equation}
        \frac{\partial^{n}\mathrm{F}(\hat{\mathrm{x}})}{\partial\delta^{n}}=e^{n\mathrm{x}}(\mathrm{x}-1)^{n}\left(\frac{e^{-\hat{\mathrm{x}}}}{\hat{\mathrm{x}}}\frac{\de}{\de\hat{\mathrm{x}}}\right)^{n-1}\frac{e^{-\hat{\mathrm{x}}}}{\hat{\mathrm{x}}}\frac{\de\mathrm{F}(\hat{\mathrm{x}})}{\de\hat{\mathrm{x}}}.
    \end{equation}
    The same way as in the theorem \ref{T1} we define an array of functions:
    \begin{equation}
    \mathrm{F}_{n-1}\left(\frac{1}{\mathrm{x}}\right)=(-1)^{n-1}\mathrm{x}^{n}e^{n\mathrm{x}}\left(\frac{e^{-\mathrm{x}}}{\mathrm{x}}\frac{\de}{\de\mathrm{x}}\right)^{n-1}\frac{e^{-\mathrm{x}}}{\mathrm{x}}\frac{\de\mathrm{F}(\mathrm{x})}{\de\mathrm{x}}.
\end{equation}
    It is now obvious that the recurrence relation for the functions $\mathrm{F}_{n}$ would be the same as for the polynomials defined in the theorem \ref{T1}. The only difference would be the initial condition, which is now given by $\mathrm{F}_{0}(1/\mathrm{x})=\frac{\de\mathrm{F}(\mathrm{x})}{\de\mathrm{x}}$. Applying the substitution $\mathrm{x}\mapsto1/\mathrm{x}$, we arrive at the expression given in the statement of the theorem.
\end{proof}
Once again, the question of uniform convergence must be asked. This time, it depends on the form of the function $\mathrm{F}(\hat{\mathrm{x}})$. The first requirement is that $\mathrm{F}$ does not diverge when $\hat{\mathrm{x}}\to\infty$, since the radius of convergence would depend on $\mathrm{x}$, and the discussion of uniform convergence would become meaningless. From the proof of theorem \ref{T2} we can see that the necessary condition is:
\begin{equation}
    (\forall\mathrm{x}\geqslant1)(\exists n_{0})(\forall n\geqslant n_{0})\bigg{|}\left(1-\frac{1}{\mathrm{x}}\right)^{n}\mathrm{F}_{n-1}\left(\frac{1}{\mathrm{x}}\right)\bigg{|}<n!.
\end{equation}
\noindent The function $\ln{\hat{\mathrm{x}}}$ is the only one directly involved in the calculation. Since this is the case, we will take extra care of this function.
\begin{corollary}\label{T7.1}
    The function $\ln{\hat{\mathrm{x}}}(\mathrm{x},\delta)$ when $\delta\leqslant1$, $\hat{\mathrm{x}}\geqslant0$ and $\mathrm{x}\geqslant1$, is given by the following functional series:
    \begin{equation}
    \ln{\hat{\mathrm{x}}}=\ln{\mathrm{x}}-\sum_{n=1}^{\infty}\frac{(1-\delta)^{n}}{n!}\left(1-\frac{1}{\mathrm{x}}\right)^{n}\mathrm{Q}_{n-1}\left(\frac{1}{\mathrm{x}}\right),\label{ln_x>1_resenje}
\end{equation}
    where $\mathrm{Q}_{n}(x)$ are polynomials of degree $n+1$, defined by the following recurrence relation:
    \begin{equation}
    \mathrm{Q}_{n}(x)=\left[n(x+1)+x^{2}\frac{\rm{d}}{\mathrm{d} x}\right]\mathrm{Q}_{n-1}(x),\label{Q_rekurentna}
\end{equation}
    and the initial condition $\mathrm{Q}_{0}(x)=x$.
\end{corollary}
\begin{proof}
    This corollary is a direct consequence of theorem \ref{T7}. The initial condition is obvious since $-x^{2}(\ln{(1/x)})'=x$. We only need to prove that $\mathrm{Q}_{n}$ are polynomials of degree $n+1$, which follows directly from the recurrence relation (\ref{Q_rekurentna}) since it increases the degree of the polynomial by one, and the degree of the zeroth polynomial is 1. This concludes the proof.
\end{proof}
\begin{theorem}\label{T8}
    The series $\sum_{n=1}^{\infty}\frac{(1-\delta)^{n}}{n!}f_{n}(x)$ is uniformly convergent when $\mathrm{x}\geqslant1$. The functional array $f_{n}(x)$ is defined by the following equation:
    \begin{equation}
        f_{n}(x)=x\left(1-\frac{1}{x}\right)^{n+1}\mathrm{Q}_{n-1}\left(\frac{1}{x}\right).
    \end{equation}
\end{theorem}
\begin{proof}
    This proof will not be as detailed as the previous proofs since it is based on the same ideas. By mathematical induction, it is easy to prove that the functions $f_{n}(x)$ are monotonically increasing, the same way as in lemma \ref{T3}. This fact implies that the functions $f_{n}(x)$ take their maximum value when $x\to\infty$. Then, the recurrence relations (\ref{Q_rekurentna}) tell us that the maximal value is $(n-1)!$. The rest of the proof is exactly the same as in theorem \ref{T2}.
\end{proof}
\noindent Now we examine the $x<1$ case.
\begin{theorem}\label{T9}
    The function $\ln{\hat{\mathrm{x}}}(\mathrm{x},\delta)$, when $\delta\leqslant1$, $\hat{\mathrm{x}}\geqslant0$ and $0\leqslant\mathrm{x}\leqslant1$ is given by the following functional series:
    \begin{equation}
        \ln{\hat{\mathrm{x}}}=-\sum_{n=1}^{\infty}\frac{\delta^{n}}{n!}\mathrm{Q}_{n-1}(1)e^{n(\mathrm{x}-1)}(1-\mathrm{x})^{n},\label{ln_x<1resenje}
    \end{equation}
    where $\mathrm{Q}_{n}(x)$ are the polynomials defined in corollary \ref{T7.1}.
\end{theorem}
\begin{proof}
    The proof of this theorem is essentially the same as the proof of theorem \ref{T4}. The only difference is the appearance of the polynomials $\mathrm{Q}_{n}$ instead of the polynomials $\mathrm{P}_{n}$.
\end{proof}
\noindent Finally, we need to prove the uniform convergence of the series appearing in equation \ref{ln_x<1resenje}.
\begin{theorem}\label{T10}
    When $0\leqslant x\leqslant1$, the functional series $\sum_{n=1}^{\infty}\frac{\delta^{n}}{n!}\mathrm{Q}_{n-1}(1)e^{n(x-1)}(1-x)^{n}$ is uniformly convergent. 
\end{theorem}
\begin{proof}
    As in theorem \ref{T5}, we can write:
    \begin{equation}
        \bigg{|}\frac{\delta^{n}}{n!}\mathrm{Q}_{n-1}(1)e^{n(x-1)}(1-x)^{n}\bigg{|}<\frac{\delta^{n}}{n!}\mathrm{Q}_{n-1}(1)e^{-n}.\label{k0}
    \end{equation}
    Now we need the closed form of $\mathrm{Q}_{n-1}(1)$, which will be given by the following lemma without proof.
    \begin{lemma}\label{T11}
        The polynomials defined in corollary \ref{T7.1} have the following property: $(\forall n\geqslant1)\hspace{1mm}\mathrm{Q}_{n-1}(1)=(n-1)!\sum_{k=0}^{n-1}\frac{n^{k}}{k!}$.
    \end{lemma}
    \noindent Using lemma \ref{T11}, the following inequality holds:
    \begin{equation}
        \mathrm{Q}_{n-1}(1)=(n-1)!\sum_{k=1}^{n-1}\frac{n^{k}}{k!}<(n-1)!e^{n}.\label{k1}
    \end{equation}
    Using (\ref{k1}), the inequality (\ref{k0}) becomes:
    \begin{equation}
        \bigg{|}\frac{\delta^{n}}{n!}\mathrm{Q}_{n-1}(1)e^{n(x-1)}(1-x)^{n}\bigg{|}<\frac{\delta^{n}}{n}.
    \end{equation}
    The next step would be to prove the convergence of the series $\sum_{n=1}^{\infty}\frac{\delta^{n}}{n}$. This series is convergent by the comparison test when $\delta<1$, which is satisfied.
\end{proof}
\section{More suitable form of the solution}\label{app_C}
Here we will rewrite equations (\ref{x_resenje_>}) and (\ref{x_resenje_<}) in a more suitable form so that they are well defined when $\mathrm{x}=1+\mathcal{O}(\tilde{\varepsilon})$ or $\hat{\mathrm{x}}=1+\mathcal{O}(\tilde{\varepsilon})$, where we have defined $\tilde{\varepsilon}=\frac{\varepsilon}{(\lambda a)^{2}}$. This is important since many hypersurfaces of interest are given by expressions of this kind (apparent horizon, horizon, island, QES, etc.). Looking at the equations of motion (\ref{ppmm_q}-\ref{chi2_q}) and the way they are solved (\ref{jna_pp}-\ref{theta_sol}) one can expect that the expressions for $\partial_{\pm}\mathrm{x}$ (\ref{theta_izvod_plus}-\ref{theta_izvod_minus}) are well defined when $\mathrm{x}=1+\mathcal{O}(\tilde{\varepsilon})$ since they are derived by simple integration of the original equations of motion. Indeed, this is the case:
\begin{widetext}
    \begin{align}
        \partial_{+}\mathrm{x}&=\frac{1}{2a}\left\{1-\frac{1+\frac{\tilde{\varepsilon}}{4}\ln{\delta}}{\sqrt{\mathrm{x}^{2}-\tilde{\varepsilon}}}\left[1-\frac{\tilde{\varepsilon}}{8}\left(2(\mathrm{x}-2)\ln{\mathrm{x}}-8+5\mathrm{x}+\frac{3}{\mathrm{x}}-(\mathrm{x}-1)\sum_{n=1}^{\infty}\frac{\delta^{n}}{n!}\left(\frac{\mathrm{x}-1}{\mathrm{x}}\frac{\mathrm{d}\mathrm{S}_{n}^{<}}{\mathrm{dx}}-n\mathrm{S}_{n}^{<}\right)\right)\right]\right\},\label{d_+x}\\
        \partial_{-}\mathrm{x}&=-\frac{1}{2a\mathrm{F}^{-}}\left\{1-\frac{1+\frac{\tilde{\varepsilon}}{4}\ln{\delta}}{\sqrt{\mathrm{x}^{2}-\tilde{\varepsilon}}}\left[1-\frac{\tilde{\varepsilon}}{8}\left(2(\mathrm{x}-2)\ln{\mathrm{x}}-2(\mathrm{x}-1)\ln{(1-\mathrm{x})}+1-2\mathrm{x}+\frac{3}{\mathrm{x}}-\frac{(\mathrm{x}-1)^{2}}{\mathrm{x}}\sum_{n=1}^{\infty}\frac{\delta^{n}}{n!}\frac{\mathrm{d}\mathrm{S}_{n}^{<}}{\mathrm{dx}}\right)\right]\right\}.\label{d_-x}
    \end{align}
\end{widetext}
This suggests that the issue arises during the execution of the final integration (\ref{theta_sol_1}-\ref{theta_sol}). Note that when taking the limit $\delta\to1$, both equations (\ref{d_+x}) and (\ref{d_-x}) reduce to the exact solutions (\ref{x_izvod_delta=1}). Now we define a new coordinate $\mathrm{y}=\frac{\mathrm{x}}{\tilde{a}}$, where $\tilde{a}=1+\frac{\tilde{\varepsilon}}{4}\ln{\delta}$, so that $a\tilde{a}$ can be seen as a quantum-corrected constant $a$. Applying this, equations (\ref{d_+x}-\ref{d_-x}) become:
\begin{widetext}
    \begin{align}
        \partial_{+}\mathrm{x}&=\frac{1}{2a}\left\{1-\frac{1}{\sqrt{\mathrm{y}^{2}-\tilde{\varepsilon}}}\left[1-\frac{\tilde{\varepsilon}}{8}\left(2(\mathrm{y}-2)\ln{\mathrm{y}}-8+5\mathrm{y}+\frac{3}{\mathrm{y}}-(\mathrm{y}-1)\sum_{n=1}^{\infty}\frac{\delta^{n}}{n!}\left(\frac{\mathrm{y}-1}{\mathrm{y}}\frac{\mathrm{d}\mathrm{S}_{n}^{<}}{\mathrm{dy}}-n\mathrm{S}_{n}^{<}\right)\right)\right]\right\},\label{d_+xy}\\
        \partial_{-}\mathrm{x}&=-\frac{1}{2a\mathrm{F}^{-}}\left\{1-\frac{1}{\sqrt{\mathrm{y}^{2}-\tilde{\varepsilon}}}\left[1-\frac{\tilde{\varepsilon}}{8}\left(2(\mathrm{y}-2)\ln{\mathrm{y}}-2(\mathrm{y}-1)\ln{(1-\mathrm{y})}+1-2\mathrm{y}+\frac{3}{\mathrm{y}}-\frac{(\mathrm{y}-1)^{2}}{\mathrm{y}}\sum_{n=1}^{\infty}\frac{\delta^{n}}{n!}\frac{\mathrm{d}\mathrm{S}_{n}^{<}}{\mathrm{dy}}\right)\right]\right\}.\label{d_-xy}
    \end{align}
\end{widetext}
For a moment let us exclude the $\mathcal{O}(\tilde{\varepsilon})$ part from the square brackets within (\ref{d_+xy}) and define: 
\begin{equation}
    \partial_{+}\mathrm{x}^{(0)}=\frac{1}{2a}\left(1-\frac{1}{\sqrt{\mathrm{y}^{2}-\tilde{\varepsilon}}}\right).\label{d_pmx^0}
\end{equation}
And then, let us define a function $\mathcal{J}(\mathrm{x})$ so that it satisfies:
\begin{equation}
    \tilde{a}\mathcal{J}(\mathrm{y})=\frac{\sigma^{+}}{2a}.\label{J_deff}
\end{equation}
After taking the derivative of (\ref{J_deff}) with respect to $\sigma^{+}$ and substituting the expression (\ref{d_pmx^0}) one arrives at the following differential equation for $\mathcal{J}$:
\begin{widetext}
    \begin{equation}
        \frac{\mathrm{d}\mathcal{J}}{\mathrm{d\mathrm{y}}}-\frac{\tilde{\varepsilon}}{4}\frac{\sqrt{\mathrm{y}^{2}-\tilde{\varepsilon}}}{(1+\frac{\tilde{\varepsilon}}{4}\mathrm{y})\sqrt{\mathrm{y}^{2}-\tilde{\varepsilon}}-1}\mathcal{J}=\frac{\sqrt{\mathrm{y}^{2}-\tilde{\varepsilon}}}{(1+\frac{\tilde{\varepsilon}}{4}\mathrm{y})\sqrt{\mathrm{y}^{2}-\tilde{\varepsilon}}-1}.\label{J_jednacina}
    \end{equation}
\end{widetext}
The solution to equation (\ref{J_jednacina}) is then given by:
\begin{equation}
    \mathcal{J}(\mathrm{y})=-\frac{4}{\tilde{\varepsilon}}+S\exp{\left(\frac{\tilde{\varepsilon}}{4}\tilde{f}(\mathrm{y})\right)},\label{J_resenje}
\end{equation}
where $S$ is an unknown constant that needs to be determined, and $\tilde{f}(\mathrm{y})$ is given by:
\begin{equation}
    \tilde{f}(\mathrm{y})=\int\frac{\sqrt{\mathrm{y}^{2}-\tilde{\varepsilon}}}{(1+\frac{\tilde{\varepsilon}}{4}\mathrm{y})\sqrt{\mathrm{y}^{2}-\tilde{\varepsilon}}-1}\mathrm{dy}.\label{f_tilde_deff}
\end{equation}
After making a substitution $\mathrm{z}=\sqrt{\frac{\mathrm{y}-\sqrt{\tilde{\varepsilon}}}{\mathrm{y}+\sqrt{\tilde{\varepsilon}}}}$, the integral (\ref{f_tilde_deff}) becomes:
\begin{equation}
    \tilde{f}=8\tilde{\varepsilon}\int\frac{\mathrm{z}^{2}\mathrm{dz}}{(\mathrm{z}^{2}-1)R(z)},\label{integral_z}
\end{equation}
where $R(\mathrm{z})=\mathrm{z}^{4}+2\sqrt{\tilde{\varepsilon}}\left(1-\frac{\tilde{\varepsilon}^{3/2}}{4}\right)\mathrm{z}^{3}-2\mathrm{z}^{2}-2\sqrt{\tilde{\varepsilon}}\left(1+\frac{\tilde{\varepsilon}^{3/2}}{4}\right)\mathrm{z}+1$. This polynomial can be factorized as $R(\mathrm{z})=(\mathrm{z}^{2}+\alpha_{1}\mathrm{z}+\beta_{1})(\mathrm{z}^{2}+\alpha_{2}\mathrm{z}+\beta_{2})$. The constants $\alpha_{1,2}$ and $\beta_{1,2}$ satisfy the following equations:
\begin{align}
    \alpha_{1}+\alpha_{2}&=2\sqrt{\tilde{\varepsilon}}\left(1-\frac{\tilde{\varepsilon}^{\frac{3}{2}}}{4}\right),\nonumber\\
    \beta_{1}+\beta_{2}&=-2-\alpha_{1}\alpha_{2},\nonumber\\
    \alpha_{1}\beta_{2}+\alpha_{2}\beta_{1}&=-2\sqrt{\tilde{\varepsilon}}\left(1+\frac{\tilde{\varepsilon}^{\frac{3}{2}}}{4}\right),\nonumber\\
    \beta_{1}\beta_{2}&=1.\label{alpha_beta_jne}
\end{align}
The system of equations (\ref{alpha_beta_jne}) can be exactly solved, but we will solve it perturbatively up to $\mathcal{O}(\tilde{\varepsilon}^{4})$. The solution is given by:
\begin{align}
    \alpha_{1}&=2\sqrt{\tilde{\varepsilon}}\left(1-\frac{\tilde{\varepsilon}^{\frac{3}{2}}}{4}\right)\left[1-\frac{\tilde{\varepsilon}^{2}}{16}\left(1+\frac{\tilde{\varepsilon}^{\frac{3}{2}}}{2}\right)\right],\nonumber\\
    \alpha_{2}&=\frac{\tilde{\varepsilon}^{\frac{5}{2}}}{8}\left(1+\frac{\tilde{\varepsilon}^{\frac{3}{2}}}{4}\right),\nonumber\\
    \beta_{1}&=-1+\frac{\tilde{\varepsilon}^{\frac{3}{2}}}{2}\left(1-\frac{\tilde{\varepsilon}^{\frac{3}{2}}}{4}+\frac{\tilde{\varepsilon}^{2}}{8}\right),\nonumber\\
    \beta_{2}&=-1-\frac{\tilde{\varepsilon}^{\frac{3}{2}}}{2}\left(1+\frac{\tilde{\varepsilon}^{\frac{3}{2}}}{4}+\frac{\tilde{\varepsilon}^{2}}{8}\right).\label{alpha_beta_resenje}
\end{align}
The expressions for the zeros of the polynomial $R(\mathrm{z})$ are:
\begin{align}
    \mathrm{z}_{2}^{+}&=1+\frac{\tilde{\varepsilon}^{\frac{3}{2}}}{4}-\frac{\tilde{\varepsilon}^{\frac{5}{2}}}{16}+\frac{\tilde{\varepsilon}^{3}}{32}+\frac{\tilde{\varepsilon}^{\frac{7}{2}}}{32}-\frac{\tilde{\varepsilon}^{4}}{64}\label{z2+}\\
    \mathrm{z}_{2}^{-}&=-1-\frac{\tilde{\varepsilon}^{\frac{3}{2}}}{4}-\frac{\tilde{\varepsilon}^{\frac{5}{2}}}{16}-\frac{\tilde{\varepsilon}^{3}}{32}-\frac{\tilde{\varepsilon}^{\frac{7}{2}}}{32}-\frac{\tilde{\varepsilon}^{4}}{64}\label{z2-}\\
    \mathrm{z}_{1}^{+}&=1-\tilde{\varepsilon}^{\frac{1}{2}}+\frac{\tilde{\varepsilon}}{2}-\frac{\tilde{\varepsilon}^{\frac{3}{2}}}{4}+\frac{\tilde{\varepsilon}^{2}}{8}-\frac{\tilde{\varepsilon}^{\frac{5}{2}}}{16}+\frac{\tilde{\varepsilon}^{3}}{32}-\frac{\tilde{\varepsilon}^{4}}{128},\label{z1+}\\
    \mathrm{z}_{1}^{-}&=-1-\tilde{\varepsilon}^{\frac{1}{2}}-\frac{\tilde{\varepsilon}}{2}+\frac{\tilde{\varepsilon}^{\frac{3}{2}}}{4}+\frac{3\tilde{\varepsilon}^{2}}{8}+\frac{3\tilde{\varepsilon}^{\frac{5}{2}}}{16}-\frac{\tilde{\varepsilon}^{3}}{32}+\frac{5\tilde{\varepsilon}^{4}}{128}.\label{z1-}
\end{align}
Now, the integral (\ref{integral_z}) can be calculated directly. The result is given by the following expression:
\begin{widetext}
    \begin{align}
        \tilde{f}(\mathrm{z})&=\frac{4}{\tilde{\varepsilon}}\ln{\bigg{|}\frac{\mathrm{z}^{2}+\alpha_{2}\mathrm{z}+\beta_{2}}{1-\mathrm{z}^{2}}\bigg{|}}\label{f_tilde_resenje}\\
        &\hspace{5mm}+\left(1-\frac{5}{4}\tilde{\varepsilon}\right)\ln{|\mathrm{z}-\mathrm{z}_{1}^{+}|}-\left(1-\frac{1}{4}\tilde{\varepsilon}\right)\ln{|\mathrm{z}-\mathrm{z}_{1}^{-}|}-\left(1-\frac{3}{4}\tilde{\varepsilon}\right)\ln{|\mathrm{z}-\mathrm{z}_{2}^{+}|}+\left(1+\frac{3}{4}\tilde{\varepsilon}\right)\ln{|\mathrm{z}-\mathrm{z}_{2}^{-}|}+\frac{\tilde{\varepsilon}}{4}\ln{\frac{\tilde{\varepsilon}}{16}}.\nonumber
    \end{align}
\end{widetext}
The appropriate choice for the constant $S$ appearing in (\ref{J_resenje}) is $S=\frac{4}{\tilde{\varepsilon}}+\frac{9}{2}$, which guarantees that the divergent part of equations (\ref{x_resenje_>}) and (\ref{x_resenje_<}) can be adsorbed within the function $\mathcal{J}(\mathrm{y})$. Expanding the function $\mathcal{J}$, we arrive at the following expression:
\begin{widetext}
    \begin{equation}
        \mathcal{J}(\mathrm{y})=\mathrm{y}+\ln{|\mathrm{y}-1|}-\frac{\tilde{\varepsilon}}{8}\left[-4\ln{\mathrm{y}}-(2\mathrm{y}-1)\ln{|\mathrm{y}-1|}-\ln^{2}{|\mathrm{y}-1|}+\frac{2}{\mathrm{y}-1}-5\mathrm{y}\right]+\frac{9}{2}.\label{J_razvoj}
    \end{equation}
\end{widetext}
From (\ref{J_razvoj}) we can conclude that the function $\mathcal{J}$ has absorbed all the divergent terms, as well as $\ln{\delta}$ (after substituting $\mathrm{y}=\frac{\mathrm{x}}{\tilde{a}}$) within equations (\ref{x_resenje_>}) and (\ref{x_resenje_<}). Those equations can now be rewritten as:
\begin{widetext}
    \begin{align}
        &\left(1+\frac{\varepsilon}{4(\lambda a)^{2}}\ln{\delta}\right)\mathcal{J}(\mathrm{y})-\frac{\varepsilon}{8(\lambda a)^{2}}\left[(2\mathrm{y}+3)\ln{\mathrm{y}}+5\mathrm{y}-2\mathrm{L}(\mathrm{y})-\sum_{n=1}^{\infty}\frac{(1-\delta)^{n}}{n!}\mathrm{S}_{n}^{>}(\mathrm{y})\right]\nonumber\\
        &\hspace{6cm}=-\ln{\delta}+\mathcal{J}(\hat{\mathrm{x}})-\frac{\varepsilon}{8(\lambda a)^{2}}\left[(2\hat{\mathrm{x}}+3)\ln{\hat{\mathrm{x}}}+5\hat{\mathrm{x}}-2\mathrm{L}(\hat{\mathrm{x}})\right],\hspace{4mm}\text{when }\mathrm{y}>1\label{y_>_resenje}\\
        &\left(1+\frac{\varepsilon}{4(\lambda a)^{2}}\ln{\delta}\right)\mathcal{J}(\mathrm{y})-\frac{\varepsilon}{8(\lambda a)^{2}}\left[\frac{2\mathrm{y}^{2}}{\mathrm{y}-1}\ln{\mathrm{y}}+3\mathrm{y}-\sum_{n=1}^{\infty}\frac{\delta^{n}}{n!}\mathrm{S^{<}_{n}}(\mathrm{y})\right]\nonumber\\
        &\hspace{6cm}=-\ln{\delta}+\mathcal{J}(\hat{\mathrm{x}})-\frac{\varepsilon}{8(\lambda a)^{2}}\left[(2\hat{\mathrm{x}}+3)\ln{\hat{\mathrm{x}}}+5\hat{\mathrm{x}}-2\mathrm{L}(\hat{\mathrm{x}})\right],\hspace{4mm}\text{when }\mathrm{y}<1.\label{y_<_resenje}
    \end{align}
\end{widetext}
Note that both equations are well defined around $\mathrm{y}=1+\mathcal{O}(\tilde{\varepsilon})$, which has been the main goal of this section. Let us find the general solution when $\mathrm{y}=1+\frac{\varepsilon}{(\lambda a)^{2}}\eta$ and $\tilde{a}\gg0$. This implies $\hat{\mathrm{x}}=1+\frac{\varepsilon}{(\lambda a)^{2}}\mu$. Both equations (\ref{y_>_resenje}) and (\ref{y_<_resenje}) become:
\begin{equation}
    \tilde{a}\mathcal{J}\left(1+\tilde{\varepsilon}\eta\right)=-\ln{\delta}+\mathcal{J}\left(1+\tilde{\varepsilon}\mu\right).\label{sns1}
\end{equation}
This formula (\ref{sns1}) is further reduced to:
\begin{equation}
    \tilde{a}\exp{\left(\frac{\tilde{\varepsilon}}{4}\tilde{f}(1+\tilde{\varepsilon}\eta)\right)}=\exp{\left(\frac{\tilde{\varepsilon}}{4}\tilde{f}(1+\tilde{\varepsilon}\mu)\right)}.\label{sns2}
\end{equation}
Careful examination of the $\tilde{f}(\mathrm{z})$ function leads to the following:
\begin{equation}
    \tilde{a}^{\frac{4}{\tilde{\varepsilon}}}=\bigg{|}\frac{4\mu-1}{4\eta-1}\bigg{|}.\label{resenje_oko_1_opste}
\end{equation}
Since the solution is given up to the first-order of perturbation theory, we have $\tilde{a}^{4/\tilde{\varepsilon}}=\delta$. We expect that this relation holds in higher orders of perturbation theory as well. Then, equation (\ref{resenje_oko_1_opste}) simplifies to:
\begin{equation}
    \delta=\bigg{|}\frac{4\mu-1}{4\eta-1}\bigg{|}.\label{resenje_oko_1}
\end{equation}
\end{document}